%% file: main.tex
\documentclass[11pt]{article}

\input{header}

\begin{document}






\title{Dynamic Algorithms for Packing-Covering LPs\\via Multiplicative Weight Updates}


\author{Sayan Bhattacharya\footnote{Supported by Engineering and Physical Sciences Research Council, UK (EPSRC) Grant EP/S03353X/1.}\\University of Warwick\\S.Bhattacharya@warwick.ac.uk \and Peter Kiss\\University of Warwick\\Peter.Kiss@warwick.ac.uk \and Thatchaphol Saranurak\\University of Michigan\\thsa@umich.edu }

\date{}

\maketitle

\pagenumbering{gobble}
\input{abstract}

\pagebreak 

\pagenumbering{gobble}
\setcounter{tocdepth}{2}
\tableofcontents

\pagebreak
\pagenumbering{arabic}
\input{introduction}

\input{whack-a-mole-mwu}

\input{greedy-mwu}

\input{lower-bounds}

\appendix{}

\input{appendix-whack-a-mole-mwu}

\input{appendix-greedy-mwu-1}
\input{appendix-greedy-mwu-2}

\bibliographystyle{alpha}
\bibliography{references}

\end{document}

%% file: header.tex
\usepackage{times}
\usepackage{amsmath}
\usepackage[ruled,vlined,linesnumbered]{algorithm2e}
\usepackage{amsthm}
\usepackage{graphicx,epsfig}
\usepackage{caption}
\usepackage[noend]{algpseudocode}
\usepackage{tikzit}
\usepackage{bbold}
\usepackage{amssymb}
\usepackage{algorithm2e}
\usepackage{comment}

\newlength{\commentWidth}
\newcommand{\atcp}[1]{\tcp*[r]{\makebox[\commentWidth]{#1\hfill}}}

\theoremstyle{plain} \numberwithin{equation}{section}
\newtheorem{theorem}{Theorem}[section]
\newtheorem{corollary}[theorem]{Corollary}

\newtheorem{conjecture}{Conjecture}
\newtheorem{lemma}[theorem]{Lemma}

\theoremstyle{plain}
\newtheorem{invariant}[theorem]{Invariant}
\newtheorem{property}[theorem]{Property}
\newtheorem{definition}[theorem]{Definition}

\newtheorem{assumption}[theorem]{Assumption}
\newtheorem{problem}[theorem]{Problem}

\newtheorem{claim}{Claim}[section]
\newcommand{\A}{\mathcal{A}}

\newcommand{\eps}{\epsilon}

\global\long\def\opt{\mathrm{OPT}}
\global\long\def\Pack{\mathrm{\mathbb{P}}}
\global\long\def\Cover{\mathrm{\mathbb{C}}}
\global\long\def\one{\mathbb{1}}

\advance \topmargin by -\headheight
\advance \topmargin by -\headsep
\evensidemargin \oddsidemargin
\usepackage{bbm}
\usepackage{xcolor}

\usepackage{enumitem}
\usepackage{xspace}
\renewcommand{\paragraph}[1]{\medskip\noindent{\bf #1}\xspace}
\usepackage{geometry}
\newtheorem{observation}[theorem]{Observation}

\usepackage{xcolor}
\usepackage{nameref}
\definecolor{ForestGreen}{rgb}{0.1333,0.5451,0.1333}
\definecolor{DarkRed}{rgb}{0.65,0,0}
\definecolor{Red}{rgb}{1,0,0}
\usepackage[linktocpage=true,
pagebackref=true,colorlinks,
linkcolor=DarkRed,citecolor=ForestGreen,
bookmarks,bookmarksopen,bookmarksnumbered]
{hyperref}
\usepackage{cleveref}

\usepackage{subfiles}

\global\long\def\Otil{\tilde{O}}

\newcommand{\norm}[1]{\left\lVert#1\right\rVert}

\renewcommand{\S}{\mathcal{S}}

\newcommand{\poly}{\operatorname{poly}}

\ifdefined\ShowComment

\def\thatchaphol#1{\marginpar{$\leftarrow$\fbox{T}}\footnote{$\Rightarrow$~{\sf\textcolor{blue}{#1 --Thatchaphol}}}}

\def\sayan#1{\marginpar{$\leftarrow$\fbox{S}}\footnote{$\Rightarrow$~{\sf\textcolor{red}{#1 --Sayan}}}}

\def\red#1{\textcolor{red}{#1}}

\else
\def\thatchaphol#1{}
\def\sayan#1{}

\def\red#1{#1}

\fi

%% file: abstract.tex
\begin{abstract}
In the \emph{dynamic linear program} (LP) problem, we are given an
LP undergoing updates and we need to maintain an approximately optimal
solution. Recently, significant attention (e.g.~{[}Gupta~et~al.~STOC'17;
Arar~et~al.~ICALP'18, Wajc~STOC'20{]}) has been devoted to the study of 
special cases of dynamic packing and covering LPs, such as the dynamic
fractional matching and set cover problems. But until now, there is
no non-trivial dynamic algorithm for general packing and covering
LPs.

In this paper, we settle the complexity of dynamic packing and covering
LPs, up to a polylogarithmic factor in update time. More precisely, in the \emph{partially
dynamic} setting (where updates can either only relax or only restrict the feasible
region), we give near-optimal deterministic $\epsilon$-approximation
algorithms with polylogarithmic amortized update time. Then, we show
that both partially dynamic updates and amortized update time are
necessary; without any of these conditions, the trivial algorithm that
recomputes the solution from scratch after every update  is essentially the best possible,
assuming SETH.

To obtain our results, we initiate a systematic study of the \emph{multiplicative
weights update} (MWU) method in the dynamic setting. As by-products
of our techniques, we also obtain the first online $(1+\epsilon)$-competitive
algorithms for both covering and packing LPs with polylogarithmic
recourse, and the first streaming algorithms for covering
and packing LPs with linear space and polylogarithmic passes.

\end{abstract}

%% file: introduction.tex
\section{Introduction}
\label{sec:intro}

Packing and covering linear programs (LPs) are defined as follows:
\begin{align*}
\text{Covering LP \ensuremath{\mathbb{\Cover}}:}\quad & \min_{x\in\mathbb{R}_{\ge0}^{n}}\{c^{\top}x\mid Ax\ge b\},\\
\text{Packing LP \ensuremath{\mathbb{\Pack}}:\quad} & \max_{y\in\mathbb{R}_{\ge0}^{m}}\{b^{\top}y\mid A^{\top}y\le c\},
\end{align*}
where $A\in\mathbb{R}_{\ge0}^{m\times n},b\in\mathbb{R}_{\ge0}^{m},c\in\mathbb{R}_{\ge0}^{n}$
contain only non-negative entries. The two LPs are duals to each
other, so we let $\opt$ denote their shared optimum. We say that
$x$ is an \emph{$\epsilon$-approximation} for $\Cover$ if $Ax\ge b$
and $c^{\top}x\le(1+\epsilon)\opt$. Similarly, $y$ is an \emph{$\epsilon$-approximation}
for $\Pack$ if $A^{\top}y\le c$ and $b^{\top}y\ge \opt/(1+\epsilon)$. 

Packing-covering LPs have a wide range of applications in various different contexts, such as in approximation algorithms \cite{luby1993parallel,trevisan1998parallel},
flow control \cite{bartal2004fast}, scheduling \cite{plotkin1995fast},
graph embedding \cite{plotkin1995fast} and auction mechanisms \cite{zurel2001efficient}. Thus, a long line of work is concerned with the study of  efficient algorithms for packing-covering LPs \cite{ZhuO15,allen2019nearly,wang2015unified,KoufogiannakisY14,Young2014nearly,chekuri2018randomized,Quanrud20}, and currently, it is known how to compute  $\epsilon$-approximation in $\Otil(N/\epsilon)$ time \cite{allen2019nearly,wang2015unified}. Here, $N$ denotes the total number of non-zero entries in the constraint matrix.\footnote{$\Otil(\cdot)$ notation hides a $\poly(\log(\frac{NW}{\epsilon}))$ factor, 
where $W$ is the ratio between the largest and  smallest nonzero entries
in the LP.}  

In this paper, we consider a very natural dynamic setting, where an adversary can
update any entry of $A,b$ or $c$. After each update, the goal is
to obtain a new $\epsilon$-approximation quickly. Although this question has been studied since the early 80's \cite{OvermarsV80,OvermarsL81,Eppstein91,AgarwalEM92,Matousek93},
all the existing dynamic algorithms can only handle a very small number of variables
such as 2 or 3 (but they work for general LPs). In contrast,  significant
attention \cite{ArarCCSW18,BhattacharyaHN19,BhattacharyaK21,GuptaK0P17,Wajc20}
has recently been devoted to a special case of dynamic packing and covering
LPs: the dynamic fractional matching and set cover problems. These new algorithms, however, are highly specific and do not extend
to general packing-covering LPs. For the latter, the current best dynamic algorithm  just recomputes a solution from scratch after each update.

\subsection{Our Results}

We resolve  the complexity of dynamic algorithms for general
packing-covering LPs, up to a polylogarithmic factor in the update time.  The following are the two key take-home messages from this paper. (1) In partially dynamic settings, there are deterministic
$\epsilon$-approximation algorithms for this problem that have polylogarithmic amortized
update times. (2) Both partially dynamic updates and amortized
update times are necessary: without any of these conditions, the trivial
algorithm which recomputes an $\epsilon$-approximation to the input LP from scratch is essentially
the best possible. Below, we explain our results more formally.

We say an update  is \emph{relaxing} $\mathbb{\Cover}$ (and \emph{restricting
$\Pack$}) if $u$ increases an entry of $A$ or decreases an entry
of $b, c$. Note that $\opt$ can only decrease due to a relaxing update.
In contrast, an update  is \emph{restricting} $\mathbb{\Cover}$ (and
\emph{relaxing $\Pack$}) if it decreases an entry of $A$ or increases
an entry of $b, c$. A sequence of updates to an LP  is \emph{partially dynamic}
if  either \{all of them are relaxing\}  or  \{all of them are restricting\}. Otherwise,
 the sequence of updates  is \emph{fully dynamic}. Let
$N$ denote the \emph{maximum} number of non-zero entries in the input LP throughout
the sequence of updates. Our main algorithmic result is the following:
\begin{theorem}
\label{thm:main ub}
We can deterministically maintain an $\epsilon$-approximation to a packing-covering LP 
going through $t$ partially dynamic updates in $\Otil\left(N/\epsilon^{3}+t/\epsilon\right)$
total update time. Thus, the amortized update time of our algorithm is $\Otil(1/\epsilon^{3})$.\footnote{The total update time is defined to be the preprocessing time plus the total time taken
to handle all updates. If an algorithm has total update time $O(N\tau_{1}+t\tau_{2})$, then its amortized update time is $O(\max\{\tau_{1},\tau_{2}\})$.} 
\end{theorem}

By generality of packing and covering LPs,  several applications  immediately follow. 
For example,  \Cref{thm:main ub} implies the first decremental $(1+\epsilon)$-approximation algorithm for bipartite maximum matching in {\em weighted} graphs in $\poly(\log(n)/\epsilon)$ amortized update time. Previously, the best known algorithms either took $\Omega(\sqrt{m})$ update time \cite{gupta2013fully} or only worked for unweighted graphs \cite{bernstein2020deterministic,JambulapatiJST22,assadi2022decremental}.  \Cref{thm:main ub} also implies near-optimal incremental/decremental algorithms for maintaining near-optimal fractional solutions for set covers and dominating sets \cite{GuptaK0P17,HjulerIPS19}. 
Generally, \Cref{thm:main ub} reduces every dynamic problem that can be modeled as a packing/covering LP to a dynamic \emph{rounding problem}. Since rounding LPs is one of the most  successful paradigms in designing approximation algorithms in the static setting, we expect that our result will find many other future applications in  dynamic algorithms.

We complement Theorem~\ref{thm:main ub} by giving strong conditional lower bounds. Assuming
SETH,\footnote{The Strong Exponential Time Hypothesis \cite{impagliazzo2001problems}
says that there is no algorithm for solving $k$-CNF SAT with $n$
variables for all $k$ in $2^{n-\Omega(1)}$ time.} we show that a $N^{o(1)}$-approximation algorithm which either handles
fully dynamic updates or guarantees worst-case update time, must have an update time of
$\Omega(N^{1-o(1)})$. As we can  solve the problem from scratch after every update  in $\Otil(N)$ time, our lower bound rules out any non-trivial dynamic algorithm 
in these settings. 
\begin{theorem}
\label{thm:main lb}Under SETH, there exists a $\beta = N^{o(1)}$ such that any dynamic algorithm $\A$ that maintains a $\beta$-approximation to the optimal objective of  a packing or covering LP $\Pack$ satisfies the following conditions.
\begin{itemize}
\item If $\Pack$ undergoes fully dynamic updates, then ${\cal A}$ must
have $\Omega(N^{1-o(1)})$ amortized update time. 
\item If $\Pack$ undergoes partially dynamic updates, then ${\cal A}$
must have $\Omega(N^{1-o(1)})$ worst-case update time. 
\end{itemize}
\end{theorem}

Our lower bounds hold even if the dynamic algorithm only \emph{implicitly}
maintains the solution, meaning that it only maintains an approximation to the value 
of $\opt$. Furthermore, they hold even if all the updates are on the entries
of the constraint matrix $A$. Thus, \Cref{thm:main ub} and \Cref{thm:main lb}
essentially settle the complexity of the dynamic packing-covering
LP problem, up to a polylogarithmic factor in update time.

\subsection{Techniques}

To prove our main result (\Cref{thm:main ub}), we initiate a systematic
study of the \emph{multiplicative weight update} (MWU) framework in
the dynamic setting. MWU is one the most versatile iterative optimization
framework in theoretical computer science. There is a big
conceptual barrier, however, in extending this framework to the dynamic setting. We start by 
explaining how to overcome this conceptual barrier.

\subsubsection{MWU in the dynamic setting: The main challenge}

\label{sec:dynamic MWU possible}

Like other iterative methods, a MWU-based algorithm builds
its solution $x$ in iterations. At each iteration $t$, the solution
$x^{(t)}$ heavily depends on the previous solution $x^{(t-1)}$. So
when we update our original input, we expect that the update affects
the solution $x^{(1)}$ from the first iteration. Then, the change
in $x^{(1)}$ affects $x^{(2)}$, which in turn affects $x^{(3)}$
and so on. There are at least logarithmic many iterations and eventually
the final solution might completely change. Intuitively, this propagation of changes suggests that using MWU in the dynamic setting is a dead
end. How do we overcome this strong conceptual barrier?

As a take-home message, we show that for MWU there is {\em no}
propagation of changes under partially dynamic updates! To be a bit more concrete,
suppose that the current  solution is $x^{(t)}$ and now the adversary
updates the input. As long as the updates are partially dynamic, we can
argue that all the previous solutions $x^{(1)},\dots,x^{(t-1)},x^{(t)}$
remain valid with the new input. Our only task is to compute a valid
$x^{(t+1)}$ for the updated input, which is much more manageable.
To understand why this is possible, we now explain how MWU works
at a  high level.

The explanation below will be generic (i.e., it will hold beyond packing-covering LPs). But for concreteness, consider a feasibility
problem where given a matrix $C\in\{0,1\}^{m\times n}$, we need to
find $x\in\mathbb{R}_{\ge0}^{n}$ where $\one^{\top}x\le1$ and $Cx\ge\one$.
Its dual problem is to find $y\in\mathbb{R}_{\ge0}^{m}$ where $\one^{\top}y\ge1$
and $C^{\top}y\le\one$. MWU-based algorithms can be viewed as two-player
games where the two players, which we call a \emph{Whack-a-Mole} player
and a \emph{Greedy} player, together maintain a pair $(x,y)$ using
the following simple outline:
\begin{enumerate}
\item \textbf{Whack-a-Mole player builds a multiplicative solution:} Start
with $x=\one/n$. If $x$ violates some constraint $i\in[m]$ (i.e.~$C_{i}x<1$),
then \emph{multiplicatively} update $x$ according to $i$ (i.e.~$x_{j}\gets(1+\epsilon)x_{j}$
for all $j$ where $C_{ij}=1$). We refer to this operation as \emph{whacking}
constraint $i$. Whenever $x$ satisfies all constraints, terminate
the game. Basically, this player keeps ``whacking'' violated constraints. 
\item \textbf{Greedy player builds an additive solution: } Start with the
zero vector $y=\boldsymbol{0}$. If there exists a \emph{cheap }coordinate
$i\in[m]$, then \emph{additively} increment $y_{i}$ by one. (In our case, coordinate
$i$ is cheap iff $C_{i}x<1$.)\footnote{For readers familiar with the MWU literature, this is the ``oracle'' problem. Given weight $x$, we need to find a $y'$ which satisfies the ``average constraint'' $x^{\top}C^{\top}y'\le x^{\top}1$. Then, we set $y\gets y+y'$.} Terminate the game after $T$ rounds, for some $T$, and return the
average $y^* =y/T$. Basically, this player keeps greedily ``incrementing'' 
in a cheap direction.
\end{enumerate}
The game  terminates when: either (1)  $x$ satisfying all constraints, or (2) 
via the regret minimization property of MWU \cite{AroraHK12}, we
can show that $y^*$ is approximately feasible (i.e., $\one^{\top}y^* = 1$
and $C^{\top}y^* \le(1+\epsilon)\one$). 

Coming back to the existing literature on MWU, we observe that  the known algorithms can be classified into one of two categories. The algorithms in the first category take the greedy player's
perspective and build their solutions $y$ additively, while updating
the weight $x$ multiplicatively as a helper. These include algorithms
for flow-related problems \cite{fleischer2000approximating,garg2002fast,garg2007faster},
packing-covering LPs \cite{plotkin1995fast,KoufogiannakisY14,Young2014nearly,chekuri2018randomized,Quanrud20},
or constructions of pseudorandom objects \cite{bshouty2016derandomizing}.
In contrast, the algorithms in the second category  take the whack-a-mole player's perspective and build their solutions $x$ 
multiplicatively. These include algorithms for learning and boosting \cite{littlestone1988learning,freund1997decision,freund1999adaptive}. The idea is also implicit in algorithms for geometric set covers \cite{clarkson1993algorithms,bronnimann1995almost,agarwal2014near,chan2020faster}
and explicit in algorithms for online set covers \cite{gupta2021random}.

Now, the high-level explanation as to why MWU works under partially dynamic
updates is simple. Suppose that the updates are restricting only.
Then we take the whack-a-mole player's perspective. All we do is whacking violated
constraints. Given an update, the constraints that we whacked in the
past would still be violated at that time because the update is restricting.
So our past actions are still valid. On the other hand, if the updates are
relaxing only, then we take the greedy player's perspective. Given a relaxing update,
all the cheap directions that we incremented in the past remain cheap at
the present moment as well. This explains why all the past actions
that we committed to are still valid. Hence, we only need to ensure that the
next action is valid with respect to the updated input. 

As the reasoning above is very generic, we expect that this insight will lead to further applications. For
now, we discuss how we efficiently carry out this approach for
packing-covering LPs.

\subsubsection{Width-Independent Packing Covering LP Solvers via Black-box Regret
Minimization}

\label{sec:indep blackbox}

 Existing  static $\epsilon$-approximation algorithms for packing-covering LPs
can be grouped into two types: \emph{width-dependent} and \emph{width-independent}
solvers. The width-dependent  solvers have running time depending
on $\lambda \cdot \opt$, where $\lambda$ is the largest entry in the LP
(e.g.~$\Otil(N(\frac{\lambda\opt}{\epsilon})^{2})$ \cite{plotkin1995fast}
and $\Otil(N\frac{\lambda\opt}{\epsilon^{2}})$ \cite{AroraHK12}).
While the analysis in \cite{plotkin1995fast} is adhoc, \cite{AroraHK12}
gives a very insightful explanation as to how the guarantee of these solvers
follows in a black-box manner from the well-established regret minimization
property of MWU \red{in the {\em experts setting}}.\sayan{I removed the citation here.} This undesirable
dependency on $\lambda  \opt$ in the runtime led to an extensive line of work on
\red{width-independent solvers \cite{ZhuO15,KoufogiannakisY14,Young2014nearly,chekuri2018randomized,Quanrud20}}.
These stronger solvers run in  $\Otil(N/\poly(\epsilon))$ 
time, which is independent of $\lambda\opt$, and many of them are based on MWU. 

There is, however, one curious and unsatisfactory aspect in the literature
on width-independent LP solvers. Namely, each of these MWU-based algorithms
 requires a separate and fine-tuned analysis involving calculations that resemble
the proof of the regret minimization property of MWU. But can we just apply the regret bound of MWU in a blackbox
manner while analyzing a width-indepdendent LP solver? 

We design a new static algorithm for packing-covering LPs that answers this question in affirmative. To the best of our 
knowledge, it is the first width-independent near-linear-time LP solver
whose analysis follows from a black-box application of the regret bound for MWU.
This is explained in details in  \Cref{sec:dual_MWU}. Moreover, the
algorithm is very intuitive and its high-level description is as follows: 
\begin{quote}
There are $O(\log(n)/\epsilon^{2})$ phases. During each phase, we loop
over the constraints: if constraint $i$ is violated, then we ``whack''
it until it is satisfied and move on to the next constraint $i+1$. Whenever $\one^{\top}x>(1+\epsilon)$, we 
start the next phase and normalize $x$ so that $\one^{\top}x=1$.
If we finish the loop, then we return $x$ as an $\epsilon$-approximation to the primal.
After the last phase, we are guaranteed to obtain a
dual $\epsilon$-approximation $y$ certifying that there is ``almost''
no feasible primal solution, i.e., no $x$ where $\one^{\top}x \leq 1$ and
$Cx \geq (1+\epsilon) \cdot \one$.
\end{quote}
Henceforth, we refer to this algorithm as the \emph{whack-a-mole MWU algorithm}.\footnote{Our algorithm resembles the whack-a-mole-based algorithms for computing 
geometric set covers \cite{clarkson1993algorithms,bronnimann1995almost,agarwal2014near,chan2020faster}.
However, these previous algorithms give $O(1)$-approximation, are specifically
described for set covers, and use adhoc calculation related to MWU.
Thus, our whack-a-mole MWU algorithm refines their approximation,
generalizes their applicability, and modularizes their analysis. Furthermore, for
readers familiar with the algorithms of \cite{Young2014nearly,Quanrud20}, our whack-a-mole MWU algorithms
can be interpreted as implementing their algorithms from the whack-a-mole player's
perspective. } The algorithm takes near-linear time because it has few phases and
each phase requires only one-pass scan. In more details, we refer to the
operation of ``whacking constraint $i$ until it is satisfied''
as  \emph{enforceing} constraint $i$. Enforcing a constraint $i$ can be
done via binary search in $\Otil(N_{i})$ time, where $N_{i}$
is the number of non-zero entries in constraint $i$. As each constraint
is scanned and enforced once per phase, the total running time is
$\Otil(\sum N_{i}/\epsilon^{2})=\Otil(N/\epsilon^{2})$. 

The simplicity of this algorithm not only leads us to efficient dynamic implementations
(see \Cref{sec:overview dynamic}), but also efficient streaming
algorithms and online algorithms with recourse (see
\Cref{sec:byproduct}). 

\subsubsection{Dynamic Implementation }

\label{sec:overview dynamic}

We can extend our whack-a-mole MWU algorithm from \Cref{sec:indep blackbox}
to handle restricting updates, in the following natural manner. Suppose that the algorithm obtained
a feasible solution $x$ (because $x$ satisfies all constraints during
the one-pass scan). Now, consider a restricting update to constraint
$i$, say to an entry $C_{ij}$. We just need to check whether $x$
violates constraint $i$ (i.e., if $C_{i}x<1$). If yes, then we  enforce constraint
$i$ (i.e.,~whack it until it is satisfied). No new
constraint gets violated due to this enforcement. As in the static setting, we start the next phase whenever
$\one^{\top}x>(1+\epsilon)$. We proceed like this until we obtain a feasible solution $x$. Once the
algorithm concludes that there is almost no feasible solution, then
we are done because this will remain the case after future restricting updates.
The correctness  follows from the static algorithm.

We have to overcome an issue while analyzing the total update time: Each constraint is
\emph{not} enforced once per phase anymore (because constraint
$i$ can be updated repeatedly within a phase). Accordingly, we
strengthen our {\em analysis} to show that, even under restricting updates,
each constraint is still enforced  $\Otil(1/\epsilon^{2})$
times. But the algorithm remains the same. This leads to a  total update time of $\Otil(N/\epsilon^{2}+t)$, for $t$ updates.\sayan{The introduction is already long. I think it is ok to skip any more details here.}

Our dynamic  whack-a-mole MWU algorithm
applies to covering LPs under restricting updates, and by duality, to 
packing LPs under relaxing updates. There is a technical challenge, however, that prevents us from extending this approach to packing (resp.~covering) LPs under restricting (resp.~relaxing) 
updates.

To handle this case, we instead consider another static $\Otil(N/\epsilon^{2})$-time
algorithm that take the greedy player's perspective \cite{Young2014nearly,Quanrud20}.
We call it the \emph{greedy MWU algorithm}. Given a covering LP (or
even a mixed packing-covering LP), this greedy MWU algorithm can naturally
handle relaxing updates (as discussed in \Cref{sec:dynamic MWU possible}).
Using similar ideas, we show how to dynamize it with only polylogarithmic
overhead in the total update time. Therefore, we obtain algorithms
for covering LPs under relaxing updates, and by duality, for packing LPs under
restricting update via duality.\footnote{Since the greedy MWU algorithm can handle mixed packing-covering LPs under relaxing
updates, via duality this implies algorithms for handling restricting
updates to both covering and packing LPs. So strictly speaking, our dynamic whack-a-mole MWU algorithm is subsumed by the dynamic greedy MWU algorithm. Nevertheless,
we present the dynamic whack-a-mole MWU algorithm as it is simpler to understand and has more modular analysis.}

\paragraph{Discussion on the Simplicity of Our Techniques.}
We show how to seamlessly use  the MWU framework to obtain near-optimal partially dynamic algorithms for packing-covering LPs. We view the simplicity of our approach as an important merit of this paper.
Despite its simplicity, our result implies immediate applications such as a near-optimal decremental maximum {\em weight} matching algorithm, while previous algorithms 
\cite{bernstein2020deterministic,JambulapatiJST22} for this problem require more involved arguments and work only on unweighted graphs.

\subsection{By-Products of Our Techniques }

\label{sec:byproduct}

Our techniques turn out to have implications beyond dynamic
algorithms for packing-covering LPs. 

\paragraph{Online Algorithms with Recourse.}
The online covering LP problem was introduced in the influential work
of Buchbinder and Naor \cite{buchbinder2009online}. This 
is a fundamental problem in the literature on the online primal-dual method; with applications in online set cover, routing, ad-auctions, metrical task systems and many other settings (see the
book \cite{buchbinder2009design}). Several variants of the problem are studied in \cite{gupta2012approximating,azar2013online,gupta2021random}. 

In this problem, we have a covering LP $\Cover$ with $m$ constraints
and $n$ variables. The constraints of $\Cover$ arrive in an online fashion one after another, and we need to maintain a solution to $\Cover$. Crucially, the values of the variables can only
increase through time (i.e.,~any decision made cannot be retracted).
An online covering LP algorithm is $\alpha$-competitive if it maintains
a solution to $\Cover$ whose cost is always at most $\alpha\cdot\opt$.
Buchbinder and Naor gave a tight $\Theta(\log n)$-competitive algorithm for this problem.

Can we beat the lower bound on competitive ratio for a given problem by allowing some amount
of \emph{recourse}? Here, recourse is defined to be the number
of times the algorithm ``retracts'' its decision.  An \red{influential line of work in online algorithms \cite{BernsteinHR18, GuGK13, GuptaK14, GuptaKS14, GuptaK0P17}}
has been devoted towards answering this question for various fundamental problems.  For online covering
LPs, recourse equals the number of times the value of any variable is 
decreased.  Although covering LPs play a central role in online algorithms, until now it was not known whether we can beat the $\Omega(\log n)$ lower bound on the competitive ratio of this problem if we allow for small recourse.  
We answer this question in the affirmative. 

\begin{theorem}
\label{thm:online covering}There is a $(1+\epsilon)$-competitive
algorithm for online covering LPs with $O(n\log(n)\log(nW)/\epsilon^{3})$ total 
recourse. In fact, in this algorithm each variable incurs at most $O(\log(n)\log(nW)/\epsilon^{3})$ recourse. 
\end{theorem}

Note that the recourse bound is completely independent of
the number of constraints.\sayan{I don't think we have any concrete application of this result as of now.} \Cref{thm:online covering} follows almost immediately from our whack-a-mole MWU algorithm
from \Cref{sec:indep blackbox}. Interestingly, our algorithm can be
viewed as a small adjustment to the well-known $O(\log n)$-competitive
algorithm of \cite{buchbinder2009online}. To see this, suppose that we know $\opt$
by guessing. Then here is the  description of our algorithm.
\begin{quote}
Run the algorithm of \cite{buchbinder2009online}. Whenever the objective  of the solution exceeds
$(1+\epsilon)\opt$, scale down the solution by a factor of $(1+\epsilon)$ and continue.
\end{quote}
Our argument shows that we will never scale down too many times. The
algorithm is clearly $(1+\epsilon)$-competitive. In the paper, we
also show similar $(1+\epsilon)$-competitive algorithm for online
packing LPs. Curiously, online algorithms \emph{without} recourse
of this problem is not known.\footnote{Many papers in online algorithms  consider a problem
where the \emph{variables} of a packing LP are revealed one by one, which is a dual of the online covering LP problem. We consider a different setting, where the packing {\em constraints} arrive one after another.}\thatchaphol{Say that online algorithms without recourse for this problem is impossible.}

\paragraph{Streaming algorithms.}
Our whack-a-mole MWU algorithm also works in the streaming setting.
When the rows of the constraint matrix arrive one by one, we obtain
an $\epsilon$-approximation for both packing and covering LPs using
$O(n)$ space and $\Otil(1/\epsilon^{3})$ passes. In contrast, when the
columns of the constraint matrix arrive one by one, we get the same result but with a  space complexity of 
$O(n+m)$. These give the first streaming algorithms for general packing covering
LPs. Previous algorithms in the literature worked only for  special cases such as fractional 
set cover \cite{indyk2017fractional} and fractional matching \cite{ahn2011linear,AhnG18,AssadiJJST22}.\sayan{I omitted the reference to Assadi et al.} 

\paragraph{Mixed Packing-Covering LPs.}
In a \emph{mixed packing-covering} LP $\Pack$ (or \emph{positive}
LP for short), we are given $A_{c}\in\mathbb{R}_{\ge0}^{m_{c}\times n},b_{c}\in\mathbb{R}_{\ge0}^{m_{c}}$
and $A_{p}\in\mathbb{R}_{\ge0}^{m_{p}\times n},b_{p}\in\mathbb{R}_{\ge0}^{m_{p}}$,
and we need to find $x\in\mathbb{R}_{\ge0}^{n}$ such that $A_{c}x\ge b_{c}$
and $A_{p}x\le b_{p}$. This is a generalization of both packing and
covering LPs. An $x$ is an $\epsilon$-approximation to this positive LP if $A_{c}x\ge b_{c}/(1+\epsilon)$ and $A_{p}x\le(1+\epsilon)b_{p}$.  Our dynamic  greedy MWU algorithm can handle positive LPs under relaxing updates.
\begin{theorem}
\label{thm:main mixed}
We can deterministically maintain an $\epsilon$-approximation to a positive LP
undergoing $t$ relaxing updates in $\Otil\left(N/\epsilon^{3}+t/\epsilon \right)$
total update time. Hence, the amortized update time of our algorithm is $\Otil(1/\epsilon^{3})$.
\end{theorem}
\Cref{thm:main mixed} immediately gives the following natural applications captured by mixed packing-covering LPs. For the well-studied load balancing problem (e.g., \cite{harvey2006semi,assadi2020improved}), we obtain a near-optimal decremental algorithm for maintaining $(1+\epsilon)$-approximate fractional assignments, where each update deletes a job. By duality, for the dynamic densest subhypergraph problem \cite{hu2017maintaining,bera2022new}, this implies a near-optimal decremental algorithm for maintaining $(1+\eps)$-approximation, where each update deletes a hyperedge. The near-optimal algorithms for both problems were not known before.

Finally, we leave the question of designing a dynamic algorithm for positive LPs under restricting updates as an interesting open problem.

\subsection{Roadmap for the Rest of the Paper}

In Section~\ref{sec:whack:a:mole}, we present our whack-a-mole MWU algorithm and its extensions to dynamic, streaming and online settings. Section~\ref{sec:prob} gives an overview of the static greedy MWU algorithm for positive LPs~\cite{Young2014nearly,Quanrud20}, and then explains how we extend this algorithm to the setting where the input LP undergoes relaxing updates. Finally, in Section~\ref{sec:lowerbounds:pure}, we present our conditional lower bounds for dynamic packing-covering LPs.

If the reader wishes to treat this as an extended abstract, then we recommend reading until the end of Section~\ref{sec:dynamic_MWU} as that contains the main technical result of this paper.

%% file: whack-a-mole-mwu.tex
\section{The Whack-a-Mole MWU Algorithm}
\label{sec:whack:a:mole}

This section focuses on the whack-a-mole MWU algorithm. We start by considering the following problem. 

\begin{problem}
\label{prob:basic}
Given a matrix $C \in [0, \lambda]^{m \times n}$ where $\lambda > 0$, either return a vector $x \in \mathbb{R}_{\geq 0}^n$ with $\mathbb{1}^{\top} x \leq 1 +\Theta(\epsilon)$ and $C x \geq (1-\Theta(\epsilon)) \cdot \mathbb{1}$, or return a vector $y \in \mathbb{R}_{\geq 0}^m$ with $\mathbb{1}^{\top} y \geq 1-\Theta(\epsilon)$ and $C^{\top} y \leq (1+\Theta(\epsilon)) \cdot \mathbb{1}$.
\end{problem} 

Problem~\ref{prob:basic} corresponds to the following covering LP: $\text{Minimize } \mathbb{1}^{\top} x, \text{ s.t. } C x \geq \mathbb{1} \text{ and } x \in \mathbb{R}_{\geq 0}^n$. We have to either return an approximately feasible solution to this covering LP with objective $\leq 1+\Theta(\epsilon)$, or  return an approximately feasible solution to the dual packing LP with objective $\geq 1-\Theta(\epsilon)$. In Sections~\ref{sec:dual_MWU},~\ref{sec:dynamic_MWU},~\ref{sec:streaming_MWU} and~\ref{sec:online_MWU}, we respectively show that our MWU algorithm solves Problem~\ref{prob:basic}  in static, dynamic, streaming and online settings. Subsequently, in Section~\ref{sec:reductions}, we explain how we can compute an $\epsilon$-approximate optimal solution to a general packing-covering LP if we have an algorithm for  Problem~\ref{prob:basic}.

Throughout this section, an index $i \in [m]$ (resp.~$j \in [n]$) refers to a row (resp.~column) of the matrix $C$, and  $C_{ij} \in [0, \lambda]$ denotes the entry corresponding to the $i^{th}$ row and $j^{th}$ column of $C$. Furthermore, the symbol $v_k$ denotes the $k^{th}$ coordinate of a vector $v \in \mathbb{R}^n$, where $k \in [n]$.

\subsection{Static Whack-a-Mole MWU Algorithm}
\label{sec:dual_MWU}

We  now present the whack-a-mole MWU algorithm for Problem~\ref{prob:basic} in the static setting. We describe the basic template behind the algorithm in Section~\ref{sec:basic:template}. In Section~\ref{sec:whack:a:mole:implementation}, we show how to implement this basic template in near-linear time.

\subsubsection{The Basic Template}
\label{sec:basic:template}

\label{sec:basic:algorithm}

 \begin{figure*}[htbp]
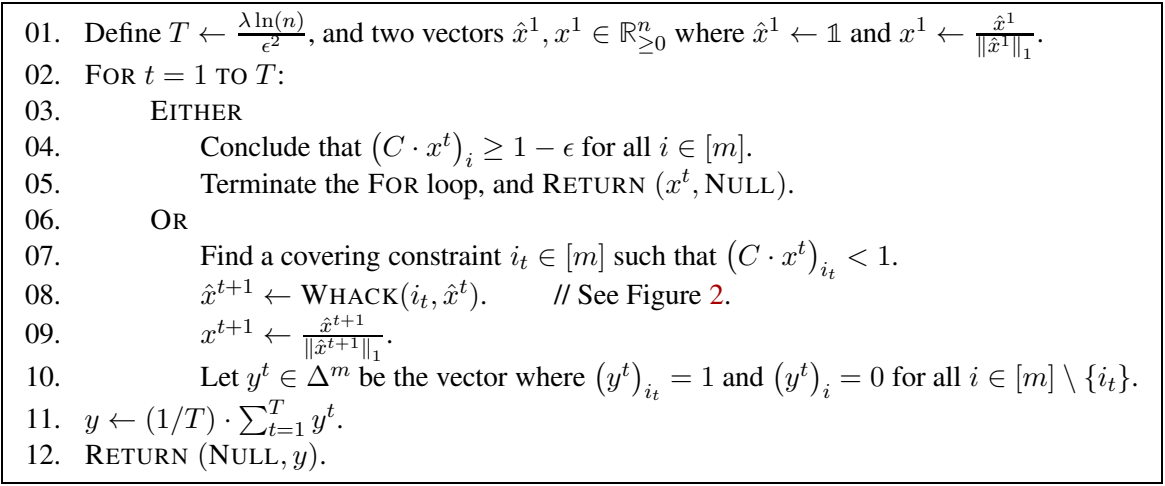

                                                \centerline{
                                                \framebox{
                                                {
                                                                \begin{minipage}{5.5in}
                                                                        \begin{tabbing}                                                                            
                                                                                01. \ \ \=  Define $T \leftarrow \frac{\lambda \ln (n)}{\epsilon^2}$, and two vectors $\hat{x}^1, x^1 \in \mathbb{R}_{\geq 0}^n$ where $\hat{x}^1 \leftarrow \mathbb{1}$ and $x^1 \leftarrow \frac{\hat{x}^1}{\norm{\hat{x}^1}_1}$. \\
                                                                                02. \> \sc{For} $t = 1$ to $T$:\\
                                                                                03. \> \ \ \ \ \ \ \ \ \ \= \sc{Either}\\
                                                                                04. \> \> \ \ \ \ \ \ \ \= Conclude that  $\left(C \cdot x^t\right)_{i} \geq 1 - \epsilon$ for all $i \in [m]$. \\
                                                                                05. \> \> \> Terminate the {\sc For} loop, and {\sc Return} $(x^t, \text{{\sc Null}})$.\\
                                                                                06. \> \>  \sc{Or} \\
                                                                                07. \> \> \> Find a covering constraint $i_t \in [m]$ such that  $\left(C \cdot x^t\right)_{i_t} < 1$. \\
                                                                                08. \> \> \> $\hat{x}^{t+1} \leftarrow \text{\sc{Whack}}(i_t,\hat{x}^t)$.  \qquad // See Figure~\ref{fig:enforce}. \\                                                                 
                                                                                09. \> \> \> $x^{t+1} \leftarrow \frac{\hat{x}^{t+1}}{\norm{\hat{x}^{t+1}}_1}$.\\
                                                                                10. \> \> \> Let $y^t \in \Delta^m$ be the vector where $\left(y^t\right)_{i_t} = 1$ and $\left( y^t \right)_i = 0$ for all $i \in [m] \setminus \{ i_t\}$. \\
                                                                                11. \> $y \leftarrow (1/T) \cdot \sum_{t = 1}^{T} y^t$. \\
                                                                                12. \> \sc{Return} $(\text{{\sc Null}}, y)$.
                                                                                \end{tabbing}
                                                                \end{minipage}
                                                        }
                                                }
                                                }
                                                        \caption{\label{fig:dual_algorithm} The Whack-a-Mole MWU Algorithm}
                                                \end{figure*}
                                                
  \begin{figure*}[htbp]
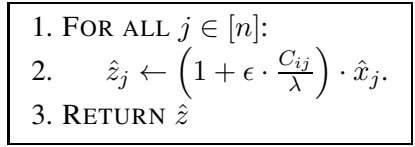

                                                \centerline{
                                                \framebox{
                                                {
                                                                \begin{minipage}{5.5in}
                                                                        \begin{tabbing}                                                                            
                                                                                1.  \= {\sc For all} $j \in [n]$: \\
                                                                                2. \> \ \ \ \ \ \ \= $\hat{z}_j \leftarrow \left(1+\epsilon \cdot \frac{C_{ij}}{\lambda} \right) \cdot \hat{x}_j$. \\
                                                                                3. \> \sc{Return} $\hat{z}$ 
                                                                                \end{tabbing}
                                                                \end{minipage}
                                                        }
                                                }
                                                }
                                                        \caption{\label{fig:enforce} {\sc Whack}$(i, \hat{x})$.}
                                                \end{figure*}
  
The algorithm is described in Figure~\ref{fig:dual_algorithm} and Figure~\ref{fig:enforce}.  It maintains a vector $\hat{x} \in \mathbb{R}_{\geq 0}^n$, where $\hat{x}_j$ denotes the {\em weight} associated with a variable $j \in [n]$ in the covering LP, and the normalized vector $x := \hat{x}/\norm{\hat{x}}_1$.  This ensures that $\mathbb{1}^{\top} \cdot x = 1$. The algorithm runs in $T = \lambda \ln(n)/\epsilon^2$ iterations, where $0 < \epsilon < 1/2$. 

Let $\hat{x}^t$ and $x^t$ respectively denote the status of  $\hat{x}$ and $x$ at the start of iteration $t \in [T]$ of the main {\sc For} loop in Figure~\ref{fig:dual_algorithm}. Before the very first iteration, we initialize $\hat{x}^1 \leftarrow \mathbb{1}$. Subsequently, during any given iteration $t \in [T]$, the algorithm branches into one of the following two cases. 

\medskip
\noindent {\em Case (1).}  It observes that $C x^t \geq (1-\epsilon) \cdot \mathbb{1}$ and returns $(x^t, \text{{\sc Null}})$. In this case, $x^t \in \mathbb{R}_{\geq 0}^n$ is an approximately feasible solution to the covering LP, with objective $\mathbb{1}^{\top} x^t = 1$. 

 \medskip
 \noindent {\em Case (2).} It identifies a violated covering constraint $i_t \in [m]$ with $\left( C x^t \right)_{i_t} < 1$. It then  {\em whacks}  constraint $i_t$, by setting $\hat{x}_j \leftarrow  \left( 1+ \epsilon \cdot \frac{C_{i_t j}}{\lambda}\right)\cdot \hat{x}_j$ for all $j \in [n]$, and accordingly updates the normalized vector $x$. Note that this step increases the relative importance (in the solution $x^t$) of the coordinates $j \in [n]$ that have large $C_{i_t j}$ values.  Thus, intuitively, whacking a violated covering constraint makes progress towards making the solution $x^t$ feasible for the covering LP.  We let $y^t \in \Delta^m$ denote the {\em indicator} vector for the covering constraint $i_t \in [m]$ that gets whacked.

\medskip
\noindent After $T$ iterations, the algorithm returns $(\text{{\sc Null}}, y)$, where $y$ is the average of the vectors $y^1, \cdots, y^T$.  Clearly, we have $\mathbb{1}^{\top} y = 1$.  The next lemma shows that $y$ is an approximately feasible solution to the dual packing LP.\footnote{We remark here that this basic algorithm will work even if the matrix $C$ had negative entries (i.e., if we had $C \in [-\lambda, \lambda]^{m\times n}$), provided we increase the number of iterations by setting $T = \frac{\lambda^2 \ln (n)}{\epsilon^2}$.}  Its proof appears in Section~\ref{sec:lm:whack:a:mole:key}.

\begin{lemma}
\label{lm:whack:a:mole:key}
Suppose that the algorithm in Figure~\ref{fig:dual_algorithm} returns $(\text{{\sc Null}}, y)$ in step~(12). Then $C^{\top} y \leq (1+4\epsilon) \cdot \mathbb{1}$.
\end{lemma}

\begin{theorem}
\label{th:whack:a:mole:key}
The algorithm in Figure~\ref{fig:dual_algorithm} either returns an $x^t \in \mathbb{R}_{\geq 0}^n$  with $\mathbb{1}^{\top} x^t = 1$ and $C x^t \geq (1-\epsilon) \cdot \mathbb{1}$, or it returns a  $y \in \mathbb{R}_{\geq 0}^m$  with $\mathbb{1}^{\top} y = 1$ and $C^{\top} y \leq (1+4\epsilon) \cdot \mathbb{1}$.
\end{theorem}

\begin{proof}
Follows from Lemma~\ref{lm:whack:a:mole:key} and the preceding discussion.
\end{proof}



\subsubsection{Proof of Lemma~\ref{lm:whack:a:mole:key}}
\label{sec:lm:whack:a:mole:key}
We start by recalling the {\em experts} setting~\cite{AroraHK12}. We have $n$ experts $\{1, \ldots, n\}$, who participate the following process that goes on for $T$ rounds. At the start of round $t \in [T]$, we have a weight vector $w^t \in \mathbb{R}_{\geq 0}^n$, where $\left( w^t \right)_j$ denotes the weight associated with expert $j \in [n]$. Initially, before round $1$ begins, we set $w^t = \mathbb{1}$. These weights define a distribution $\mathcal{D}^t$ over the experts, where the probability of picking expert $j \in [n]$ is given by $\frac{\left(w^t\right)_j}{\norm{w^t}_{1}}$. At the start of round $t \in [T]$, our algorithm picks an expert $j \in [n]$ from this distribution $\mathcal{D}^t$. Subsequently, nature reveals a payoff vector $p^t \in [-1, 1]^n$, where $\left( p^t \right)_j \in [-1, 1]$ is the payoff for expert $j \in [n]$ in round $t$. Based on these payoffs, the algorithm then updates the weights of the experts for the next round, by setting $w^{t+1}_j \leftarrow \left(1+\epsilon \cdot \left(p^t\right)_j\right) \cdot w^t_j$ for all $j \in [n]$. 

The lemma below, which is identical to Theorem~$2.5$ in~\cite{AroraHK12},  bounds the total expected payoff obtained by the algorithm in this setting, in terms of the total payoff obtained by any fixed expert.
\begin{lemma}
\label{lm:MWU}
For all experts $j \in [n]$, we have: $$\sum_{t = 1}^T \left( p^t \right)^{\top} \cdot \frac{w^t}{\norm{w^t}_1} \geq \sum_{t = 1}^T \left( p^t\right)_j - \sum_{t = 1}^T \epsilon \cdot \left| \left( p^t\right)_j  \right| - \frac{\ln (n)}{\epsilon}.$$
\end{lemma}

We now map the whack-a-mole algorithm to the experts setting as follows. Each covering constraint $j \in [n]$ corresponds to an expert, each iteration of the {\sc For} loop in Figure~\ref{fig:dual_algorithm} corresponds to a round,  the vector $\hat{x}^t$ corresponds to the weight vector $w^t$ at the start of round $t \in [T]$, and finally, the payoff for an expert $j \in [n]$ in round $t \in [T]$ is given by $\left(p^t\right)_j := (1/\lambda) \cdot C_{i_t, j}$.  

Since $C_{i_t, j} \in [0, \lambda]$, we have $\left| \left( p^t \right)_j \right| = \left( p^t \right)_j$ for all $j \in [n], t \in [T]$. Thus, Lemma~\ref{lm:MWU} implies that:
\begin{eqnarray}
\sum^{T}_{t=1} \left(p^{t}\right)^{\top} \cdot x^t  \geq \sum_{t=1}^T (1-\epsilon) \cdot \left(p^t\right)_j - \frac{\ln(n)}{\epsilon}, \text{ for all experts } j \in [n]. \label{eq:experts:1}
\end{eqnarray}
Diving both sides of the above inequality by $T$, and then rearranging the terms, we get: 
\begin{eqnarray}
\frac{(1-\epsilon)}{T} \cdot \sum_{t=1}^T  \left(p^t\right)_j \leq \frac{1}{T} \cdot \sum^{T}_{t=1} \left(p^{t}\right)^{\top} \cdot x^t  + \frac{\ln (n)}{\epsilon \cdot T}, \text{ for all experts } j \in [n]. \label{eq:experts:10}
\end{eqnarray}
We next upper bound the right hand side (RHS) of~(\ref{eq:experts:10}). Since the algorithm picks a violated covering constraint to whack in each round, we have: $\left(p^{t}\right)^{\top} \cdot x^t = (1/\lambda) \cdot  \left(C x^t\right)_{i_t} \leq (1/\lambda)$. Taking the average of this inequality across all the $T$ rounds, we get:
$(1/T) \cdot \sum^{T}_{t=1} \left(p^{t}\right)^{\top} \cdot x^t \leq (1/\lambda)$.
Since $T = \lambda \ln (n)/\epsilon^2$, we derive the following upper bound on the RHS of~(\ref{eq:experts:10}).
\begin{eqnarray}
\frac{1}{T} \cdot \sum^{T}_{t=1} \left(p^{t}\right)^{\top} \cdot x^t  + \frac{\ln (n)}{\epsilon \cdot T} \leq \frac{1}{\lambda} \cdot (1+\epsilon). \label{eq:experts:11}
\end{eqnarray}
We now focus our attention on the left hand side (LHS) of~(\ref{eq:experts:10}). Fix any expert $j \in [n]$. We first express the payoff obtained by this expert at round $t \in [T]$ in terms of the vector $y^t$, and get: $\left( p^t \right)_j = \frac{1}{\lambda} \cdot C_{i_t, j} = \frac{1}{\lambda} \cdot \left( C^{\top} y^t\right)_j$.
Since $y = \frac{1}{T} \cdot \sum_{t=1}^T y^t$, the average payoff for the expert $j$ across all the $T$ rounds is given by:
\begin{eqnarray}
\frac{1}{T} \cdot \sum_{t=1}^T  \left(p^t\right)_j = \frac{1}{T} \cdot \frac{1}{\lambda} \cdot \sum_{t=1}^T \left( C^{\top} y^t\right)_j = \frac{1}{\lambda} \cdot \left( C^{\top} y\right)_j. \label{eq:experts:12}
\end{eqnarray} 
From~(\ref{eq:experts:10}),~(\ref{eq:experts:11}) and~(\ref{eq:experts:12}), we get:
$$\frac{(1-\epsilon)}{\lambda} \cdot \left( C^{\top} y\right)_j \leq \frac{1}{\lambda} \cdot (1+\epsilon), \text{ and hence } \left( C^{\top} y\right)_j \leq (1+4\epsilon) \text{ for all } j \in [n].$$
The last inequality holds since $\epsilon < 1/2$. This concludes the proof of Lemma~\ref{lm:whack:a:mole:key}.

\subsubsection{A Near-Linear Time Implementation}
\label{sec:whack:a:mole:implementation}

\begin{figure*}[htbp]
\centerline{\framebox{
                                                                \begin{minipage}{5.5in}
                                                                        \begin{tabbing}                                                                            
                                                                                01.  \= $\hat{x}^1 \leftarrow \mathbb{1}$, $t \leftarrow 1$, and $T \leftarrow  \lambda \ln (n)/\epsilon^2$. \\
                                                                                02. \> {\sc Loop} \\ 
                                                                                03. \> \ \ \ \ \   \= $W \leftarrow \norm{\hat{x}^t}_1$. \\
                                                                                04. \> \> {\sc For all $i \in [m]$:} \\
                                                                                05. \> \> \ \ \ \ \  \= {\sc If}  $\left(C \cdot \frac{\hat{x}^t}{W}\right)_i < 1 - \epsilon/2$ {\sc Then}   \\
                                                                                06. \> \> \> \ \ \ \ \ \ \ \=   $\delta \leftarrow \text{{\sc Enforce}}(i, t, \hat{x}^t, W)$.  \qquad \qquad \qquad \qquad // See Figure~\ref{fig:Enforce:new}.  \\
                                                                                07. \> \> \> \> $t \leftarrow t + \delta$. \\
                                                                                08. \> \> \> \> {\sc If} $t = T$, {\sc Then}\\
                                                                                09. \> \> \> \> \ \ \ \ \ \ \ \ \= Terminate the {\sc Loop} and return $(\text{{\sc Null}}, y)$, where $y := (1/T) \cdot \sum_{t' = 1}^T y^{t'}$. \\
                                                                                10. \> \> \> \>  {\sc If} $\norm{\hat{x}^t}_1 > (1 - \epsilon/2)^{-1}\cdot W$, {\sc Then}        \\               
                                                                                11. \> \> \> \> \> {\sc Go to} step (03).  \qquad \qquad \qquad \qquad \qquad // Initiate a new phase. \\
                                                                                12. \> \> Terminate the {\sc Loop} and return $(x^t, \text{{\sc Null}})$, where $x^t := \frac{\hat{x}^t}{\norm{\hat{x}^t}_1}$.
                                                                               \end{tabbing}
                                                                \end{minipage}
                                                        }}
                                                        \caption{\label{fig:MWU:static} An Implementation of the Whack-a-mole MWU Algorithm.}
                                                \end{figure*}

We now show how to implement the whack-a-mole MWU algorithm from Section~\ref{sec:basic:algorithm} in near-linear time. This implementation is outlined in Figure~\ref{fig:MWU:static}, and is based on two key ideas.  (1) We split up the working of the Whack-a-mole MWU algorithm into {\em phases}, and ensure that within any given phase the total weight $\norm{\hat{x}^t}_1$ of all the experts (see the discussion in the beginning of Section~\ref{sec:lm:whack:a:mole:key}) changes by at most a  $(1-\epsilon/2)^{-1}$ multiplicative factor. (2) The weight of an expert $j \in [n]$ can only increase during the course of the algorithm.  Because of this, within a given phase we only need to consider each constraint at most once, provided upon considering the constraint we keep repeatedly whacking it until it gets satisfied.


At the start of a new phase, we set $W \leftarrow \norm{\hat{x}^t}_1$.  Throughout  the phase, the value of $W$ will not change and it will serve as an estimate of the total weight of the experts. Since $W$ remains within a multiplicative $(1-\epsilon/2)^{-1}$ factor of $\norm{\hat{x}^t}_1$,  the  vector $\frac{\hat{x}^t}{W}$ remains a good estimate of the actual solution $x^t := \frac{\hat{x}^t}{\norm{\hat{x}^t}_1}$.  

\begin{figure*}[htbp]
\centerline{\framebox{
                                                                \begin{minipage}{5.5in}
                                                                        \begin{tabbing}                                                                            
                                                                                1.  \=  $\delta \leftarrow \text{{\sc Step-size}}(i, t, \hat{x}^t, W)$. \qquad \qquad \qquad \ \  \  // See Figure~\ref{fig:MWU:Enforce}. \\
                                                                                2. \>  {\sc For} $t' = t$ to $(t+\delta-1)$: \\
                                                                                3. \>  \ \ \ \ \ \ \= $\hat{x}^{t'+1} \leftarrow \text{{\sc Whack}}(i, \hat{x}^{t'})$.  \qquad \qquad \qquad // See Figure~\ref{fig:enforce}. \\
                                                                                4. \> \>  $i_{t'} \leftarrow i$. \\
                                                                                5. \> \>  Let $y^{t'} \in \Delta^m$ be the vector where $\left(y^{t'}\right)_i = 1$ and $\left(y^{t'}\right)_{i'} = 0$ for all $i' \in [m] \setminus \{i\}$.  \\
                                                                                6. \> {\sc Return} $\delta$.
                                                                           \end{tabbing}
                                                                \end{minipage}
                                                        }}
                                                        \caption{\label{fig:Enforce:new} {\sc Enforce}$(i, t, \hat{x}^t, W)$.}
                                                \end{figure*}

                                                \begin{figure*}[htbp]
\centerline{\framebox{
                                                                \begin{minipage}{5.5in}
                                                                        \begin{tabbing} 
                                                                                1.  \= For every integer $\kappa \geq 1$, let $z^\kappa \in \mathbb{R}_{\geq 0}^n$ be such that $\left(z^\kappa\right)_{j} = \left(  1+ \epsilon \cdot \frac{C_{ij}}{\lambda} \right)^\kappa \cdot \left( \hat{x}^t \right)_{j}$ for all $j \in [n]$. \\
                                        2. \>   If  $\left( C \cdot \frac{z^{T-t}}{W} \right)_i < 1$, {\sc Then} \\ 
                                        3. \> \ \ \ \ \ \ \ \= $\delta \leftarrow T-t$. \\                                         
                                        4.  \> {\sc Else} \\
                                        5. \> \>  Using binary search, compute the  the smallest integer $\delta \in [T-t]$  such that $\left( C \cdot \frac{z^{\delta}}{W} \right)_i \geq 1$. \\
                                        6. \>  Return $\delta$ 
                                                                               \end{tabbing}
                                        \end{minipage}
                                                        }}
                                                        \caption{\label{fig:MWU:Enforce} {\sc Step-size}$(i, t, \hat{x}^t, W)$}
                                                \end{figure*}

During a given phase, we scan through all the covering constraints in any arbitrary order.  We now explain how to implement a typical iteration  of this scan, where (say) we are considering the constraint  $i \in [m]$. We first check if $\left( C \cdot \frac{\hat{x}^t}{W} \right)_i \geq 1 - \epsilon/2$, that is, whether the constraint is approximately satisfied. If the answer is yes, then we do nothing with this constraint and proceed to the next iteration of the scan. In contrast, if the answer is no, then we {\em enforce} this constraint by calling the subroutine {\sc Enforce}$(i, t, \hat{x}^t, W)$. This subroutine works as follows. It  finds an integer $\delta \geq 1$ which indicates the minimum number of times the constraint $i \in [m]$ needs to be whacked before it gets satisfied, assuming that we continue to be in the same phase even at the end of all these whacks. Constraint $i$ then repeatedly gets whacked  $\delta$ times, and after that  we set $t \leftarrow t + \delta$. At this point, if $t = T$, then we terminate the algorithm and return a solution for the dual packing LP, as in Section~\ref{sec:basic:algorithm}.  Else if $\norm{\hat{x}^t}_1 > (1-\epsilon/2)^{-1} \cdot W$, because the total weight of the experts increased a lot due to the previous $\delta$ whacks, then we initiate a new phase. 

If at the end of the scan, we observe that $W$ is still an accurate estimate of the total weight $\norm{\hat{x}^t}_1$ {\em and} $t < T$, then we terminate the algorithm and return $x^t := \frac{\hat{x}^t}{\norm{\hat{x}^t}_1}$ as a solution to the primal covering LP.

\begin{lemma}
\label{lm:equivalence}
The procedure described in Figure~\ref{fig:MWU:static} implements the algorithm from Section~\ref{sec:basic:algorithm}.
\end{lemma}

\begin{proof} 
We refer to each call to  {\sc Whack}$(i, \hat{x}^t)$  as a {\em round} (see Section~\ref{sec:lm:whack:a:mole:key}). Both Figure~\ref{fig:dual_algorithm} and Figure~\ref{fig:MWU:static} initialize $\hat{x}^1 \leftarrow \mathbb{1}$ and run for at most $T = \lambda \ln (n)/\epsilon^2$ rounds. We will show that the way Figure~\ref{fig:MWU:static} implements each round and the way it eventually returns a solution  are both consistent with  Figure~\ref{fig:dual_algorithm}.

Towards this end, first consider the scenario where the procedure in Figure~\ref{fig:MWU:static} enforces a constraint $i \in [m]$ at the start of a round $t \in [T]$ by calling the subroutine {\sc Enforce}$(i, t, \hat{x}^t, W)$, which repeatedly whacks the constraint $\delta$ times. From Figure~\ref{fig:Enforce:new} and Figure~\ref{fig:MWU:Enforce}, we infer that 
\begin{equation}
\left( C \cdot \frac{\hat{x}^{t'}}{W} \right)_i < 1 \text{ for all } t \leq t' \leq t+\delta-1. \label{eq:formal:1}
\end{equation}
Since $W$ is the total weight of all the experts at the start of the concerned phrase, and since the weight of any expert can only increase with time, we have: $\norm{x^{t'}}_1 \geq W$ for all $t' \geq t$. Hence, from~(\ref{eq:formal:1}) we get:
\begin{equation}
\left( C \cdot x^{t'} \right)_i < 1 \text{ for all } t \leq t' \leq t+\delta-1, \text{ where } x^{t'} := \frac{\hat{x}^{t'}}{\norm{x^{t'}}_1}.
\end{equation}
This implies that the decisions to whack constraint $i \in [m]$ in successive rounds $t' \in [t, t+\delta-1]$ are consistent with the rule governing the whacking of constraints in Figure~\ref{fig:dual_algorithm}.

Next, note that after enforcing the constraint $i \in [m]$, the procedure in Figure~\ref{fig:MWU:static} sets $t \leftarrow t+\delta$. At this point, if $t = T$, then it decides to return the vector $y := (1/T) \cdot \sum_{t'=1}^T y^{t'}$ as a solution to the dual packing LP. Clearly, this decision is also consistent with steps (11) -- (12) in Figure~\ref{fig:dual_algorithm}.

Finally, consider  the scenario where the procedure in Figure~\ref{fig:MWU:static} returns  a vector $x^{t''}$ in round $t''$ (see step (12) in Figure~\ref{fig:MWU:static}). Focus on  the very last phase, which spans from (say)  round $t'$ to round $t''$, where $t' < t''$. Let $W$ be the value of $\norm{\hat{x}^t}_1$ at the start of this phase. Fix any constraint $i \in [m]$, which was considered (say) at the start of round $t_i \in [t', t'']$ by the {\sc For} loop in Figure~\ref{fig:MWU:static}. Now, there are two possible cases. 

\smallskip
\noindent Case I: The constraint $i$ did not get enforced in this phase. This happens if $\left( C \cdot \frac{\hat{x}^{t_i}}{W} \right)_i \geq 1-\epsilon/2$.  Here, we derive that $\left( C \cdot \frac{\hat{x}^{t''}}{W} \right)_i \geq \left( C \cdot \frac{\hat{x}^{t_i}}{W} \right)_i \geq 1 - \epsilon/2$, since each co-ordinate of $\hat{x}$ can only increase over time. 

\smallskip
\noindent Case II: The constraint $i$ got enforced in this phase, by getting repeatedly whacked $\delta$ times starting from round $t_i$. Thus, we have $t_i+\delta < T$ (otherwise,  the algorithm would return a dual packing solution $y$) and $\left( C \cdot \frac{\hat{x}^{t_i+\delta}}{W} \right)_i \geq 1$. Analogous to Case I, here we  derive that $\left( C \cdot \frac{\hat{x}^{t''}}{W} \right)_i \geq \left( C \cdot \frac{\hat{x}^{t_i+\delta}}{W} \right)_i \geq 1 \geq 1-\epsilon/2$.

\smallskip
To summarize, we have the following guarantee for every constraint $i \in [m]$ at the start of round $t''$. 
\begin{equation}
\left( C \cdot \frac{\hat{x}^{t''}}{W} \right)_i \geq 1 - \epsilon/2. \label{eq:formal:2}
\end{equation}
Since no new phase got initiated just before round $t''$ (see steps (10) -- (11) in Figure~\ref{fig:MWU:static}), we infer that $\norm{\hat{x}^{t''}}_1 \leq (1-\epsilon/2)^{-1} \cdot W$.  Thus, from~(\ref{eq:formal:2}), we get the following guarantee for every constraint $i \in [m]$.
\begin{equation}
\label{eq:formal:3}
\left( C \cdot x^{t''} \right)_i \geq \left( C \cdot \frac{\hat{x}^{t''}}{W} \right)_i \cdot (1-\epsilon/2) \geq (1-\epsilon/2)^2 \geq 1-\epsilon, \text{ where } x^{t''} := \frac{\hat{x}^{t''}}{W}.
\end{equation}
In other words, the vector $x^{t''}$ satisfies the inequality $C \cdot x^{t''} \geq (1-\epsilon) \cdot \mathbb{1}$, and hence the decision to return $x^{t''}$ as a solution to the  covering LP is also consistent with the template described in Figure~\ref{fig:dual_algorithm}. 
\end{proof}

\begin{theorem}
\label{th:MWU:static:approx}
The procedure in Figure~\ref{fig:MWU:static} either returns an $x^t \in \mathbb{R}_{\geq 0}^n$  with $\mathbb{1}^{\top}  x^t = 1$ and $C x^t \geq (1-\epsilon) \cdot \mathbb{1}$, or it returns a  $y \in \mathbb{R}_{\geq 0}^m$  with $\mathbb{1}^{\top} y = 1$ and $C^{\top} y \leq (1+4\epsilon) \cdot \mathbb{1}$.
\end{theorem}

\begin{proof}
Follows from Theorem~\ref{th:whack:a:mole:key} and Lemma~\ref{lm:equivalence}.
\end{proof}

It now remains to bound the running time of this algorithm. Towards this end, the lemma below analyzes the time taken by a call to the subroutine {\sc Step-size}$(i, \hat{x}^t, W)$.

\begin{lemma}
\label{lm:step-size}
A call to {\sc Step-size}$(i, t, \hat{x}^t, W)$ can be implemented in $O\left(N_i \cdot \log^2 \left( \frac{\lambda \log (n)}{\epsilon} \right)\right)$ time, where $N_i$ is the number of non-zero entries in row $i \in [m]$ of the matrix $C \in [0, \lambda]^{m\times n}$.
\end{lemma} 

\begin{proof}
It takes $O(N_i \cdot \log \kappa)$ time to compute $\left( C \cdot \frac{z^\kappa}{W} \right)_i$ for any  $\kappa \geq 1$ (see step (1) of Figure~\ref{fig:MWU:Enforce}). Furthermore, the value of $\left( C \cdot \frac{z^\kappa}{W} \right)_i$ can only increase as we increase $\kappa$. Thus, we can find $\delta$, as in steps (2) - (5) of Figure~\ref{fig:MWU:Enforce}, by doing a binary search in $O(N_i \cdot \log T \cdot \log T)$ time. The lemma follows since  $T = \lambda \ln (n)/\epsilon^2$.
\end{proof}

The next  lemma bounds the time taken to enforce a constraint.

\begin{lemma}
\label{lm:enforce}
A call to  {\sc Enforce}$(i, t, \hat{x}^t, W)$ can be implemented in $O\left( N_i \cdot \log^2 \left( \frac{\lambda \log (n)}{\epsilon} \right)\right)$ time, where $N_i$ is the number of non-zero entries in row $i \in [m]$ of the matrix $C \in [0, \lambda]^{m\times n}$. 
\end{lemma}

\begin{proof}
The key idea is that we do {\em not} need to {\em explicitly} store all the vectors $\hat{x}^{t+1}, \ldots, \hat{x}^{t+\delta}$ and $y^{t}, \ldots, y^{t+\delta-1}$. Instead, using standard data structures, the overall algorithm (outlined in  Figure~\ref{fig:MWU:static}) keeps track of only the following information corresponding to a specific call to {\sc Enforce}$(i, t, \hat{x}^t, W)$.

\smallskip
\noindent (1) The value of $t$ at the start of the call, and the step-size $\delta$ computed by the call. 

\smallskip
\noindent (2) The index $i \in [m]$ corresponding to the call. This index implicitly defines all the vectors $y^{t}, \ldots, y^{t+\delta-1}$. 

\smallskip
\noindent (3) A vector $\hat{x} \in \mathbb{R}_{\geq 0}^n$, which was equal to $\hat{x}^{t}$ just before the call, and needs to be set to $\hat{x}^{t+\delta}$ at the end of the call. This will be required to implement the next iteration of the {\sc For} loop in Figure~\ref{fig:MWU:static}.

\smallskip
The overall algorithm (in Figure~\ref{fig:MWU:static}) can easily recover the answer it needs to return, provided it keeps track of the above mentioned pieces of information corresponding to each call to the {\sc Enforce} subroutine.

Thus, the runtime of a  call to {\sc Enforce}$(i, t, \hat{x}^t, W)$ is dominated by the time spent on the following two tasks: (1) computing the value of $\delta$, and (2) ensuring that the vector $\hat{x}$, which was equal to $\hat{x}^t$ just before the call, equals  $\hat{x}^{t+\delta}$ at the end of the call.  The time needed for the former task is bounded by Lemma~\ref{lm:step-size}, whereas the time needed for the latter task is $O(N_i \cdot \log \delta)$. The lemma follows since $\delta \leq T = \frac{\lambda \ln (n)}{\epsilon^2}$.
\end{proof}

We next bound the total number of phases in the algorithm.

\begin{lemma}
\label{cor:bound:weight}
Throughout the duration of the algorithm outlined in Figure~\ref{fig:MWU:static}, we have: $\norm{\hat{x}^t}_1 \leq n^{(1/\epsilon)}$.
\end{lemma}

\begin{proof}
Consider any round $t \in [T]$, which corresponds to an iteration of the {\sc For} loop in Figure~\ref{fig:dual_algorithm} (see Lemma~\ref{lm:equivalence}). Observe that:
\begin{eqnarray*}
\norm{\hat{x}^{t+1}}_1 - \norm{\hat{x}^t}_1  =  \sum_{j \in [n]} \left( \left( \hat{x}^{t+1}\right)_j - \left( \hat{x}^{t}\right)_j \right) = \sum_{j\in [n]} \left( \hat{x}^t \right)_j \cdot \left( \epsilon \cdot \frac{C_{i_t, j}}{\lambda}\right)   =  \frac{\epsilon}{\lambda} \cdot \left( C \hat{x}^t \right)_{i_t} < \frac{\epsilon}{\lambda} \cdot \norm{\hat{x}^t}_1.
\end{eqnarray*}
The last inequality follows since $\left( C x^t \right)_{i_t} < 1$ and $x^t := \hat{x}^t/\norm{\hat{x}^t}_1$. Rearranging the terms, we get:
\begin{equation}
\label{eq:cor:bound:weight:1}
\norm{\hat{x}^{t+1}}_1 \leq \left( 1 + \frac{\epsilon}{\lambda} \right) \cdot \norm{\hat{x}^t}_1.
\end{equation}
As $T = \frac{\lambda \ln (n)}{\epsilon^2}$ and $\norm{\hat{x}^0}_1 = n$, from~(\ref{eq:cor:bound:weight:1}) we get:
$\norm{\hat{x}^t}_1 \leq \left(1 + \frac{\epsilon}{\lambda}\right)^{T} \cdot \norm{\hat{x}^0}_1 \leq n^{(1/\epsilon)}$ for all $t \in [T]$.
\end{proof}

\begin{corollary}
\label{lm:phase}
The algorithm outlined in Figure~\ref{fig:MWU:static} has at most $O\left( \frac{\log (n)}{\epsilon^2}\right)$ many phases. 
\end{corollary}

\begin{proof}
We initiate a new phase whenever $\norm{\hat{x}^t}$ increases by a multiplicative factor of $(1-\epsilon/2)^{-1}$.  Hence, by Lemma~\ref{cor:bound:weight}, the number of phases is at most $O\left(\log_{(1-\epsilon/2)^{-1}} n^{(1/\epsilon)}\right) = O\left(\frac{\log n}{\epsilon^2}\right)$. 
\end{proof}

We are now ready to bound the total runtime of our algorithm.

\begin{theorem}
\label{th:MWU:static:runtime}
The algorithm outlined  in Figure~\ref{fig:MWU:static} can be implemented in $O\left(N \cdot \frac{\log (n)}{\epsilon^2}  \cdot \log^2 \left( \frac{\lambda \log (n)}{\epsilon}\right) \right)$ time, where $N$ is the total number of non-zero entries in the matrix $C \in [0, \lambda]^{m \times n}$.
\end{theorem}

\begin{proof}
Consider a given phase of the algorithm. From Lemma~\ref{lm:enforce}, it follows that the time taken to enforce any specific constraint $i \in [m]$ in this phase is at most $O\left( N_i \cdot \log^2 \left( \frac{\lambda \log (n)}{\epsilon} \right) \right)$. Since each constraint gets enforced at most once during this phase, the total time spent in this phase is at most: 
$$O\left( \sum_{i\in [m]} N_i  \cdot \log^2 \left( \frac{\lambda \log (n)}{\epsilon} \right) \right) = O\left( N \cdot \log^2 \left( \frac{\lambda \log (n)}{\epsilon} \right) \right).$$

The theorem now follows from Corollary~\ref{lm:phase}.
\end{proof}

\subsection{Dynamic Whack-a-Mole MWU Algorithm for Covering LPs}
\label{sec:dynamic_MWU}


In this section, we focus on designing a dynamic algorithm for Problem~\ref{prob:basic} in the following setting. At preprocessing, we receive a constraint matrix $C \in [0, \lambda]^{m \times n}$. Subsequently,  the matrix $C$ undergoes a sequence of {\em restricting} updates, where each update decreases the value of some entry $C_{ij}$ of the matrix $C$.\footnote{As we will see later in \Cref{sec:reductions:dynamic}, we can assume that all updates are applied only to entries of the matrix $C$ (and not on the objective nor the RHS of the constraints).}  Throughout this sequence of updates, we need to maintain either  a vector $\tilde{x} \in \mathbb{R}_{\geq 0}^n$  with $\mathbb{1}^{\top} \tilde{x} \leq 1+\Theta(\epsilon)$ and $C \tilde{x} \geq (1-\Theta(\epsilon)) \cdot \mathbb{1}$, or  a vector $y \in \mathbb{R}_{\geq 0}^m$  with $\mathbb{1}^{\top} y = 1$ and $C^{\top} y \leq (1+\Theta(\epsilon)) \cdot \mathbb{1}$.

We will show that our implementation of the whack-a-mole MWU algorithm from Section~\ref{sec:whack:a:mole:implementation} seamlessly extends to this dynamic setting. In more details, our dynamic algorithm works as follows.

\medskip
\noindent {\bf Preprocessing:} At preprocessing, we run the static algorithm from Section~\ref{sec:whack:a:mole:implementation} on the input matrix $C$. Depending on its outcome, we consider one of the following two cases.

\smallskip
\noindent Case I: The  static algorithm returns a vector $y \in \mathbb{R}_{\geq 0}^m$, as in step (09) of Figure~\ref{fig:MWU:static}.  It follows that $\mathbb{1}^{\top} y = 1$ and $C^{\top} y \leq (1+4\epsilon) \cdot \mathbb{1}$ (see Theorem~\ref{th:MWU:static:approx}). In this case, the vector $y$ will continue to remain an approximately feasible solution to the dual packing LP as the matrix $C$ undergoes restricting entry updates in future. Hence, our dynamic algorithm terminates without having to process any update.

\smallskip
\noindent Case II: The static algorithm  returns a vector $x^t := \hat{x}^t/\norm{\hat{x}^t}_1$, as in step (12) of Figure~\ref{fig:MWU:static}. It follows that $\mathbb{1}^{\top} x^t = 1$ and $C x^t \geq (1-\epsilon) \cdot \mathbb{1}$ (see Theorem~\ref{th:MWU:static:approx}).  Our dynamic algorithm, however, will explicitly maintain only the vector $\tilde{x}^t := \hat{x}^t/W$, which is a very good approximation to $x^t$ since $W \leq \norm{\hat{x}^t}_1 \leq (1-\epsilon/2)^{-1} \cdot W \leq (1+\epsilon) \cdot W$. Thus, in this case, after preprocessing our dynamic algorithm returns a vector $\tilde{x}^t$ which satisfies $\mathbb{1}^{\top} \tilde{x}^t \leq 1+\epsilon$ and $C \tilde{x}^t \geq (1-\epsilon) \cdot \mathbb{1}$. 

\medskip
\noindent {\bf Handling a restricting entry update to $C$:} Consider an update which decreases the value of some entry $C_{ij}$ of the matrix $C$. To handle this update, we simply run the steps (05) - (11) of Figure~\ref{fig:MWU:static}. In words, we  observe that if any constraint in the covering LP gets violated due to this update, then it must be the constraint $i \in [m]$. Thus, we  check whether the current solution $\tilde{x}^t := \hat{x}^t/W$ approximately satisfies constraint $i$. If not, then we enforce that constraint by repeatedly whacking it. At the end of this process, we end up in one of three possible cases. (Case 1): $W$ is no longer an accurate estimate of $\norm{\hat{x}^t}_1$. In this case, we initiate a new phase. (Case 2): $t = T$. Here, we return an approximately feasible solution $y$ to the dual packing LP, and we terminate the dynamic algorithm since $y$ remains a valid dual solution after any future update. (Case 3): If we are neither in Case 1 nor in Case 2, then we infer that $\tilde{x}^t$ is now an approximately feasible solution to the covering LP, with $\mathbb{1}^{\top} \tilde{x}^t \leq 1+\epsilon$ and $C \tilde{x}^t \geq (1-\epsilon) \cdot \mathbb{1}$. This holds because repeatedly whacking constraint $i$ does not lead to any other constraint being violated, provided we remain in the same phase.

\medskip
\noindent {\bf Rounds/phases:} Before proceeding any further, we recall that a new {\em round} begins whenever we whack a constraint, and the variable $t$ denotes the total number of rounds we have seen so far. In contrast, a new {\em phase} begins whenever  $\norm{\hat{x}^t}_1$ increases by a multiplicative factor of $(1-\epsilon/2)^{-1}$. 

\medskip
It now remains to analyze our dynamic algorithm. We start by noting that Theorem~\ref{th:MWU:static:approx} and Lemma~\ref{cor:bound:weight} seamlessly extend to the dynamic setting. We respectively summarize the analogues of these statements in Theorem~\ref{th:MWU:static:approx:dyn} and Lemma~\ref{cor:bound:weight:dyn} below. Their proofs are deferred to Appendix~\ref{sec:dyn:missing:proofs}.

\begin{theorem}
\label{th:MWU:static:approx:dyn}
If our dynamic algorithm returns a vector $y \in \mathbb{R}_{\geq 0}^m$ after handling an update, then $\mathbb{1}^{\top} y = 1$ and $C^{\top} y \leq (1+4\epsilon) \cdot \mathbb{1}$. In contrast, if our dynamic algorithm returns a vector $\tilde{x}^t := \hat{x}^t/W \in \mathbb{R}_{\geq 0}^n$ after handling an update, then $\mathbb{1}^{\top} \tilde{x}^t \leq 1+\epsilon$ and $C \tilde{x}^t \geq (1-\epsilon) \cdot \mathbb{1}$.
\end{theorem}

\begin{lemma}
\label{cor:bound:weight:dyn}
Throughout the duration of our dynamic algorithm, we have $\norm{\hat{x}}^t_1 \leq n^{(1/\epsilon)}$.
\end{lemma}

We now focus on bounding the total update time of our dynamic algorithm. The key  challenge here is to derive an upper bound on the maximum number of times a given constraint $i \in [m]$ can get enforced, over the entire duration of the algorithm.\footnote{We say that constraint $i \in [m]$ gets enforced whenever we call {\sc Enfore}$(i, t, \hat{x}^t, W)$.} This is done in Lemma~\ref{lm:key:dynamic}, whose proof appears in Section~\ref{sec:lm:key:dynamic}.

\begin{lemma}
\label{lm:key:dynamic}
Throughout the entire duration of our dynamic algorithm, a given constraint $i \in [m]$ can get enforced at most $O\left( \frac{\log n}{\epsilon^2} \cdot \log \left( \frac{\lambda \log n}{\epsilon} \right) \right)$ times.
\end{lemma}

We are now ready to bound the total update time of our dynamic algorithm.

\begin{theorem}
\label{th:MWU:static:runtime:dyn}
To handle any sequence of  $\tau$ restricting entry updates to  $C$,  our dynamic algorithm takes   $O\left(\tau + N \cdot \frac{\log (n)}{\epsilon^2} \cdot \log^3 \left( \frac{\lambda \log (n)}{\epsilon} \right)\right)$ time, where $N$ is the number of non-zero entries in $C$ at preprocessing.
\end{theorem}

\begin{proof}(Sketch)
The total time update time is dominated by the time spent on the following two tasks. 

\medskip
\noindent {\em Task I.} After an update decreases the value of some entry $C_{ij}$ of the input matrix $C$, decide whether the constraint $i \in [m]$ is approximately satisfied, i.e., whether $\left( C \cdot \frac{\hat{x}^t}{W} \right) > 1-\epsilon/2$.

\medskip
\noindent {\em Task II.} Enforce a given constraint $i \in [m]$. 

\medskip
We first focus on bounding the total time spent on Task II. Towards this end, we first adapt the argument in the proof of Lemma~\ref{lm:enforce}. This leads us to conclude that whenever we enforce a constraint $i \in [m]$, it takes $O\left( N_i \cdot \log^2 \left( \frac{\lambda \log (n)}{\epsilon} \right)\right)$ time, where $N_i$ is the number of non-zero entries in row $i \in [m]$ of the input matrix $C$ at preprocessing. Hence, Lemma~\ref{lm:key:dynamic} implies that the total time spent on enforcing a given constraint $i \in [m]$ is at most $O\left( N_i \cdot \frac{\log (n)}{\epsilon^2} \cdot \log^3 \left( \frac{\lambda \log (n)}{\epsilon} \right)\right)$. Summing this up over all the constraints $i \in [m]$, we infer that the total time spent on Task II is at most:
\begin{equation} \label{eq:task:2}
O\left( \sum_{i \in [m]} N_i \cdot \frac{\log (n)}{\epsilon^2} \cdot \log^3 \left( \frac{\lambda \log (n)}{\epsilon} \right)\right) = O\left( N \cdot \frac{\log (n)}{\epsilon^2} \cdot \log^3 \left( \frac{\lambda \log (n)}{\epsilon} \right)\right).
\end{equation}

It now remains to bound the total time spent on Task I. Towards this end, we maintain a variable $\hat{z}_j$ for each co-ordinate $j \in [n]$. We always ensure that $\hat{z}_j = (1+\epsilon)^{\kappa}$ for some nonnegative integer $\kappa \geq 0$.  Furthermore, we ensure that $\hat{z}_j$ always lies within a multiplicative $(1+\epsilon)$ factor of $\left(\hat{x}^t\right)_j$. In other words, the value of $\hat{z}_j$ always forms an accurate estimate of $\left(\hat{x}^t\right)_j$. Finally, for each constraint $i \in [m]$, we explicitly maintain the value of $\left( C \cdot \frac{\hat{z}}{W} \right)_i$. This way, we can keep track of the value of $\left( C \cdot \frac{\hat{x}^t}{W} \right)_i$ for all $i \in [m]$, within a multiplicative factor of $(1+\epsilon)$.  This is sufficient for us to detect whether a given constraint $i \in [m]$ is approximately satisfied, in $O(1)$ time after an update.

Fix any co-ordinate $j \in [n]$. Note that $\left( \hat{x}^1 \right)_j = 1$, and by Lemma~\ref{cor:bound:weight:dyn} we have $\left(\hat{x}^t\right)_{j} \leq \sum_{j' \in [n]}\left(\hat{x}^t\right)_{j'} \leq n^{(1/\epsilon)}$ for all $t \in [T]$. In  words, the value of $\left(\hat{x}^t\right)_j$ always lies in the interval $\left[1, n^{(1/\epsilon)}\right]$. We need to update the estimate $\hat{z}_j$ each time the value of $\left(\hat{x}^t\right)_j$ increases by a multiplicative factor of $(1+\epsilon)$. Thus, throughout the duration of our algorithm,  the value of $\hat{z}_j$ gets updated at most $\log_{(1+\epsilon)} n^{(1/\epsilon)} = O(\log n/\epsilon^2)$ times. 

Finally, whenever the value of $\hat{z}_j$ changes, we need to spend  an additional $O(N^j)$ time to reflect this change in the values of $\left( C \cdot \frac{\hat{z}}{W}\right)_i$ for all $i \in [m]$, where $N^j$ denotes the number of non-zero entries in column $j \in [n]$ of the matrix $C$ at preprocessing. Hence, the total time spent in this manner, in order to maintain the estimates $\hat{z}_j$ and $\left( C \cdot \frac{\hat{z}}{W}\right)_i$, is at most $O\left(\sum_{j \in [n]} N^j \cdot (\log n/\epsilon^2) \right) = O(N \log n/\epsilon^2)$.  So the total time spent by our algorithm on Task I is at most $O(\tau+N \log n/\epsilon^2)$. The theorem now follows from~(\ref{eq:task:2}).
 \end{proof}

\subsubsection{Proof of Lemma~\ref{lm:key:dynamic}}
\label{sec:lm:key:dynamic}

Throughout the proof, fix a constraint $i \in [m]$.  We associate a {\em step-size} with each enforcement of this constraint. Specifically, suppose that the constraint gets enforced in round $t$ with step-size $\delta$. Then the constraint gets whacked $\delta$ times during this enforcement (see Figure~\ref{fig:Enforce:new}). Furthermore, note that the step-size $\delta$ lies in the range $[1, T]$, according to Figure~\ref{fig:MWU:Enforce}. We discretize this range $[1, T]$ into $O(\log T)$ many intervals in powers of $2$. Armed with this discretization, we now assign an integral {\em rank} to each enforcement of this constraint. Specifically, a given enforcement has rank $1 \leq \kappa \leq O(\log T)$ iff its step-size $\delta \in \left[2^{\kappa-1}, 2^{\kappa}\right)$. We will now bound the maximum number of enforcements of a given rank encountered by constraint $i$.

\begin{claim}
\label{cl:rank:bound}
Fix any integer $1 \leq \kappa \leq O(\log T)$. Throughout the duration of our dynamic algorithm, the constraint $i \in [m]$ encounters at most $O\left(\frac{\log n}{\epsilon^2}\right)$ many enforcements of rank $\kappa$.
\end{claim}

\begin{proof}
Suppose that the constraint $i \in [m]$ encounters $\gamma+1$ enforcements of rank $\kappa$. Furthermore, suppose that these enforcements occurr in rounds $t_1 < t_2 < \cdots < t_{\gamma} < t_{\gamma+1}$. We will show that $\gamma = O(\log n/\epsilon^2)$.

Focus on the second-last of these enforcements, that occur in round $t_{\gamma}$. Just before this enforcement, we have $\left( C \cdot \frac{\hat{x}^t}{W} \right)_i < 1-\epsilon/2$. In contrast, just after this enforcement, we have $\left( C \cdot \frac{\hat{x}^t}{W} \right)_i \geq 1$.\footnote{Here, we rely on the fact that this is not the last enforcement of constraint $i$, for otherwise we might  execute step (3) in Figure~\ref{fig:MWU:Enforce}.} So the value of $\left( C \cdot \frac{\hat{x}^t}{W} \right)_i$ increases by at least a multiplicative factor of $1+\Theta(\epsilon)$ during this enforcement. Hence, there must exist some co-ordinate $j^* \in [n]$ such that $\left( \hat{x}^t \right)_{j^*}$ also increases by at least a multiplicative factor of $1+\Theta(\epsilon)$ during the same enforcement in round $t_{\gamma}$. On the other hand, since the concerned enforcement has rank $\kappa$, it has a step-size of at most $2^{\kappa}$. Thus, due to this enforcement $\left( \hat{x}^t \right)_{j^*}$ increases by at most a  factor of $\left( 1+ \epsilon \cdot \frac{C_{ij^*}^{t_{\gamma}}}{\lambda}\right)^{2^{\kappa}}$, where $C^{t_{\gamma}}_{ij^*}$ denotes the value of $C_{ij^*}$ just before round $t_{\gamma}$. Accordingly, we infer that:
\begin{equation}
\label{eq:key:inference}
\left( 1+ \epsilon \cdot \frac{C_{ij^*}^{t_{\gamma}}}{\lambda}\right)^{2^{\kappa}} \geq 1+\Theta(\epsilon).
\end{equation}

Next, consider any previous enforcement with rank $\kappa$ that  occurs in some round $t_{\gamma'}$ (where $1 \leq \gamma' \leq \gamma$), and focus on the same co-ordinate $j^*$. Since the concerned enforcement has step-size at least $2^{\kappa-1}$, we infer that due to this enforcement $\left( \hat{x}^t \right)_{j^*}$ increases by at least a multiplicative factor of:
$$\left( 1+ \epsilon \cdot \frac{C_{ij^*}^{t_{\gamma'}}}{\lambda}\right)^{2^{\kappa-1}} = \sqrt{\left( 1+ \epsilon \cdot \frac{C_{ij^*}^{t_{\gamma}}}{\lambda}\right)^{2^{\kappa}}} \geq \sqrt{1 + \Theta(\epsilon)} = 1+\Theta(\epsilon).$$
The inequality follows from~(\ref{eq:key:inference}). To summarize, we conclude that during each of the enforcements in rounds $t_1, \ldots, t_\gamma$, the value of $\left(\hat{x}^t\right)_{j^*}$ increases  by at least a  factor of $1+\Theta(\epsilon)$. 

Finally, observe that   $\left(\hat{x}^t\right)_{j^*} = 1$ at the start of the algorithm, and that the value of $\left( \hat{x}^t\right)_{j^*}$ increases monotonically with time. Thus, at the end of the enforcement in round $t_{\gamma}$, we have:
\begin{equation}
\label{eq:key:inference:2}
 \left( 1 + \Theta(\epsilon) \right)^\gamma \leq \left( \hat{x}^t \right)_{j^*} \leq \sum_{j \in [n]}  \left( \hat{x}^t \right)_{j}  \leq n^{(1/\epsilon)}.
\end{equation}
The last inequality holds due to Lemma~\ref{cor:bound:weight:dyn}. From~(\ref{eq:key:inference:2}), we infer that $\gamma = O(\log n/\epsilon^2)$. 
\end{proof}

Claim~\ref{cl:rank:bound} implies that a given constraint $i \in [m]$ is enforced at most $O\left(\frac{\log n}{\epsilon^2} \cdot \log T\right)$ times by our dynamic algorithm. Lemma~\ref{lm:key:dynamic} now follows from the fact that $T = \lambda \ln (n)/\epsilon^2$.

\subsection{Streaming Whack-a-Mole MWU Algorithm}
\label{sec:streaming_MWU}

In this section, we focus on designing a streaming algorithm for Problem~\ref{prob:basic} in the following setting. We get a matrix $C \in R_{\geq 0}^{m \times n}$ as input. The rows of  $C$ are stored one after the other in a read-only repository. We wish to design an algorithm which makes one or more passes through this repository, and then either outputs a vector $x \in \mathbb{R}_{\geq 0}^n$ with $\mathbb{1}^{\top} x = 1$ and $C x \geq (1-\epsilon) \cdot \mathbb{1}$, or outputs a vector $y \in \mathbb{R}_{\geq 0}^m$ with $\mathbb{1}^{\top} y = 1$ and $C^{\top} y \leq (1+4\epsilon) \cdot \mathbb{1}$.  The goal is to minimize the number of passes and the space complexity (excluding the space taken up by the repository to store the input) of the algorithm. 

We observe that the static whack-a-mole MWU algorithm, as described in Section~\ref{sec:whack:a:mole:implementation}, immediately extends to this streaming setting if we allow the algorithm to have a space complexity of $O(m+n)$. Specifically, we can implement the main {\sc For} loop in  Figure~\ref{fig:MWU:static} in a single pass as follows. The algorithm explicitly maintains a vector $\hat{x} = \hat{x}^t \in \mathbb{R}_{\geq 0}^n$, the values  $W, t \in \mathbb{R}_{\geq 0}$, and a vector $y^*  = (1/T) \cdot \sum_{t' = 1}^{t} y^{t'} \in \mathbb{R}_{\geq 0}^m$ in its memory. Note that when $t = T$, we have $y^* = y$. While making a pass through the read-only repository, suppose that the algorithm  encounters row $i \in [m]$ of the matrix $C$. The algorithm checks whether $\left(C \cdot \frac{\hat{x}}{W} \right)_i < 1- \epsilon/2$, and if yes, then it enforces the constraint $i \in [m]$. Using the vectors $\hat{x}$ and $y^*$, this enforcement step can be performed without incurring any extra overhead in the space complexity. 

Since we can implement the main {\sc For} loop in Figure~\ref{fig:MWU:static} in one pass, the total number of passes is equal to the number of phases of the whack-a-mole algorithm. The next theorem now follows from Corollary~\ref{lm:phase}. 

\begin{theorem}
\label{th:streaming:main}
Consider a streaming setting where the rows of the  matrix $C \in [0, \lambda]^{m \times n}$ arrive one after the other. Then there is a deterministic streaming algorithm with space complexity $O(m+n)$ that makes $O\left( \frac{\log n}{\epsilon^2}\right)$ many passes through this stream, and either returns a vector $x \in \mathbb{R}_{\geq 0}^n$ with $\mathbb{1}^{\top} x = 1$ and $C x \geq (1-\epsilon) \cdot \mathbb{1}$, or returns a vector $y \in \mathbb{R}_{\geq 0}^m$ with $\mathbb{1}^{\top} y = 1$ and $C^{\top} y \leq (1+4\epsilon) \cdot \mathbb{1}$. 
\end{theorem}

Next, note that if we are only required to return either the vector $x \in \mathbb{R}_{\geq 0}^n$ or a special symbol {\sc Null} (indicating that the dual packing LP has an approximately feasible solution with objective $= 1$), then we can further reduce the space complexity of our streaming algorithm. This holds because in such a scenario we only need to maintain the vector $\hat{x} \in \mathbb{R}_{\geq 0}^n$ and the values $W, t \in \mathbb{R}_{\geq 0}$. In particular, we no longer need to maintain the vector $y^* \in \mathbb{R}_{\geq 0}^m$ while making a pass through the read-only repository. Instead, when we observe that $t = T$, we simply return {\sc Null}. This leads to the following corollary.

\begin{corollary}
\label{cor:streaming:main}
Consider a streaming setting where the rows of the matrix $C \in [0,\lambda]^{m \times n}$ arrive one after the other.  Then there is a deterministic streaming algorithm with space complexity $O(n)$ that makes $O\left( \frac{\log n}{\epsilon^2}\right)$ many passes through this stream, and either returns a vector $x \in \mathbb{R}_{\geq 0}^n$ with $\mathbb{1}^{\top} x = 1$ and $C x \geq (1-\epsilon) \cdot \mathbb{1}$, or returns {\sc Null}. In the latter case, it is guaranteed that there exists a vector $y \in \mathbb{R}_{\geq 0}^m$ with $\mathbb{1}^{\top} y = 1$ and $C^{\top} y \leq (1+4\epsilon) \cdot \mathbb{1}$ (although the algorithm does not return $y$). 
\end{corollary}

\subsection{Online Whack-a-Mole MWU Algorithm}
\label{sec:online_MWU}
In this section, we focus on designing an online algorithm for Problem~\ref{prob:basic} in the following setting.  An adversary chooses a matrix $C \in [0, \lambda]^{m \times n}$ as input. In the beginning, the values of $n$ and $\lambda$ are public knowledge, whereas only the adversary knows the value of $m$ and the entries  of the matrix  $C$. Subsequently, the adversary  reveals to us the rows of this matrix one after another. We need to maintain a vector $x \in \mathbb{R}_{\geq 0}^{n}$ such that $\mathbb{1}^{\top} x \leq 1+\Theta(\epsilon)$ and $C x \geq (1-\Theta(\epsilon)) \cdot \mathbb{1}$ until a certain point in time, and after that we need to terminate our algorithm by returning a  $y \in \mathbb{R}_{\geq 0}^m$ such that $\mathbb{1}^{\top} y \geq 1 - \Theta(\epsilon)$ and $C^{\top} y \leq (1+\Theta(\epsilon)) \cdot \mathbb{1}$. We incur a recourse of one each time we decrease the value of some variable $\left(x\right)_j, j \in [n]$. Our goal is to design an algorithm in this online setting with as small total recourse as possible.

We observe that the dynamic whack-a-mole MWU algorithm, as described in Section~\ref{sec:dynamic_MWU}, seamlessly extends to this online setting. Thus, we maintain the vector $\tilde{x}^t := \hat{x}^t/W$. Whenever a new row $i \in [m]$ of the matrix $C$ arrives, we check whether $\left( C \cdot \frac{\hat{x}^t}{W} \right) < 1 - \epsilon/2$, and if the answer is yes, then we enforce the corresponding constraint $i \in [m]$.  The correctness of this algorithm follows from Theorem~\ref{th:MWU:static:approx:dyn}.

We now derive an upper bound on the number of phases. Corollary~\ref{cor:bound:weight:dyn} implies that $\norm{\hat{x}^t}_1 \leq n^{(1/\epsilon)}$ throughout the duration of the algorithm. Since  $\norm{\hat{x}^1}_1 = n$ at the start of the algorithm, and since we initiate a new phase whenever $\norm{\hat{x}^t}_1$ increases by a multiplicative factor of $1+\Theta(\epsilon)$, it follows that the total number of phases is at most $\log_{(1+\Theta(\epsilon))} n^{(1/\epsilon)} = O\left(\frac{\log n}{\epsilon^2}\right)$.

Next, note that within a given phase our algorithm incurs {\em zero} recourse. This holds because the values $\left(\hat{x}^t/W\right)_j$ can only increase when we enforce a constraint within the phase (as $W$ does not change). On the other hand,  we incur a  recourse of $n$ each time we initiate a new phase, since we have $n$ variables and  each of these variables decreases its value as we increase $W$ at the start of the phase. Hence, the total recourse of the algorithm is  $n$ times the number of phases. This leads  to the following theorem.

\begin{theorem}
\label{th:online:main}
Consider an online setting where the rows of the matrix $C \in [0, \lambda]^{m \times n}$ arrive one after another. Let $N$ denote the number of non-zero entries in $C$. There is a deterministic online algorithm with total recourse $O\left( \frac{n \log n}{\epsilon^2}\right)$
which has the following property. It maintains a vector $x \in \mathbb{R}_{\geq 0}^n$ with $\mathbb{1}^{\top} x \leq 1+\Theta(\epsilon)$ and $C x \geq (1-\Theta(\epsilon)) \cdot \mathbb{1}$ until a certain point in time, and after that it terminates and returns a vector $y \in \mathbb{R}_{\geq 0}^m$ with $\mathbb{1}^{\top} y \geq 1 - \Theta(\epsilon)$ and $C^{\top} y \leq (1+\Theta(\epsilon)) \cdot \mathbb{1}$.
\end{theorem}

\subsection{Reductions to General Packing-Covering LPs}
\label{sec:reductions}

Consider a matrix $C \in \mathbb{R}_{\geq 0}^{m \times n}$, vectors $a \in \mathbb{R}_{> 0}^n$, $b \in \mathbb{R}_{> 0}^m$, and a generic covering LP defined below.
\begin{eqnarray}
\label{eq:covering:general}
\text{Minimise } a^{\top}  x  & \ \ & \text{ s.t. } C x \geq b \text{ and } x \in \mathbb{R}_{\geq 0}^n. 
\end{eqnarray}
The dual of the above LP is given by:
\begin{eqnarray}
\label{eq:packing:general}
\text{Maximise } b^{\top}  y  & \ \ & \text{ s.t. } C^{\top} y \leq a \text{ and } y \in \mathbb{R}_{\geq 0}^m. 
\end{eqnarray}

Let {\sc Opt} be the optimal objective value of this pair of primal and dual LPs. Say that  $x \in \mathbb{R}_{\geq 0}^n$ is an {\em $\epsilon$-approximate optimal solution} to LP~(\ref{eq:covering:general}) iff $a^{\top} x \leq (1 + \Theta(\epsilon)) \cdot \text{{\sc Opt}}$ and $Cx \geq (1-\Theta(\epsilon)) \cdot b$. Similarly, say that $y \in \mathbb{R}_{\geq 0}^m$ is an {\em $\epsilon$-approximate optimal solution} to LP~(\ref{eq:packing:general}) iff $b^{\top} y \geq (1 - \Theta(\epsilon)) \cdot \text{{\sc Opt}}$ and $C^{\top} y \leq (1+\Theta(\epsilon)) \cdot a$.\footnote{This notion of an $\epsilon$-approximation is equivalent to the one defined at the start of Section~\ref{sec:intro}, upto a multiplicative factor of $1+\Theta(\epsilon)$. To see why this is true, consider an $x$ such that $1^{\top} x \leq (1+\Theta(\epsilon)) \cdot \text{{\sc Opt}}$ and $Cx \geq (1-\Theta(\epsilon)) \cdot b$. Then it follows that $1^{\top} \tilde{x} \leq (1+\Theta(\epsilon)) \cdot \text{{\sc Opt}}$ and $C \tilde{x} \geq b$, where $\tilde{x} := \frac{x}{1-\Theta(\epsilon)}$. A similar argument holds for an $\epsilon$-approximation to LP~(\ref{eq:packing:general}).}  We will now explain how to use the whack-a-mole MWU algorithm to obtain  $\epsilon$-approximate optimal solutions to this pair of LPs in static, dynamic, streaming and online settings. 

Specifically, we describe how  an algorithm for Problem~\ref{prob:basic} can be used in a black-box manner to solve general packing-covering LPs. Such a reduction works by  first guessing the value of {\sc Opt}, and then showing an equivalence between: (a) finding a solution with objective $= \text{{\sc Opt}}$ and (b)  solving Problem~\ref{prob:basic}. The equivalence follows from standard scaling techniques.

\subsubsection{Static Setting}
\label{sec:reductions:static}

In this section, we show how to prove the following theorem.

\begin{theorem}
\label{th:main:static}
There is a deterministic $\epsilon$-approximation algorithm for solving LP~(\ref{eq:covering:general}) that runs in $O\left(N \cdot \frac{\log (n)}{\epsilon^3}  \cdot \log^2 \left( \frac{nU \log (n)}{\epsilon L}\right) \cdot \log (nU/L)\right)$ time. Here, $N$ denotes  the number of non-zero entries in $C$, whereas $L$ (resp.~$U$) respectively denotes a lower (resp.~upper) bound on the minimum (resp.~maximum) value of any non-zero entry in $C, a, b$.  The same guarantee also holds for solving  LP~(\ref{eq:packing:general}). 
\end{theorem}

Define a matrix $C' \in \mathbb{R}_{\geq 0}^{m \times n}$ such that $C'_{ij} := C_{ij}/(a_j b_i)$ for all $i \in [m], j \in [n]$. It is easy to check that LP~(\ref{eq:covering:general}) and LP~(\ref{eq:packing:general}) can be equivalently written as:
\begin{eqnarray}
\label{eq:covering:general:1}
\text{Minimise } \mathbb{1}^{\top}  x  & \ \ & \text{ s.t. } C' x \geq \mathbb{1} \text{ and } x \in \mathbb{R}_{\geq 0}^n. \\
\label{eq:packing:general:1}
\text{Maximise } \mathbb{1}^{\top}  y  & \ \ & \text{ s.t. } (C')^{\top} y \leq \mathbb{1} \text{ and } y \in \mathbb{R}_{\geq 0}^m. 
\end{eqnarray}
We can compute the matrix $C'$ in $O(N)$ time. Furthermore, the maximum (resp.~minimum) value of any non-zero entry of $C'$ is upper (resp.~lower) bounded by $U/L^2$ (resp.~$L/U^2$).  Let {\sc Opt} denote the optimal objective value of LP~(\ref{eq:covering:general:1}) and LP~(\ref{eq:packing:general:1}). Note that  $L^2/U \leq \text{{\sc Opt}} \leq nU^2/L$. We discretize this range $[L^2/U, nU^2/L]$ in powers of $(1+\epsilon)$, to obtain $\log_{(1+\epsilon)} (n U^3/L^3) = \Theta(\log (n U/L)/\epsilon)$ many {\em guesses} for {\sc Opt}. For each such guess $\mu$, we consider the following problem.
\begin{problem}
\label{eq:covering:general:2}
Either return an  $x \in \mathbb{R}_{\geq 0}^n$ such that  $\mathbb{1}^{\top}  x \leq (1+\Theta(\epsilon)) \cdot \mu$ and  $C' x \geq (1-\Theta(\epsilon)) \cdot \mathbb{1}$, or return a $y \in \mathbb{R}_{\geq 0}^m$ such that $\mathbb{1}^{\top}  y \geq (1-\Theta(\epsilon)) \cdot \mu$ and  $(C')^{\top} y \leq (1+\Theta(\epsilon)) \cdot \mathbb{1}$.
\end{problem}
Next, define a matrix $C'' \in \mathbb{R}_{\geq 0}^{m \times n}$ where $C''_{ij} := \mu \cdot C_{ij}$ for all $i \in [m], j \in [n]$, and observe that Problem~\ref{eq:covering:general:2}  is equivalent to the following problem.
\begin{problem}
\label{eq:covering:general:3}
Either return an  $x \in \mathbb{R}_{\geq 0}^n$ such that  $\mathbb{1}^{\top}  x \leq 1+\Theta(\epsilon)$ and  $C'' x \geq (1-\Theta(\epsilon)) \cdot \mathbb{1}$, or return a $y \in \mathbb{R}_{\geq 0}^m$ such that $\mathbb{1}^{\top}  y \geq 1-\Theta(\epsilon)$ and  $(C'')^{\top} y \leq (1+\Theta(\epsilon)) \cdot \mathbb{1}$.
\end{problem}
Note that $C'' \in [0, \lambda]^{m\times n}$ where $\lambda = \mu \cdot (U/L^2) \leq n U^3/L^3$. Accordingly, we can solve Problem~\ref{eq:covering:general:3} by using our whack-a-mole MWU algorithm from Section~\ref{sec:dual_MWU}. According to Theorem~\ref{th:MWU:static:runtime}, this takes $O\left(N \cdot \frac{\log (n)}{\epsilon^2}  \cdot \log^2 \left( \frac{nU^3 \log (n)}{\epsilon L^3}\right) \right) = O\left(N \cdot \frac{\log (n)}{\epsilon^2}  \cdot \log^2 \left( \frac{nU \log (n)}{\epsilon L}\right) \right)$ time.

It is easy to check that we can recover $\epsilon$-approximate optimal solutions to LP~(\ref{eq:covering:general}) and LP~(\ref{eq:packing:general}) if we solve Problem~\ref{eq:covering:general:3} for each of the $\Theta(\log (n U/L)/\epsilon)$ many guesses $\mu$. This leads us to Theorem~\ref{th:main:static}.

\subsubsection{Dynamic Setting}
\label{sec:reductions:dynamic}

In this section, we show how to prove the following theorem.

\begin{theorem}
\label{th:main:dynamic:restricting}
Consider any sequence of $\tau$ restricting updates to LP~(\ref{eq:covering:general}), where each update either decreases an entry of  $C$, or increases an entry of $a, b$. Let $N$ denote the total number of non-zero entries in $C$ at preprocessing. Let $L$ (resp.~$U$) denote a lower (resp.~upper) bound on the minimum (resp.~maximum) value of a non-zero entry of $C, a, b$. We can deterministically maintain an $\epsilon$-approximate optimal solution to LP~(\ref{eq:covering:general}) in $O\left(\tau \cdot \frac{\log (nU/L)}{\epsilon} +  N \cdot \frac{\log (n)}{\epsilon^3} \cdot \log^3 \left( \frac{n U \log (n)}{\epsilon L} \right) \cdot \log (nU/L)\right)$ total  time. The same guarantee also holds for maintaining an $\epsilon$-approximate solution to LP~(\ref{eq:packing:general}), which undergoes relaxing updates.
\end{theorem}

As in Section~\ref{sec:reductions:static}, we make $\Theta(\log (n U/L)/\epsilon)$ many {\em guesses} for {\sc Opt}. For each such guess $\mu$, we maintain a solution to Problem~\ref{eq:covering:general:2}, which, in turn, is equivalent to Problem~\ref{eq:covering:general:3}. This allows us to maintain $\epsilon$-approximate optimal solutions to LP~(\ref{eq:covering:general}) and LP~(\ref{eq:packing:general}) in the dynamic setting.

Fix any guess $\mu$, and focus on the corresponding instance of Problem~\ref{eq:covering:general:3}.  Since we are satisfied with an $\epsilon$-approximate optimal solution, we can safely assume that any update with a {\em meaningful} impact on the solution changes the corresponding entry by at least a multiplicative factor of $(1+\epsilon)$. Henceforth, we only focus on handling these meaningful updates. Note that each entry of $C, a, b$ participates in at most $\log_{(1+\epsilon)} (U/L) = O\left(\frac{\log (U/L)}{\epsilon} \right)$ many meaningful updates. Furthermore, each time there is a meaningful restricting update to some $b_i$, it leads to at most $O(N_i)$ many restricting entry updates to the matrix $C''$, where $N_i$ is the number of non-zero entries in the $i^{th}$ row of $C$ at preprocessing. Similarly, whenever there is a meaningful restricting update to some $a_j$, it leads to at most $O(N^j)$ many restricting entry updates to the matrix $C''$, where $N^j$ is the number of non-zero entries in the $j^{th}$ column of $C$ at preprocessing. 

Consider any sequence of $\tau$ restricting updates to LP~(\ref{eq:covering:general}). By the above discussion, this leads to at most $\tau + \left( \sum_{j \in [n]} N^j  + \sum_{i \in [m]} N_i \right) \cdot O\left(\frac{\log (U/L)}{\epsilon}\right)  = \tau + O\left(N \cdot \frac{\log (U/L)}{\epsilon} \right)$ many restricting entry updates in  Problem~\ref{eq:covering:general:3}, which we solve using the  algorithm from Section~\ref{sec:dynamic_MWU}. As $\lambda \leq nU^3/L^3$,   Theorem~\ref{th:MWU:static:runtime:dyn} implies a total update time of:
$$O\left(\tau + N \cdot \frac{\log (U/L)}{\epsilon} + N \cdot \frac{\log (n)}{\epsilon^2} \cdot \log^3 \left( \frac{\lambda \log (n)}{\epsilon} \right)\right) = O\left(\tau +  N \cdot \frac{\log (n)}{\epsilon^2} \cdot \log^3 \left( \frac{n U \log (n)}{\epsilon L} \right)\right).$$

Theorem~\ref{th:main:dynamic:restricting} follows since we run $\Theta\left(\frac{\log (n U/L)}{\epsilon}\right)$  copies of the dynamic algorithm, one for each $\mu$.

\subsubsection{Streaming Setting}
\label{sec:reductions:streaming}

In this section, we consider the streaming setting where the rows of the constraint matrix of a covering LP (or equivalently, the columns of the constraint matrix of the dual packing LP) are arriving one after another. We show how to prove the following two theorems.

\begin{theorem}
\label{th:main:streaming:packing}
Consider a streaming setting where the columns of the constraint matrix of LP~(\ref{eq:packing:general}) arrive one after another. There is a deterministic  algorithm with space complexity $O(m+n)$ that makes $O\left( 
\frac{\log n \cdot \log (nU/L)}{\epsilon^3} \right)$ passes through the stream, and returns an $\epsilon$-approximate optimal solution to LP~(\ref{eq:packing:general}). Here, $L$ (resp.~$U$) denotes a lower (resp.~upper) bound on the value of any non-zero entry of $C, a, b$.
\end{theorem}

\begin{theorem}
\label{th:main:streaming:covering}
Consider a streaming setting where the rows of the constraint matrix of LP~(\ref{eq:covering:general}) arrive one after another. There is a deterministic algorithm with space complexity $O(n)$ that makes $O\left( 
\frac{\log n \cdot \log (nU/L)}{\epsilon^3} \right)$ passes through the stream, and returns an $\epsilon$-approximate optimal solution to LP~(\ref{eq:covering:general}). Here, $L$ (resp.~$U$) denotes a lower (resp.~upper) bound on the value of any non-zero entry of $C, a, b$.
\end{theorem}


As in Section~\ref{sec:reductions:static}, we make $\Theta(\log (n U/L)/\epsilon)$ many {\em guesses} for {\sc Opt}. For each such guess $\mu$, we solve Problem~\ref{eq:covering:general:2} in the streaming setting, which, in turn, is equivalent to Problem~\ref{eq:covering:general:3}. This allows us to return $\epsilon$-approximate optimal solutions to LP~(\ref{eq:covering:general}) and LP~(\ref{eq:packing:general}).

Fix any guess $\mu$, and focus on the corresponding instance of Problem~\ref{eq:covering:general:3}. Suppose that we are allowed to have a space complexity of $O(m+n)$. By Theorem~\ref{th:streaming:main}, we can solve Problem~\ref{eq:covering:general:3} by making $O(\frac{\log n}{\epsilon^2})$ many passes through the stream. Theorem~\ref{th:main:streaming:packing} now follows since we need to solve Problem~\ref{eq:covering:general:3} for each guess $\mu$, which increases the number of passes by a multiplicative factor of $\Theta(\log (n U/L)/\epsilon)$.

Finally, Theorem~\ref{th:main:streaming:covering} follows if we apply the same strategy described above and use Corollary~\ref{cor:streaming:main}.


\subsubsection{Online Setting}
\label{sec:reductions:online}

In this section, we explain how to prove the following theorem.

\begin{theorem}
\label{th:main:online}
Consider an online setting where the constraints of  LP~(\ref{eq:covering:general}) arrive one after another. There is a deterministic  algorithm which maintains an $\epsilon$-approximate optimal solution to this LP with total recourse $O\left( n \cdot \frac{\log n}{\epsilon^3} \cdot \log \left(\frac{nU}{L}\right) \right)$. Here, $L$ (resp.~$U$) is a lower (resp.~upper) bound on the value of any non-zero entry of $C, a, b$.
\end{theorem}

As in Section~\ref{sec:reductions:static}, we make $\Theta(\log (n U/L)/\epsilon)$ many {\em guesses} for {\sc Opt}. For each such guess $\mu$, we solve Problem~\ref{eq:covering:general:2} in the online setting, which, in turn, is equivalent to Problem~\ref{eq:covering:general:3}. This allows us to maintain an $\epsilon$-approximate optimal solutions to LP~(\ref{eq:covering:general}). Theorem~\ref{th:main:online} now follows from Theorem~\ref{th:online:main} as the total recourse increases by a multiplicative factor of $\Theta(\log (nU/L)/\epsilon)$.

\subsection{Challenges for Dynamic Whack-a-Mole MWU Algorithm for Packing LPs}

We have so far focussed on the whack-a-mole MWU algorithm for covering LPs. For packing LPs, there exists an analogous whack-a-mole MWU algorithm.  To see this, consider the following problem.

\begin{problem}
\label{prob:basic:packing}
Given a matrix $P \in [0, \lambda]^{m \times n}$ where $\lambda > 0$, either return an $x \in \mathbb{R}_{\geq 0}^n$ with $\mathbb{1}^{\top} x \geq 1-\Theta(\epsilon)$ and $P x \leq (1+\Theta(\epsilon)) \cdot \mathbb{1}$, or return a $y \in \mathbb{R}_{\geq 0}^m$ with $\mathbb{1}^{\top} y \leq 1+\Theta(\epsilon)$ and $P^{\top} y \geq (1-\Theta(\epsilon)) \cdot \mathbb{1}$. 
\end{problem}

The basic template for a whack-a-mole algorithm for Problem~\ref{prob:basic:packing} is described in Figure~\ref{fig:dual_algorithm:packing} and Figure~\ref{fig:enforce:packing}. Note that this is completely analogous to the template in Section~\ref{sec:basic:template}.

 \begin{figure*}[htbp]
                                                \centerline{
                                                \framebox{
                                                {
                                                                \begin{minipage}{5.5in}
                                                                        \begin{tabbing}                                                                            
                                                                                01. \ \ \=  Define $T \leftarrow \frac{\lambda \ln (n)}{\epsilon^2}$, and two vectors $\hat{x}^1, x^1 \in \mathbb{R}_{\geq 0}^n$ where $\hat{x}^1 \leftarrow \mathbb{1}$ and $x^1 \leftarrow \frac{\hat{x}^1}{\norm{\hat{x}^1}_1}$. \\
                                                                                02. \> \sc{For} $t = 1$ to $T$:\\
                                                                                03. \> \ \ \ \ \ \ \ \ \ \= \sc{Either}\\
                                                                                04. \> \> \ \ \ \ \ \ \ \= Conclude that  $\left(P \cdot x^t\right)_{i} \leq 1 + \epsilon$ for all $i \in [m]$. \\
                                                                                05. \> \> \> Terminate the {\sc For} loop, and {\sc Return} $(x^t, \text{{\sc Null}})$.\\
                                                                                06. \> \>  \sc{Or} \\
                                                                                07. \> \> \> Find a packing constraint $i_t \in [m]$ such that  $\left(P \cdot x^t\right)_{i_t} > 1$. \\
                                                                                08. \> \> \> $\hat{x}^{t+1} \leftarrow \text{\sc{Whack}}(i_t,\hat{x}^t)$.  \qquad // See Figure~\ref{fig:enforce:packing}. \\                                                                 
                                                                                09. \> \> \> $x^{t+1} \leftarrow \frac{\hat{x}^{t+1}}{\norm{\hat{x}^{t+1}}_1}$.\\
                                                                                10. \> \> \> Let $y^t \in \Delta^m$ be the vector where $\left(y^t\right)_{i_t} = 1$ and $\left( y^t \right)_i = 0$ for all $i \in [m] \setminus \{ i_t\}$. \\
                                                                                11. \> $y \leftarrow (1/T) \cdot \sum_{t = 1}^{T} y^t$. \\
                                                                                12. \> \sc{Return} $(\text{{\sc Null}}, y)$.
                                                                                \end{tabbing}
                                                                \end{minipage}
                                                        }
                                                }
                                                }
                                                        \caption{\label{fig:dual_algorithm:packing} The Whack-a-Mole MWU Algorithm for a Packing LP.}
                                                \end{figure*}
                                                
  \begin{figure*}[htbp]
                                                \centerline{
                                                \framebox{
                                                {
                                                                \begin{minipage}{5.5in}
                                                                        \begin{tabbing}                                                                            
                                                                                1.  \= {\sc For all} $j \in [n]$: \\
                                                                                2. \> \ \ \ \ \ \ \= $\hat{z}_j \leftarrow \left(1-\epsilon \cdot \frac{P_{ij}}{\lambda} \right) \cdot \hat{x}_j$. \\
                                                                                3. \> \sc{Return} $\hat{z}$ 
                                                                                \end{tabbing}
                                                                \end{minipage}
                                                        }
                                                }
                                                }
                                                        \caption{\label{fig:enforce:packing} {\sc Whack}$(i, \hat{x})$.}
                                                \end{figure*}

The key point to note is that unlike in Section~\ref{sec:basic:template}, here the weight $\left(\hat{x}^t\right)_j$ of an expert $j \in [n]$ (see the discussion in the beginning of Section~\ref{sec:lm:whack:a:mole:key}) can only {\em decrease} whenever we whack a constraint. Nevertheless, we can easily extend the analysis from Section~\ref{sec:basic:template} to obtain the following theorem.

\begin{theorem}
\label{th:basic:template:packing}
The algorithm in Figure~\ref{fig:dual_algorithm:packing} either returns a vector $x^t \in \mathbb{R}_{\geq 0}^n$ with $\mathbb{1}^{\top} x^t = 1$ and $P x^t \leq (1+\Theta(\epsilon)) \cdot \mathbb{1}$, or it returns a vector $y \in \mathbb{R}_{\geq 0}^m$ with $\mathbb{1}^{\top} y = 1$ and $P^{\top} y \geq (1-\Theta(\epsilon)) \cdot \mathbb{1}$.
\end{theorem}

From this basic template, it is straightforward to  obtain a near-linear time static algorithm for Problem~\ref{prob:basic:packing}. The framework also seamlessly extends to give us whack-a-mole MWU algorithms for packing LPs in the streaming and online settings. Thus, every theorem derived in Sections~\ref{sec:reductions:static},~\ref{sec:reductions:streaming} and~\ref{sec:reductions:online}  holds even if we switch the occurrences of LP~(\ref{eq:covering:general}) with that of LP~(\ref{eq:packing:general}) in the concerned theorem statement. The only exception is the theorem derived in Section~\ref{sec:reductions:dynamic}. Below, we highlight the main challenge in obtaining a dynamic whack-a-mole MWU algorithm for packing LPs under restricting updates.

Recall that a key part of the analysis in Section~\ref{sec:dynamic_MWU} was deriving an upper bound on the number of times a given constraint can get enforced throughout the duration of the algorithm (see Lemma~\ref{lm:key:dynamic}). The proof of Lemma~\ref{lm:key:dynamic}, on the other hand, relied on the property that the weight $\left(\hat{x}^t\right)_j$ of every expert $j \in [n]$ is monotonically non-decreasing with time and always lies in the range $\left[1, n^{(1/\epsilon)}\right]$, and hence is polynomially bounded (see the discussion leading to~(\ref{eq:key:inference:2})). This property followed from three observations: (1) at the start of the algorithm we have $\left( \hat{x}^1 \right)_j = 1$, (2) the total weight $\sum_{j '\in [n]} \left( \hat{x}^t \right)_j$ of all the experts is at most $n^{(1/\epsilon)}$, and (3) the weight $\left( \hat{x}^t \right)_j$ of a given expert $j \in [n]$ is at most the total weight $\sum_{j '\in [n]} \left( \hat{x}^t \right)_j$.

Coming back to our current setting, recall that the weight $\left( \hat{x}^t \right)_j$ of an expert $j \in [n]$ can only decrease over time when we run the whack-a-mole MWU algorithm on a packing LP. Thus, in order to bound the maximum number of times a constraint can get enforced, we need to show that the weight $\left(\hat{x}^t\right)_j$ of every expert $j \in [n]$   lies in the range $\left[\frac{1}{n^{\text{poly}(1/\epsilon)}}, 1 \right]$ throughout the duration of the algorithm. A natural way to replicate the argument from the previous paragraph to our current setting would be to make the following sequence of claims: (a) at the start of the algorithm we have $\left( \hat{x}^1 \right)_j = 1$, (b) the total weight $\sum_{j '\in [n]} \left( \hat{x}^t \right)_j$ of all the experts is at least $n^{(1/\epsilon)}$, and (c) the weight $\left( \hat{x}^t \right)_j$ of a given expert $j \in [n]$ is {\em at least} the total weight $\sum_{j '\in [n]} \left( \hat{x}^t \right)_j$. Now, the crucial observation is that although we can extend the analysis of the whack-a-mole MWU algorithm for covering LP to prove claims (a) and (b),  claim (c) clearly does {\em not} hold for obvious reasons. Indeed,  if we consider a natural extension of the static whack-a-mole MWU algorithm for packing LPs to a setting with restricting updates, then we can {\em not} ensure that the weight $\left( \hat{x}^t \right)_j$ of an expert $j \in [n]$ remains polynomially bounded. This, in turn, implies that we cannot upper bound the maximum number of times a constraint gets enforced, and hence we cannot bound the total update time of the algorithm.

We address this issue in Section~\ref{sec:prob}, where we present an efficient dynamic algorithm for maintaining an $\epsilon$-approximate solution to a {\em positive LP} under relaxing updates. Section~\ref{sec:exception:packing} shows that as a simple corollary of this result, we obtain an efficient dynamic algorithm for maintaining an $\epsilon$-approximate optimal solution to a packing (resp.~covering) LP under restricting (resp.~relaxing) updates (see Theorem~\ref{th:main:dynamic:restricting:packing}).

%% file: greedy-mwu.tex
\section{The Greedy MWU Algorithm}
\label{sec:prob}

In this section we focus on
\emph{positive} LPs, which are defined as follows:
\begin{equation}
\text{Find }x\in\mathbb{R}_{\geq0}^{n}\text{ such that }Px\leq a\text{ and }Cx\geq b,\text{ where }P\in\mathbb{R}_{\geq0}^{m_{p}\times n},a\in\mathbb{R}_{\geq0}^{m_{p}},C\in\mathbb{R}_{\geq0}^{m_{c}\times n},b\in\mathbb{R}_{\geq0}^{m_{c}}.\label{lp:positive}
\end{equation}
An $\epsilon$-approximate solution for the positive LP is
either an $x\in\mathbb{R}_{\geq0}^{n}$ satisfying $Px\leq(1+\epsilon)a\text{ and }Cx\geq(1-\epsilon)b$, 
or a symbol $\bot$ indicating that the LP is infeasible. In the dynamic setting, an update to the LP can change an entry of
$P,a,C$ or $b$. Observe that the update is relaxing if it increases
an entry of $C$ or $a$, or decreases an entry of $P$ or $b$. In this section, we wish to design a dynamic algorithm for maintaining an $\epsilon$-approximation solution to a positive LP undergoing relaxing updates.  Note that such an algorithm will work as follows. Initially,  it will return $\bot$ for a sequence of relaxing updates. After that, at a certain point in time it will return an  $x\in\mathbb{R}_{\geq0}^{n}$ satisfying $Px\leq(1+\epsilon)a\text{ and }Cx\geq(1-\epsilon)b$. From this point onward,  the same $x$ will continue to remain an $\epsilon$-approximate solution to the input positive LP after every future update.

\medskip
\noindent {\bf Notation and preliminaries.} We now introduce a few key notations and concepts that will be used throughout the rest of Section~\ref{sec:prob}.\footnote{We emphasize that some of these notations are different from the ones used in Section~\ref{sec:whack:a:mole}.}  We classify a relaxing update as an {\em entry update} if it changes an entry of $P$ or $C$, whereas we classify a relaxing update as a {\em translation update} if it changes an entry of $a$ or $b$.   Throughout the sequence of updates, let $L$ (resp.~$U$) respectively denote an upper (resp,~lower) bound on the value of any nonzero entry in $P, C, a, b$; and let $N$ denote an upper bound on the total number of nonzero entries in $P, C$. We assume that $U/L = O(\text{poly}(m_p+m_c+n))$, where $m_p, m_c$ and $n$ respectively denote the number of rows in $P$, the number of rows in $C$, and the number of variables. We refer to $P x \leq a$ as the {\em packing} constraints, and we refer to $C x \geq b$ as the {\em covering} constraints. For all $i \in [m_p]$ and $k \in [n]$, we let $P(i, k) \in \mathbb{R}_{\geq 0}$ denote the entry in the $i^{th}$ row and $k^{th}$ column of the matrix $P$. Furthermore, we let $P_i \in \mathbb{R}_{\geq 0}^{1 \times n}$ denote the $i^{th}$ row of the matrix $P$. We analogously define the notations $C(j, k) \in \mathbb{R}_{\geq 0}$ and $C_j \in \mathbb{R}_{\geq 0}^{1 \times n}$, for all $j \in [m_c]$ and $k \in [n]$.  Finally, we use the symbol $v_i$ to denote the $i^{th}$ co-ordinate of a vector $v = (v_1, \ldots, v_m) \in \mathbb{R}^m$.

\medskip
Our main result is summarised  below. In this section, we present a high level overview of the main ideas behind the proof of Theorem~\ref{th:relaxing:positiveLP:algo}. The full details are deferred to Appendix~\ref{sec:app:proof} and Appendix~\ref{sec:app:algo}.

\begin{theorem}
\label{th:relaxing:positiveLP:algo}
We can maintain an $\epsilon$-approximate solution to a positive LP undergoing a sequence of $t > 0$ relaxing updates in $O\left(N \cdot \frac{\log^2(m_c+ m_p + U/L)}{\epsilon^2} + t\right)$ total  update time.
\end{theorem}

\noindent
{\bf A simplifying assumption:} To convey the main ideas behind our algorithm, in this section we focus on the setting where all the relaxing updates made to the input LP are entry updates, for the following reason.

Consider an input positive LP $(P, C, a, b)$, which asks us to find an $x$ such that $P x \leq a$ and $C x \geq b$. Since we are happy with a $\red{\Theta(\epsilon)}$-approximate solution to the input LP, we can keep ignoring the relaxing translation updates to a packing  constraint $(Px)_i \leq a_i$ (resp.~covering constraint $(Cx)_j \geq b_j$) as long as $a_i$ (resp.~$b_j$) does {\em not} increase (resp.~decrease) by more than a multiplicative factor of $(1+\epsilon)$. Thus, during the course of our algorithm, a given packing or covering constraint will essentially go through  at most $\kappa = O(\log_{(1+\epsilon)} (U/L))$ many relaxing translation updates.

Next, note that we can simulate a relaxing translation update to a packing constraint $(P x)_i \leq a_i$ (resp.~covering constraint $(Cx)_j \geq b_j$) by making $nzp_{(i)}$ (resp.~$nzc_{(j)}$) many relaxing entry updates, where $nzp_{(i)}$ (resp.~$nzc_{(j)}$) denotes the current number of nonzero entries in row $i$ (resp.~row $j$) of $P$ (resp.~$C$). This is because of the following reason: If the translation update asks us to scale up (resp.~scale down) the value of $a_i$ (resp.~$b_j$) by a multiplicative factor of $\alpha \geq 1$, then we can implement this update by scaling down (resp.~scaling up) every nonzero entry in row $i$ (resp.~row $j$) of $P$ (resp.~$C$) by the same factor $\alpha$.

It follows that  all the relaxing translation updates encountered by an algorithm for this problem can be simulated by  at most $\kappa N$ many relaxing entry updates. Hence, if we prove Theorem~\ref{th:relaxing:positiveLP:algo} for relaxing entry updates, then we can immediately convert it into an algorithm that handles {\em both} translation and entry updates with total update time $O\left(N \cdot  \frac{\log^2(m_c+ m_p + U/L)}{\epsilon^2} + t + \kappa N\right) = O\left( N \cdot \frac{\log^2(m_c+ m_p + U/L)}{\epsilon^2} + t\right)$.

Accordingly, for the rest of Section~\ref{sec:prob},  the phrase {\em relaxing update} will  refer to a {\em relaxing entry update}.

\medskip
\noindent {\bf Scaling the constraints:} Since we have to deal with only the entry updates in this section, for notational convenience we will scale the right hand sides of all the packing and covering constraints to $1$. Accordingly, throughout  the rest of Section~\ref{sec:prob} the input positive LP will be given by an ordered pair $(P, C)$, where $P \in \mathbb{R}_{\geq 0}^{m_p \times n}$ and $C \in \mathbb{R}_{\geq 0}^{m_c \times n}$ and it asks us to find an $x \in \mathbb{R}_{\geq 0}^n$ such that $P x \leq 1$ and $C x \geq 1$. 

\medskip
\noindent {\bf Organisation:} In Section~\ref{sec:static}, we recap a known {\em static} greedy MWU algorithm for  positive LPs~\cite{Quanrud20,Young2014nearly}. We present and analyse our dynamic algorithm in Section~\ref{sec:dynamic}. Finally, we show that  the analysis in Section~\ref{sec:dynamic} implies a dynamic algorithm for  solving a packing (resp.~covering) LP under restricting (resp.~relaxing) updates. This observation is summarized in Section~\ref{sec:exception:packing}.


\subsection{Static Greedy MWU Algorithm~\cite{Quanrud20,Young2014nearly}}
\label{sec:static}

Define two  functions $f_p(x)$ and $f_c(x)$, which respectively correspond to the {\em soft-max} of the packing constraints and the {\em soft-min} of the covering constraints, as stated below.
\begin{eqnarray}
f_p(x) = \frac{1}{\eta} \cdot \log \left( \sum_{i=1}^{m_p} \exp(\eta P_i x) \right) \text{ and } f_c(x) = -\frac{1}{\eta} \cdot \log \left(  \sum_{j=1}^{m_c} \exp(-\eta C_j x) \right), \nonumber \\ 
\text{ where } \eta = \frac{\log (m_p+m_c + U/L)}{\epsilon}. \label{eq:eta}
\end{eqnarray}
It can be shown that the function $f_p(x)$ (resp.~$f_c(x)$)  closely approximates the maximum (resp.~minimum) value among the left hand sides of the packing  (resp.~covering) constraints. Specifically, we have:
\begin{eqnarray}
\label{eq:soft:approx}
\max_{i\in [m_p]} P_i x \leq f_p(x) \leq \max_{i \in [m_p]} P_i x + \epsilon \text{ and } \min_{j\in [m_c]} C_j x -\epsilon \leq f_c(x) \leq \min_{j \in [m_c]} C_j x, \text{ for all } x \in \mathbb{R}_{\geq 0}^n.
\end{eqnarray}
The  algorithm attempts to find an $x \in \mathbb{R}_{\geq 0}^n$ where $f_p(x) \leq 1$ and $f_c(x) \geq 1$. By~(\ref{eq:soft:approx}),  such an $x$ is indeed  an $\epsilon$-approximate solution to the input positive LP. Note that both the functions $f_p(x)$ and $f_c(x)$ are continuous and differentiable, and hence their gradients exist at all points. 

\subsubsection{The main idea}
\label{subsec:template}
Say that a direction $z \in \mathbb{R}_{\geq 0}^n$ is {\em cheap} with respect to a point $x \in \mathbb{R}_{\geq 0}^n$ iff $\langle \nabla f_p(x), z \rangle \leq \langle \nabla f_c(x), z \rangle$. In other words, starting from the point $x$, if we take an infinitesimally small step towards a cheap direction w.r.t.~$x$, then the increase in $f_p(x)$ is at most the increase in $f_c(x)$. The high level idea  behind the algorithm can be summarised as follows: It starts at an $x = \vec{0} \in \mathbb{R}_{\geq 0}^n$, where we have $f_p(x) = f_c(x) = 0$. It then continuously keeps {\em moving} this point $x$ in a cheap direction. This leads to the following invariant: $f_p(x) \leq f_c(x)$. The algorithm stops once it reaches a point $x$ where $f_c(x) = 1$. Such an $x$ satisfies  $f_p(x) \leq 1$ and $f_c(x) \geq 1$, and hence is an $\epsilon$-approximate solution to the input positive LP, according to~(\ref{eq:soft:approx}).  On the other hand, if at some point in time during this continuous process, we reach a point $x$ that does {\em not} admit any cheap direction, then it can be shown  that the input LP is infeasible.

In order to actually implement the scheme described above, we first need to {\em discretise} the process  which  keeps moving the point $x$. This discretised process  will consist of a sequence of steps. As before, initially we have $x = \vec{0}$. At the start of each step, we identify an (approximately) cheap direction $z \in \mathbb{R}_{\geq 0}^n$ w.r.t.~the current $x$, which satisfies  $\langle \nabla f_p(x), z \rangle \leq (1+\red{\Theta(\epsilon)}) \cdot \langle \nabla f_c(x), z \rangle$, and then we take a small discrete {\em jump} along that direction by setting $x \leftarrow x + \delta  z$ for some sufficiently small $\delta > 0$. It can be shown that due to each of these steps, the increase in the value of $f_p(x)$ is at most $(1+\red{\Theta(\epsilon)})$ times the increase in the value of $f_c(x)$, which leads to the invariant: $f_p(x) \leq (1+\red{\Theta(\epsilon)}) \cdot f_c(x)$. Furthermore, during each of these  steps, the value of $f_c(x)$ increases by at most $\red{\Theta(\epsilon)}$, for sufficiently small $\delta$. The process terminates {\em immediately} after we reach a stage where $f_c(x) \geq 1$. Thus, at termination we have $1 \leq f_c(x) \leq (1+\red{\Theta(\epsilon)})$ and $f_p(x) \leq (1+\red{\Theta(\epsilon)}) f_c(x) \leq (1+\red{\Theta(\epsilon)})$, and hence $x$ is a $\red{\Theta(\epsilon)}$-approximate solution to the input positive LP.  In contrast, if at any  step during this process, we end up with an $x$ which does not admit any (approximately) cheap direction, then we can certify that the input LP is infeasible. 

We will work with a concrete version of the  framework described above. Say that a {\em coordinate} $k \in [n]$ is  (approximately) {\em cheap} with respect to a point $x \in \mathbb{R}_{\geq 0}^n$ iff $\langle \nabla f_p(x), \vec{e}_k \rangle \leq (1+\red{\Theta(\epsilon)}) \cdot \langle \nabla f_c(x), \vec{e}_k \rangle$, where $e_k \in \{0,1\}^n$ is the unit vector along coordinate $k$. We will use the following crucial observation: Any given point $x$ admits a cheap direction iff it admits a cheap coordinate. Thus, all we need to do is start at $x = \vec{0}$, and then move the point $x$ in a sequence of steps, where at each step we move $x$ by a small amount along a cheap coordinate. We now present this algorithmic template  more formally in Section~\ref{subsec:static}. 

\subsubsection{The basic algorithmic template}
\label{subsec:static}

We start by introducing some crucial notations.  Given any $x \in \mathbb{R}_{\geq 0}^n$, we  associate a {\em weight} with each LP constraint. To be specific, for all $i \in [m_p]$ and $j \in [m_c]$, we have $w_p(x, i) := \exp\left(\eta \cdot P_i x \right)$ and $w_c(x, j) := \exp\left(- \eta \cdot C_j x \right)$. Let $w_p(x) := \sum_{i=1}^{m_p} w_p(x, i)$ and $w_c(x) := \sum_{j=1}^{m_c} w_c(x, j)$.  

\begin{definition}
\label{def:coordinate}
Consider any $k \in [n]$. The {\em cost} of coordinate $k$ at a given $x \in \mathbb{R}^n_{\geq 0}$ is defined as:
$$\lambda(x, k) := \frac{\sum_{i = 1}^{m_p} w_p(x, i) \cdot P(i, k)}{\sum_{j = 1}^{m_c} w_c(x, j) \cdot C(j, k)}.$$
The coordinate $k \in [n]$ is said to be {\em cheap} w.r.t.~a given $x \in \mathbb{R}_{\geq 0}^n$ iff $\lambda(x, k) \leq (1+\red{\Theta(\epsilon)}) \cdot w_p(x)/w_c(x)$. 
\end{definition}

It turns out that the notion of a cheap coordinate, as per Definition~\ref{def:coordinate}, is identical to the one used in the last paragraph of Section~\ref{subsec:template} (see Property~\ref{prop:infeasible} and Property~\ref{prop:stop}). The algorithm is described in Figure~\ref{fig:static}.

  \begin{figure*}[htbp]
                                                \centerline{\framebox{
                                                                \begin{minipage}{5.5in}
                                                                        \begin{tabbing}                                                                            
                                                                                1. \ \ \=  {\sc Initialise} $x \leftarrow \vec{0} \in \mathbb{R}_{\geq 0}^n$. \\ 
                                                                                2. \> {\sc While} $\min_{j \in [m_c]} C_j x < 1$: \\
                                                                                3. \> \ \ \ \  \ \ \ \ \= {\sc If} there is no cheap coordinate at $x$, {\sc Then} \\
                                                                                4. \> \> \ \ \ \ \ \ \ \ \ \= {\sc Return} that the LP is infeasible.   \\
                                                                                5. \> \> {\sc Else} \\
                                                                                6. \> \> \> Let $k \in [n]$ be a cheap coordinate.  \\
                                                                                7. \> \> \> $x \leftarrow \text{{\sc Boost}}(x, k)$.   \qquad \qquad \qquad \qquad //  See Figure~\ref{fig:boost}.  \\
                                                                                8. \> {\sc Return} $x$.                             										                                                                                                         
                                                                                \end{tabbing}
                                                                \end{minipage}
                                                        }}
                                                        \caption{\label{fig:static} A static algorithm for solving a positive LP $(P,C)$.}
                                                \end{figure*}

  \begin{figure*}[htbp]
                                                \centerline{\framebox{
                                                                \begin{minipage}{5.5in}
                                                                        \begin{tabbing}                                                                            
                                                                                1. \ \ \=  Find the maximum $\delta$ such that $\max_{i \in [m_p]} P(i, k) \cdot \delta \leq \epsilon/\eta$ 	and 	$\max_{j \in [m_c] : C_j x < 2} C(j, k) \cdot \delta \leq \epsilon/\eta$.	\\		                                                                        2. \> {\sc Return} $x + \delta \vec{e}_k$.                                  
                                                                                \end{tabbing}
                                                                \end{minipage}
                                                        }}
                                                        \caption{\label{fig:boost} {\sc Boost}($x, k$).}
                                                \end{figure*}

It is easy to check that this algorithm follows the  framework outlined in Section~\ref{subsec:template}. To summarise, the algorithm starts at $x = \vec{0}$. Subsequently, each iteration of the {\sc While} loop in Figure~\ref{fig:static} either declares that the input LP is infeasible, or identifies a cheap coordinate w.r.t.~the current $x$ and calls the subroutine in Figure~\ref{fig:boost} to {\em boost} (i.e., change) $x$ by a small amount along that direction. It stops when $\min_{j \in [m_c]} C_j x \geq 1$.
  
 We now state two key properties of the algorithm, whose proofs are deferred to  Appendix~\ref{sec:app:proof}.  Property~\ref{prop:infeasible} justifies the algorithm's decision to declare that the input LP is infeasible if it cannot find any cheap coordinate w.r.t.~the current $x \in \mathbb{R}_{\geq 0}^n$. Property~\ref{prop:stop} ensures that whenever the algorithm decides to {\em boost} the current $x$ along a cheap coordinate $k \in [n]$, the increase in $f_p(x)$ is at most $(1+\red{\Theta(\epsilon)})$ times the increase in $f_c(x)$, which, in turn, is itself at most $\red{\Theta(\epsilon)}$. 
      
 \begin{property}
 \label{prop:infeasible}
 Consider any $x \in \mathbb{R}_{\geq 0}^n$. If there is no cheap coordinate at $x$, then the  LP $(P, C)$ is infeasible.
 \end{property}    
 
 \begin{property}
 \label{prop:stop}
 Assume $\epsilon < 1/10$ and let $k \in [n]$ be a cheap coordinate at   $y \in \mathbb{R}_{\geq 0}^n$, and let $z \leftarrow \text{{\sc Boost}}(y, k)$. Then we have: 
$$f_p(z) - f_p(y) \leq (1+\red{\Theta(\epsilon)}) \cdot \left( f_c(z) - f_c(y)\right) \leq \red{\Theta(\epsilon)}.$$
 \end{property}  
 
Recall that initially $x = \vec{0}$, and so $f_p(x) = f_c(x) = 0$.  Each subsequent iteration of the {\sc While} loop in Figure~\ref{fig:static} boosts the current $x$ along a cheap coordinate. Hence, Property~\ref{prop:stop} implies the following invariant:  
 \begin{invariant}
 \label{inv:static}
 During the course of this algorithm, we always have $f_p(x) \leq (1+\red{\Theta(\epsilon)}) \cdot f_c(x)$. Furthermore, each call to the subroutine {\sc Boost}$(x,k)$ increases the value of $f_c(x)$ by at most an additive $\red{\Theta(\epsilon)}$. 
 \end{invariant}                     
                          
 \begin{corollary}
 \label{cor:static}
 If the algorithm in Figure~\ref{fig:static} returns an $x \in \mathbb{R}_{\geq 0}^n$ in step 8, then $x$ is indeed a $\red{\Theta(\epsilon)}$-approximate solution to the input LP. Otherwise, it correctly declares that the input LP is infeasible.
 \end{corollary}  
 
 \begin{proof} {\em Case (i): The algorithm terminates in step $8$ of Figure~\ref{fig:static}.} All the claims made in the next paragraph hold because of~(\ref{eq:soft:approx}) and Invariant~\ref{inv:static}.
  
 Just before the last iteration of the {\sc While} loop in Figure~\ref{fig:static},  we had $f_p(x) \leq (1+\red{\Theta(\epsilon)}) \cdot f_c(x)$ and $f_c(x) \leq \min_{j \in [m_c]} C_j x < 1$. Accordingly, at the end of the last iteration of the concerned {\sc While} loop, we have $\max_{i \in [m_p]} P_i x \leq f_p(x) \leq (1+\red{\Theta(\epsilon)}) \cdot f_c(x) \leq  (1+\red{\Theta(\epsilon)}) \cdot (1+\red{\Theta(\epsilon)}) \leq 1+\red{\Theta(\epsilon)}$, and $\min_{j \in [m_c]} C_j x \geq 1$.  Hence, the $x$ that is returned  in step $8$ of Figure~\ref{fig:static} is  a $\red{\Theta(\epsilon)}$-approximate solution to the input LP.
 
\medskip
\noindent 
{\em Case (ii): The algorithm terminates because during some iteration of the {\sc While} loop in Figure~\ref{fig:static} it cannot find any cheap coordinate w.r.t.~the current $x$.} In this case,  by Property~\ref{prop:infeasible},  the input LP is infeasible. 
\end{proof}

 \subsection{An Overview of Our Dynamic Greedy MWU Algorithm}     
 \label{sec:dynamic}
 
 We describe the basic template that will be followed by our dynamic algorithm in Section~\ref{subsec:dynamic:template}. For ease of exposition, in Section~\ref{subsec:dynamic:covering} we describe and analyse our dynamic algorithm in a setting where all the relaxing entry updates occur to the covering constraints, whereas the packing constraints remain unchanged. Finally, in Section~\ref{subsec:dynamic:full} we present an overview of the complete algorithm which can handle relaxing entry updates to {\em both} the packing and covering constraints.

 \subsubsection{The basic template for our dynamic algorithm}
 \label{subsec:dynamic:template}

 \noindent {\bf Initialisation:}  At preprocessing, we start by implementing the procedure described in Figure~\ref{fig:static}. At this stage, if  we end up with a $\red{\Theta(\epsilon)}$-approximate solution $x$ to the positive LP, then we are done. This is because such an $x$ will continue to remain a $\red{
 \Theta(\epsilon)}$-approximate solution to the input LP even after it undergoes any sequence of relaxing updates in future. Accordingly, henceforth we assume that the procedure in Figure~\ref{fig:static} ends up with an $x$ at which there is no cheap coordinate.
 
 \medskip
 \noindent {\bf Handling a relaxing update:} After the update,  we keep on boosting $x$ along cheap coordinates until either (a) the set of cheap coordinates becomes empty again, or (b) we have $\min_{j \in [m_c]} C_j \geq 1$. In the former case, we declare that the input LP remains infeasible even after the relaxing update. In the latter case, we declare the current $x$ as a  $\red{
 \Theta(\epsilon)}$-approximate solution to the input LP from this point onward.

 \medskip
 Clearly, as per the discussion in Section~\ref{sec:static},  any dynamic algorithm following the above template solves the problem of maintaining a  $\red{\Theta(\epsilon)}$-approximate solution to a positive LP under relaxing updates. 
  
 \subsubsection{Handling a sequence of relaxing entry updates to the covering constraints}
 \label{subsec:dynamic:covering}
 
We now present an overview of our dynamic algorithm, under the assumption that every relaxing update increases the value of some entry in the matrix $C$ (whereas the matrix $P$ remains unchanged). We follow the template outlined  in Section~\ref{subsec:dynamic:template}. The main challenge  is to define the appropriate data structures, and to describe  how to {\em choose} a cheap coordinate for boosting  $x$ at any given point in  time. 

Before proceeding any further, note that there are two types of {\em events} which influence the outcome of any dynamic algorithm following the template from Section~\ref{subsec:dynamic:template}: (1) a relaxing update which increases the value of some entry in the matrix $C$, and (2) a call to the subroutine $\text{{\sc Boost}}(x, k)$ which {\em boosts} the current $x$ along the coordinate $k$. We start with the key invariant that drives the analysis of our dynamic algorithm.

\begin{invariant}
\label{inv:dynamic}
During the course of our dynamic algorithm, we always have $f_p(x) \leq (1+\red{\Theta(\epsilon)}) \cdot f_c(x)$. Furthermore, every call to the subroutine {\sc Boost}$(x, k)$ increases  $f_c(x)$ by at most an additive $\red{\Theta(\epsilon)}$. 
\end{invariant}

\begin{proof}
The algorithm starts at $x = \vec{0}$, where $f_p(x) = f_c(x) = 0$ and so the invariant holds. Subsequently, the values of $f_p(x)$ and $f_c(x)$ can change because of one of the two following types of events.

\smallskip
\noindent (1) A relaxing update increases the value of an entry  $C(j, k)$ of the matrix $C$, for $j \in [m_c]$, $k \in [n]$. This increases the value of $C_j x$; whereas for all $j' 
\in [m_c] \setminus \{j\}$ and $i \in [m_p]$, the values of $C_{j'} x$ and $P_i x$ remain unchanged. We accordingly infer that $f_c(x)$ increases, whereas  $f_p(x)$ remains unchanged. Thus, if the invariant was true just before this event, then it continues to remain true just after the event. 

\smallskip
\noindent (2) The algorithm make a call to {\sc Boost}$(x, k)$, for some $k \in [n]$.  By Property~\ref{prop:stop}, due to this event the increase in $f_p(x)$ is at most $(1+\red{\Theta(\epsilon)})$ times the increase in $f_c(x)$, and  this latter quantity is at most $\red{\Theta(\epsilon)}$. Thus, if the invariant was true just before this event, then it continues to remain true just after the event. 
\end{proof}

\begin{corollary}
\label{cor:inv:dynamic}
During the course of our dynamic algorithm, we always have $\max_{i \in [m_p]} P_i x \leq 1+\red{\Theta(\epsilon)}$. 
\end{corollary}

\begin{proof}
Follows from~(\ref{eq:soft:approx}), Invariant~\ref{inv:dynamic}, and the fact that we return $x$ the moment  $\min_{j \in [m_c]} C_j x$ is $\geq 1$. 
\end{proof}

Next, we  derive a few important observations which show that the weights $w_p(x, i), w_c(x, j)$ of the constraints and the costs $\lambda(x, k)$ of the coordinates change (almost) monotonically over time.

 \begin{observation}
 \label{ob:weight:monotone} Consider an event which consists of either a relaxing update to the matrix $C$ or a call to the subroutine {\sc Boost}$(x, k)$. Such an event can only increase the weight $w_p(x, i)$ of a packing constraint $i \in [m_p]$, and it can only decrease the weight $w_c(x, j)$ of a covering constraint $j \in [m_c]$. 
 \end{observation}

\begin{proof}
Consider a relaxing update which increases the value of some entry $C(j, k)$ of the matrix $C$. This increases the value of $C_j x$; whereas for all $j' \in [m_c] \setminus \{j\}$ and $i \in [m_p]$, the values of $C_{j'} x$ and $P_i x$ remain unchanged. Accordingly, this event decreases the weight $w_c(x, j)$; whereas for all $j' \in [m_c] \setminus \{j\}$ and $i \in [m_p]$, the weights $w_c(x, j')$ and $w_p(x, i)$ remain unchanged due to this event. 

Next, consider a call to  {\sc Boost}$(x, k)$. For all $i \in [m_p]$ with $P(i, k) > 0$ this increases the value of $P_i x$, whereas for all $i \in [m_p]$ with $P(i, k) = 0$ the value of $P_i x$ remains unchanged. Similarly, for all $j \in [m_c]$ with $C(j, k) > 0$ this increases the value of $C_j x$, whereas for all $j \in [m_c]$ with $C(j, k) = 0$ the value of $C_j x$ remains unchanged. In effect, this means that for all $i \in [m_p]$ the weight $w_p(x, i)$ can only increase due to this event, whereas for all $j \in [m_p]$ the weight $w_c(x, j)$ can only decrease due to this event.
\end{proof}

\begin{observation}
\label{ob:cost:monotone:1}
Consider an event which consists of a call to the subroutine {\sc Boost}$(x, k)$ for some $k \in [n]$. Because of this event, the cost $\lambda(x, k')$ of any coordinate $k' \in [n]$ can only increase. 
\end{observation}

\begin{proof}
Follows from Observation~\ref{ob:weight:monotone} and Definition~\ref{def:coordinate}. 
\end{proof}

\begin{observation}
\label{ob:cost:monotone:2}
Consider a relaxing update to the entry $C(j, k)$ of the matrix $C$, where $j \in [m_c]$, $k \in [n]$. Because of this event, the cost $\lambda(x, k')$ of a coordinate $k' \in [n] \setminus \{ k\}$ can only increase, whereas the cost of the coordinate $k$ can change in either direction (i.e., it can increase or decrease).
\end{observation}

\begin{proof}
Because of this relaxing update, the value of $C(j, k)$ increases, whereas the values of all other entries in the matrices $C$ and $P$ remain unchanged. The proof  follows from Observation~\ref{ob:weight:monotone} and Definition~\ref{def:coordinate}. 
\end{proof}
 
 \noindent {\bf Key data structures:} We will explicitly maintain the following quantities: (1) The point $x \in \mathbb{R}_{\geq 0}^n$. (2) The weights $w_p(x, i), w_c(x, j)$ for all $i \in [m_p], j \in [m_c]$. (3) The total weights $w_p(x)$ and  $w_c(x)$.  (4) The values of $P_i x$ and $C_j x$, for all $i \in [m_p]$ and $j \in [m_c]$. 

\medskip
We are now ready to describe the exact implementation of our dynamic algorithm in more details.

  \medskip
 \noindent {\bf Phases:} Observation~\ref{ob:weight:monotone} implies that the ratio $\gamma(x) := w_p(x)/w_c(x)$ increases monotonically over time during the course of our algorithm. Armed with this observation, we split the working of our  algorithm into {\em phases}. In the beginning, we have $x = \vec{0}$ and hence $\gamma(x) = w_p(x)/w_c(x) = 1/1 = 1$. We initiate a new phase whenever the value of $\gamma(x)$ increases by a multiplicative factor of $(1+\red{\Theta(\epsilon)})$.  Thus, a given phase can span a sequence of calls to the {\sc Boost}$(x, k)$ subroutine at preprocessing, or it can also span a sequence of calls to the {\sc Boost}$(x, k)$ subroutine interspersed with a sequence of relaxing updates to the matrix $C$.

 \medskip
 \noindent {\bf Implementing a given phase:} Let $\gamma^{0}(x)$ be the value of $\gamma(x) := w_p(x)/w_c(x)$ at the start of the phase. During the phase, $\gamma(x)$ does not change by more than an  $(1+\red{\Theta(\epsilon)})$ factor. Accordingly, at any point in time within the phase, we classify a coordinate $k \in [n]$ as being cheap iff $\lambda(x, k) \leq (1+\red{\Theta(\epsilon)}) \cdot \gamma^0(x)$. 
 
 We maintain a set $E \subseteq [n]$ that always contains all the  coordinates that are currently cheap. At the start of the phase, we initialise $E \leftarrow [n]$, and then call the subroutine {\sc Boost-All}$(x, E)$ which is described in  Figure~\ref{fig:boostall}. Note that at the end of any iteration of the inner {\sc While} loop (steps 4 -- 9 in Figure~\ref{fig:boostall}), the concerned coordinate $k$ is no longer cheap. Furthermore, the monotonicity of the costs (as captured by Observation~\ref{ob:cost:monotone:1}) imply that a coordinate $k$, once removed from the set $E$, does {\em not} become cheap again within the same phase due to some future iteration of the outer {\sc While} loop in Figure~\ref{fig:boostall}.
 
 \smallskip
 \noindent {\em Handling a relaxing update within the phase:} Consider a relaxing update to an entry $C(j, k)$ of the matrix $C$. Just before this update, we had $E = \emptyset$. If the value of $\gamma(x)$ becomes more than $(1+\red{\Theta(\epsilon)}) \gamma^0(x)$ because of this update, then we terminate the current phase and initiate a new one. Otherwise, by Observation~\ref{ob:cost:monotone:2}, the only coordinate that can become cheap because of this update is $k$. Accordingly, we set $E \leftarrow \{ k\}$, and  call the subroutine {\sc Boost-All}$(x, E)$.   Observation~\ref{ob:cost:monotone:1} again implies   that as we  keep boosting $x$ along the coordinate $k$, it does not lead to any  other  coordinate $k' \in [n] \setminus \{ k \}$ becoming cheap. 
  
  \begin{figure*}[htbp]
                                                \centerline{\framebox{
                                                                \begin{minipage}{5.5in}
                                                                        \begin{tabbing}                                                                            
                                                                                1. \ \ \=   {\sc While} $E \neq \emptyset$: \\
                                                                                2. \> \ \ \ \  \ \ \ \ \= Consider any $k \in E$. \\
                                                                                3. \> \> $E \leftarrow E \setminus \{k\}$. 	\\
                                                                                4. \> \> {\sc While} $\lambda(x, k) \leq (1+\red{\Theta(\epsilon)}) \cdot \gamma^0(x)$ \\
                                                                                5. \> \> \ \ \ \ \ \ \ \ \ \= $x \leftarrow \text{{\sc Boost}}(x, k)$. \qquad \qquad	\qquad // See Figure~\ref{fig:boost}. \\ 
                                                                                6. \> \> \> {\sc If} $\min_{j \in [m_c]} C_j x \geq 1$, {\sc Then} \\
                                                                                7. \> \> \> \ \ \ \ \ \ \ \ \ \ \= {\sc Return} $x$.    \\
                                                                                8. \> \>  \> {\sc If} $\gamma(x) > (1+\red{\Theta(\epsilon)}) \cdot \gamma^0(x)$, {\sc Then} \\
                                                                                9. \>  \> \> \> Terminate the current phase.                                                                                                
                                                                                \end{tabbing}
                                                                \end{minipage}
                                                        }}
                                                        \caption{\label{fig:boostall} {\sc Boost-All}$(x, E)$.}
                                                \end{figure*}
  \medskip
The next two lemmas capture a couple of crucial properties of this dynamic algorithm.
 
  \begin{lemma}
 \label{lm:boost:bound}
 Fix any coordinate $k \in [n]$. The dynamic algorithm described above calls the subroutine {\sc Boost}$(x, k)$ at most $\red{O\left(\frac{\log^2(m_c + m_p + U/L)}{\epsilon^2}\right)}$ many times. 
 \end{lemma}
 
 \begin{proof}
  Fix a coordinate $k \in [n]$ for the rest of the proof. 
  
 Consider any  call to the {\sc Boost}$(x, k)$ subroutine during the course of our dynamic algorithm, which increases coordinate $k$ of the vector $x \in \mathbb{R}_{\geq 0}^n$ by some amount $\delta > 0$. From  Figure~\ref{fig:boost}, observe that just before this specific call to the  subroutine, either there exists a packing constraint $i \in [m_p]$ with $P(i, k) \cdot \delta  = \epsilon/\eta$, or there exists a {\em not-too-large} covering constraint $j \in [m_c]$\footnote{We say that a covering constraint $j \in [m_c]$ is {\em not-too-large} iff $C_j x < 2$.} with $C(j, k) \cdot \delta = \epsilon/\eta$.  In the former (resp.~latter) case, we refer to the index $i$ (resp.~$j$) as the {\em pivot} and the current value of $P(i, k)$ (resp.~$C(j, k)$) as the {\em pivot-value} corresponding to this specific call to the subroutine {\sc Boost}$(x, k)$.\footnote{Note that the values of the entries in the constraint matrix change over time due to the sequence of updates. The {\em pivot-value} of a call to {\sc Boost}$(x, k)$ refers to the value of the concerned entry in the constraint matrix {\em just before}  the specific call to  {\sc Boost}$(x, k)$.}

Without loss of generality, assume that $U = 2^{\tau} \cdot L$ where $\tau = \log (U/L)$ is an integer, and that the value of any nonzero entry in the constraint matrix lies in the interval $(L, U]$ at all times. Partition the interval $(L, U]$ into $\tau$  segments $\mathcal{I}_0, \ldots, \mathcal{I}_{\tau-1}$, where $\mathcal{I}_{\ell} = \left(L \cdot 2^{\ell}, L \cdot 2^{\ell+1} \right]$ for all $\ell \in \{0, \ldots, \tau-1\}$.

\begin{claim}
\label{cl:bound:pivot:covering}
Consider any $\ell \in \{0, \ldots, \tau-1\}$.  Throughout the duration of our dynamic algorithm, at most $4 \eta/\epsilon$ calls  with pivot-values in $\mathcal{I}_{\ell}$ and pivots in $[m_c]$ are made to the {\sc Boost}$(x, k)$ subroutine (see~(\ref{eq:eta})).
\end{claim}

\begin{proof}
Let $x_k \in \mathbb{R}_{\geq 0}$ denote the $k^{th}$ coordinate of the vector $x \in \mathbb{R}_{\geq 0}^n$. For ease of exposition, we say that a call to  {\sc Boost}$(x, k)$ is  {\em covering-critical} iff its pivot is in $[m_c]$ and its pivot-value is in $\mathcal{I}_{\ell}$. We wish to upper bound the total number of covering-critical calls made during the course of our dynamic algorithm. 

Consider any specific covering-critical call to {\sc Boost}$(x, k)$  with pivot $j \in [m_c]$ and pivot-value  $\alpha \in \mathcal{I}_{\ell} = \left( L \cdot 2^{\ell-1}, L \cdot 2^{\ell}\right]$. From Figure~\ref{fig:boost}, it follows that this call increases  $x_k$ by $\delta = \epsilon/(\eta \alpha) \geq \epsilon/(\eta L 2^{\ell})$.

Suppose that   $T$ covering-critical calls have been made to the {\sc Boost}$(x, k)$ subroutine, where $T > 4\eta/\epsilon$. Let us denote these calls in increasing order of time by: $\Gamma_1, \Gamma_2, \ldots, \Gamma_T$. Thus, for all $t \in [T]$, we let $\Gamma_{t}$ denote the $t^{th}$ covering-critical call made during the course of our dynamic algorithm.  From the discussion in the preceding paragraph,  each of these critical calls $\Gamma_t$ increases $x_k$ by at least $\epsilon/(\eta L 2^{\ell})$. 

Initially $x_k = 0$,  and  it increases monotonically over time. Hence, just before the last critical call $\Gamma_{T}$ we have $x_k \geq (T-1) \cdot \epsilon/(\eta L 2^{\ell}) \geq (4 \eta/\epsilon) \cdot \epsilon/(\eta L 2^{\ell}) = 1/(L 2^{\ell-2})$. Let $j_T \in [m_c]$ be the pivot of the  call $\Gamma_T$. Since the pivot-value of $\Gamma_T$ lies in $\mathcal{I}_{\ell}$, it follows that $C(j_T, k) > L \cdot 2^{\ell-1}$ just before the  call $\Gamma_T$. Thus, just before the  call $\Gamma_T$, we have $C_{j_T} \cdot x \geq C(j_T, k) \cdot x_k > L  2^{\ell-1} \cdot 1/(L 2^{\ell-2}) = 2$. But if $C_{j_T} \cdot x > 2$ just before the call $\Gamma_T$, then $j_T$ cannot be the pivot of $\Gamma_T$ (see Figure~\ref{fig:boost}). This leads to a contradiction. Hence, we must have $T \leq 4\eta/\epsilon$, and this concludes the proof of the claim.
\end{proof}

\begin{claim}
\label{cl:bound:pivot:packing}
Consider any $\ell \in \{0, \ldots, \tau-1\}$.  Throughout the duration of our dynamic algorithm, at most $4 \eta/\epsilon$ calls with pivot-values in $\mathcal{I}_{\ell}$ and pivots $i \in [m_p]$ are made to the {\sc Boost}$(x, k)$ subroutine (see~(\ref{eq:eta})).
\end{claim}

\begin{proof}(Sketch) As in the proof of Claim~\ref{cl:bound:pivot:covering}, let $x_k \in \mathbb{R}_{\geq 0}$ denote the $k^{th}$ coordinate of the vector $x \in \mathbb{R}_{\geq 0}^n$. Say that a call to  {\sc Boost}$(x, k)$ is  {\em packing-critical} iff its pivot is in $[m_p]$ and its pivot-value is in $\mathcal{I}_{\ell}$.

Suppose that   $T$ covering-critical calls have been made to the {\sc Boost}$(x, k)$ subroutine, where $T > 4\eta/\epsilon$. Let us denote these calls in increasing order of time by: $\Gamma_1, \Gamma_2, \ldots, \Gamma_T$. Thus, for all $t \in [T]$, we let $\Gamma_{t}$ denote the $t^{th}$ packing-critical call made during the course of our dynamic algorithm. Let $i_T \in [m_p]$ be the pivot of the last packing-critical call $\Gamma_T$. Following the same argument as in the proof of Claim~\ref{cl:bound:pivot:covering}, we conclude that $P_{i_T} \cdot x > 2$ just before the call $\Gamma_T$.  But if $P_{i_T} \cdot x > 2$, then  $f_p(x) \geq \max_{i \in [m_p]} P_i x > 2$  according to~(\ref{eq:soft:approx}). Now, applying Invariant~\ref{inv:dynamic},  we conclude that just before the call $\Gamma_T$ we have: $\min_{j \in [m_c]} C_j x \geq f_c(x)  \geq (1+\epsilon)^{-1} \cdot f_p(x)  > (1+\epsilon)^{-1} \cdot 2  > 1$. This leads to a contradiction, since if $\min_{j \in [m_c]} C_j x > 1$, then our dynamic algorithm would immediately declare that it has found a $\red{\Theta(\epsilon)}$-approximate solution to the input LP from this point onward, without making any further calls to the {\sc Boost}$(x, k)$ subroutine (in particular, without making the call $\Gamma_T$).  Thus, it must be the case that $T \leq 4\eta/\epsilon$, and this concludes the proof of the claim.
\end{proof}

Any call to the  {\sc Boost}$(x, k)$ subroutine has its pivot in $[m_p] \cup [m_c]$ and pivot-value in $\mathcal{I}_{\ell}$, for some $\ell \in \{0, \ldots, \tau-1\}$.\footnote{This holds because the pivot-values lie in the range $(L, U]$, and this range has been partitioned into subintervals: $\mathcal{I}_0, \ldots, \mathcal{I}_{\tau-1}$. } Thus,  Claim~\ref{cl:bound:pivot:covering} and Claim~\ref{cl:bound:pivot:packing} imply that the total number of calls made to the {\sc Boost}$(x, k)$ subroutine is at most $2 \cdot \tau \cdot (4 \eta/\epsilon) = O\left(\red{\frac{\log^2(m_c + m_p + U/L)}{\epsilon^2}}\right)$. 
 \end{proof}
 
 \begin{lemma}
 \label{lm:phase:bound} The dynamic algorithm described above has at most $O\left(\red{\frac{\log(m_c + m_p + U/L)}{\epsilon^2}}\right)$ many phases. 
 \end{lemma}
 
 \begin{proof}
 By Observation~\ref{ob:weight:monotone},  the ratio $\gamma(x) := w_p(x)/w_c(x)$ increases monotonically over time. Initially, we have $x = \vec{0}$, and hence $\gamma(x) = w_p(x)/w_c(x) = 1$.  Let $\kappa$ be the time-instant at which the last ever call to {\sc Boost}$(x, k)$ is made by our dynamic algorithm. We will  upper bound  $\gamma(x)$  just before the time-instant $\kappa$. 
 
Just before the time-instant $\kappa$, by Corollary~\ref{cor:inv:dynamic} we have $\max_{i \in [m_p]} P_i x \leq 1+\red{\Theta(\epsilon)}$ and hence $w_p(x, i) \leq \exp(\eta (1+\red{\Theta(\epsilon)}))$ for all $i \in [m_p]$, which implies that $w_p(x) = \sum_{i=1}^{m_p} w_p(x, i) \leq m_p \cdot \exp(\eta (1+\red{\Theta(\epsilon)}))$. 

Just before the time-instant $\kappa$, we also have $\min_{j \in [m_c]} C_j x < 1$.  Hence, at that moment there exists some $j' \in [m_c]$ with $C_{j'} x < 1$ and $w_c(x, j') \geq 1/\exp(\eta)$. This means that $w_c(x) \geq w_c(x, j') \geq 1/\exp(\eta)$. 

To summarise, we infer that $\gamma(x) = w_p(x)/w_c(x) \leq m_p \cdot \exp(\eta(1+\red{\Theta(\epsilon)})) \cdot \exp(\eta) = \red{O(m_p \cdot \exp(3\cdot \eta))}$ just before the time-instant $\kappa$. Thus, the value of $\gamma(x)$ increases from $1$ at the start of the algorithm to at most $\Lambda = \red{O(\exp(3 \cdot \eta) \cdot m_p)}$ just before the last call to the {\sc Boost}$(x, k)$ subroutine. 

Note that $\gamma(x)$ increases monotonically over time in the range $[1, \Lambda]$, and we  create a new phase whenever  $\gamma(x)$ increases by a $(1+\red{\Theta(\epsilon)})$  factor. Thus, from~(\ref{eq:eta}), we conclude that the total number of phases is at most $\log_{(1+\red{\Theta(\epsilon)})} \Lambda = O\left(\red{\frac{\log (m_p + m_c + U/L)}{\epsilon^2}}\right)$. 
\end{proof}

 \noindent {\bf Analysing the total update time:}  We first focus on bounding the total time spent on all the calls to the {\sc Boost}$(x, k)$ subroutine during the course of our algorithm. By Lemma~\ref{lm:boost:bound},  the point $x$ gets boosted along a given coordinate $k \in [n]$ at  most $\Gamma = \red{O\left(\frac{\log^2(m_c + m_p + U/L)}{\epsilon^2}\right)}$ times.  Whenever we boost $x$  along some coordinate $k$, we need to spend $O(nz(k))$ time to update all the relevant weights $\{ w_p(x,i)\}$, $\{ w_c(x, j) \}$, $w_p(x)$, $w_c(x)$ and the values $\{ C_j x\}, \{ P_i x\}$; where $nz(k)$ is the total number of nonzero entries in the $k^{th}$ column of the  matrices $P, C$ at the current moment. Hence, the total time spent in  all the calls to the {\sc Boost}$(x, k)$ subroutine  is  at most $\sum_{k=1}^n \Gamma \cdot O(nz(k)) = \Gamma \cdot O(N) = \red{O\left(N \cdot \frac{\log^2(m_c + m_p + U/L)}{\epsilon^2}\right)}$.

Next, note that whenever we create a new phase, we set $E \leftarrow [n]$ and call the subroutine {\sc Boost-All}$(x, E)$ as  in Figure~\ref{fig:boostall}. Also, whenever there is a relaxing update to an entry $C(j, k)$ in the covering matrix, we  set $E \leftarrow \{ k \}$ and then make a call to {\sc Boost-All}$(x, E)$. Thus, it might very well be the case that a coordinate $k \in [n]$   moves in and out of the set $E$ on multiple occasions, without being boosted at all. The time spent on these  {\em apparently futile} operations, which modify the set $E$ without leading to any call to {\sc Boost}$(x, k)$, is captured by  steps 1--3 of Figure~\ref{fig:boostall}. Lemma~\ref{lm:phase:bound} implies that the total time spent in this manner, throughout a sequence of  $t$ relaxing updates and across all the coordinates, is  at most:
\begin{eqnarray*} 
(\text{number of phases}) \cdot O(\text{number of coordinates}) + O(t) & = & \red{O\left(n \cdot \frac{\log (m_c + m_p + U/L)}{\epsilon^2} +t \right)} \\
& = & \red{O\left(N \cdot \frac{\log (m_c + m_p + U/L)}{\epsilon^2} +t \right)}.
\end{eqnarray*}

Finally, note that the preprocessing time is  at most $O\left(N  \cdot \frac{\log(m_c + m_p)}{\epsilon^2} \right)$, since the static algorithm from Section~\ref{sec:static} can be implemented in $O\left(N  \cdot \frac{\log(m_c + m_p)}{\epsilon^2} \right)$ time~\cite{Young2014nearly}. This leads us to Theorem~\ref{th:relaxing:positiveLP:algo} for relaxing entry updates to covering constraints.

 \subsubsection{The full algorithm: Handling relaxing entry updates to both packing and covering constraints}
 \label{subsec:dynamic:full}

 We start by identifying the main difficulty in extending the dynamic algorithm from Section~\ref{subsec:dynamic:covering} to the setting where relaxing updates occur to {\em both} packing and covering constraints of the input positive LP. Consider a relaxing update to the entry $P(i, k)$, for some $i \in [m_p], k \in [n]$. Due to this relaxing update, the value of $P(i, k)$ gets reduced, and accordingly the weight $w_p(x, i)$ also decreases. This contradicts the monotonicity of the weights as captured in Observation~\ref{ob:weight:monotone}. Thus, $w_p(x)$ no longer increases monotonically with time, which in turn invalidates the proof of Lemma~\ref{lm:phase:bound}. We circumvent this difficulty by working with an {\em extended} positive LP $(P^*, C^*)$, which is derived from the input  LP $(P, C)$ in the following manner. 
 
 \medskip
 \noindent {\bf The extended LP:}
Throughout the sequence of updates to the input LP, we ensure that we always have:
\begin{itemize}
\item  $P^* \in \mathbb{R}_{\geq 0}^{m_p \times (n+1)}$, where $P^*(i, k) = P(i, k)$ and $P^*(i, n+1) \geq 0$ for all $i \in [m_p], k \in [n]$. 
\item  $C^* \in \mathbb{R}_{\geq 0}^{m_c \times (n+1)}$, where $C^*(j, k) = C(j, k)$ and $C^*(j, n+1) = 0$ for all $j \in [m_c], k \in [n]$. 
\end{itemize}
It is easy to check that the input LP $(P, C)$ is feasible iff the extended LP $(P^*, C^*)$ is feasible. Furthermore, if $x^* = (x^*_1, \ldots, x^*_{n+1}) \in \mathbb{R}_{\geq 0}^{n+1}$ is a $\red{\Theta(\epsilon)}$-approximate solution to the extended LP $(P^*, C^*)$, then $x = (x^*_1, \ldots, x^*_n) \in \mathbb{R}_{\geq 0}^n$ is a $\red{\Theta(\epsilon)}$-approximate solution to the input LP $(P, C)$. Accordingly, our new dynamic algorithm will attempt to approximately solve the extended LP $(P^*, C^*)$. Note that  we are free to set the values $P^*(i, n+1)$ of the matrix $P^*$ in any way we chose (provided they remain nonnegative). We will use this to our advantage while dealing with relaxing updates to the packing constraints in the input LP. It is important to emphasise that except the entries $P^*(i, n+1)$, at any point in time the values of all other entries in the matrices $P^*$ and $C^*$ are determined by the current status of the input LP $(P, C)$. 

\medskip
\noindent {\bf Initialisation:}  At preprocessing, we set $P^*(i, n+1) \leftarrow 0$ for all $i \in [m_p]$. We  start with an $x^* = (x^*_1, \ldots, x^*_{n+1}) \in \mathbb{R}_{\geq 0}^{n+1}$ where $x^*_1 = \cdots = x^*_n = 0$ and $x^*_{n+1} = 1$. At this point in time, we have $f_p(x^*) = f_c(x^*) = 0$. From this point onward, we will keep modifying  $x^*$ until it becomes a $\red{\Theta(\epsilon)}$-approximate solution to the extended LP $(P^*, C^*)$, by following the same template as outlined in Section~\ref{subsec:dynamic:template} and Section~\ref{subsec:dynamic:covering}. In addition,  now after every relaxing update to an entry in the packing matrix (say) $P(i, k)$, we will change the value of $P^*(i, n+1)$ in such a way that works to our advantage. This is explained in more details below. For clarity of exposition, henceforth we assume that after an update to an entry $P(i, k)$ (resp.~$C(j, k)$) of the matrix $P$ (resp.~$C$), the corresponding entry $P^*(i, k)$ (resp.~$C^*(j, k)$) implicitly gets updated so as to ensure the equality $P^*(i, k) = P(i, k)$ (resp.~$C^*(j, k) = C(j, k)$). 

\medskip
\noindent {\bf Pseudo-updates to the extended LP:}  Consider any relaxing update to an entry $P(i, k)$ of the matrix $P$, where $i \in [m_p], k \in [n]$. Immediately after this relaxing update, we {\em increase} the value of $P^*(i, n+1)$ to such an extent that the value of $P^*_i x^*$ (and hence the weight $w_p(x^*)$) remains unchanged. We refer to this step as a {\em pseudo-update} to the extended LP. The only purpose behind the pseudo-update is to ensure that the monotonicity of the weights as captured by Observation~\ref{ob:weight:monotone} continues to hold. 

\medskip
 In summary, the overall algorithm follows  the same template as  in Section~\ref{subsec:dynamic:template} and Section~\ref{subsec:dynamic:covering}, with  the following caveat: Immediately after every relaxing update to a packing constraint in the input LP, we perform the corresponding pseudo-update to the extended LP so as to ensure the validity of Observation~\ref{ob:weight:monotone}. It is easy to check that all the observations, lemmas and inferences derived in Section~\ref{subsec:dynamic:covering} continue to hold here. This leads us to the  dynamic algorithm promised in Theorem~\ref{th:relaxing:positiveLP:algo}, for relaxing entry updates.

\subsection{Handling Restricting (resp.~Relaxing) Updates to a Packing (resp.~Covering) LP}
\label{sec:exception:packing}

Consider the setting where we wish to maintain a solution to Problem~\ref{prob:basic} when the matrix $C$ is undergoing relaxing entry updates (or equivalently, the matric $C^{\top}$ is undergoing restricting entry updates). Thus, each update increases the value of some entry of $C$. The idea is simple. We run the dynamic algorithm from Section~\ref{sec:dynamic} on the positive LP defined by~(\ref{eq:mixed:LP:1}) and~(\ref{eq:mixed:LP:2}), which is undergoing relaxing entry updates. 
\begin{eqnarray}
\mathbb{1}^{\top} x & \leq & 1 \label{eq:mixed:LP:1} \\
C x & \geq & \mathbb{1} \label{eq:mixed:LP:2}
\end{eqnarray}

Suppose that this  algorithm has just finished processing an update. Consider two possible cases.

\medskip
\noindent {\bf Case I:} The algorithm returns a solution $x \in \mathbb{R}_{\geq 0}^n$ which satisfies $\mathbb{1}^{\top} x \leq 1 +\Theta(\epsilon)$ and $C x \geq (1-\Theta(\epsilon)) \cdot \mathbb{1}$. In this case, we  simply return the vector $x$ and terminate our algorithm. The vector $x$ will continue to remain an $\epsilon$-approximate solution to the input LP after any future update.

\medskip
\noindent {\bf Case II:} The algorithm declares that the positive LP is infeasible as there is no cheap co-ordinate. From Section~\ref{subsec:static}, recall that $w_c(x, j)$ denotes the weight of a covering constraint $j \in [m]$, and that $w_c(x) = \sum_{j \in [m]} w_c(x, j)$ denotes the total weight of all the covering constraints. Define the vector $\hat{y} \in \mathbb{R}_{\geq 0}^m$, where $\hat{y}_j = w_c(x, j)/w_c(x)$ for all $j \in [m]$. We now claim that $\hat{y}$ is a feasible solution to the dual packing LP.
\begin{claim}
\label{cl:restricting}
We have $\mathbb{1}^{\top} \hat{y} = 1$ and $C^{\top} \hat{y} \leq \mathbb{1}$.
\end{claim}

\begin{proof}
Since $\{\hat{y}_j\}$ denotes the normalized weights of the covering constraints, we have $\mathbb{1}^{\top} \hat{y} = 1$. Next, fix any $k \in [n]$. Since we are in Case II, the coordinate $k$ is {\em not} cheap. So, from Definition~\ref{def:coordinate}, we infer that:
$$\frac{1}{\sum_{j \in [m]} \hat{y}_j \cdot C(j,k)} \geq 1+\Theta(\epsilon), \text{ or equivalently, } \left( C^{\top} \hat{y}\right)_k =  \sum_{j \in [m]} \hat{y}_j \cdot C(j,k) \leq 1 - \Theta(\epsilon) \leq 1.$$
Accordingly, we get $C^{\top} \hat{y} \leq \mathbb{1}$. 
\end{proof}
Our algorithm maintains an estimate $w^*_c(x) \in \left[ w_c(x), (1+\epsilon) w_c(x) \right]$.  If after processing an update, we end up in Case II, then we return a $y \in \mathbb{R}_{\geq 0}^m$, where $y_j := w_c(x, j)/w^*_c(x)$ for all $j \in [m]$. As $\hat{y}_j \leq y_j \leq (1+\Theta(\epsilon)) \cdot \hat{y}_j$ for all $j \in [m]$,  Claim~\ref{cl:restricting} implies that $\mathbb{1}^{\top} y \geq 1$ and $C^{\top} y \leq (1+\Theta(\epsilon)) \cdot \mathbb{1}$.

\medskip
It is easy to maintain the vector $y$ without any significant overhead in the total update time,  because of  three reasons: (1) the algorithm from Section~\ref{sec:dynamic} explicitly maintains the weight $w_c(x, j)$ for each $j \in [m]$, (2) we do not need to change the estimate $w^*_c(x)$ as long as we remain in the same phase, and (3) by Lemma~\ref{lm:phase:bound} we have at most $O\left(\frac{\log(m + U/L)}{\epsilon^2}\right)$ many phases.  Note that whenever we start a new phase we have to spend $O(m)$ time to update the vector $y$. Theorem~\ref{th:packing:restricting:1} now follows from Theorem~\ref{th:relaxing:positiveLP:algo}.

\begin{theorem}
\label{th:packing:restricting:1}
We can deterministically maintain a solution to Problem~\ref{prob:basic} when the matrix $C$ undergoes $t$ relaxing entry updates in $O\left(t+ N \cdot \frac{\log^2(m + U/L)}{\epsilon^2} +  m \cdot \frac{\log(m + U/L)}{\epsilon^2}\right) = O\left(t+ N \cdot \frac{\log^2(m + U/L)}{\epsilon^2}\right)$ total  time, where $N$ denotes the maximum number of non-zero entries in $C$ throughout these updates.
\end{theorem}

Finally, we consider the problem of maintaining an $\epsilon$-approximate optimal solution to a generic covering (resp.~packing) LP under relaxing (resp.~restricting) updates. It is easy to verify that if we start with Theorem~\ref{th:packing:restricting:1} and apply the same reduction outlined in Section~\ref{sec:reductions:dynamic}, then we obtain the following result.

\begin{theorem}
\label{th:main:dynamic:restricting:packing}
Consider any sequence of $t$ relaxing updates to LP~(\ref{eq:covering:general}), where each update either increases an entry of  $C$, or decreases an entry of $a, b$. Throughout these updates, let $N$ denote the maximum number of non-zero entries in $C$, and let $L$ (resp.~$U$) be a lower (resp.~upper) bound on the minimum (resp.~maximum) value of a non-zero entry of $C, a, b$. We can deterministically maintain an $\epsilon$-approximate optimal solution to LP~(\ref{eq:covering:general}) in $O\left(t \cdot \frac{\log (nU/L)}{\epsilon} +  N \cdot \frac{\log^2(m + U/L)}{\epsilon^3} \cdot \log (nU/L)\right)$ total  time. The same guarantee also  holds for maintaining an $\epsilon$-approximate solution to LP~(\ref{eq:packing:general}) under restricting updates.
\end{theorem}

%% file: lower-bounds.tex
\section{Conditional Lower Bounds}

\label{sec:lowerbounds:pure}

In previous sections, we presented partially dynamic algorithms for packing-covering LP-s with polylogarithmic amortized update times. This leads to a couple of natural questions:  (1) Can we solve the same problem in polylogarithmic {\em worst case} update time? (2) Can we design a {\em fully dynamic} algorithm for packing-covering LPs in polylogarithmic {\em amortized} update time? We will now provide conditional lower bounds that rule out the possibility of having these stronger guarantees. 

Our lower bounds hold under Strong Exponential Time Hypothesis (SETH), and the starting points of our reductions are  Problem~\ref{prob:inner-product} and \Cref{th:hardness} described below~\cite{AbboudRW17}.


\begin{problem}
\label{prob:inner-product} As input, we get four parameters $N, k, m, \beta^*$ such that $\Theta\left(N^{o(1)} \right) = \beta^* \leq k \leq m = \Theta\left(N^{o(1)}\right)$, and two collections of $N$ vectors $A = \left\{a^{(1)}, \ldots, a^{(N)} \right\}$ and $B = \left\{ b^{(1)}, \ldots, b^{(N)} \right\}$ such that $a, b \in \{0,1\}^m$ for all $a \in A, b \in B$. Furthermore, each vector $b \in B$ has exactly $k$ non-zero entries. We want to design an algorithm that can distinguish between the following two cases.
\begin{itemize}
    \item (1) There exists some $a \in A$ and $b \in B$ such that $a \geq b$.
    \item (2) For all $a \in A$ and $b \in B$, we have $a \cdot b \leq k^* := \frac{k}{\beta^*}$. 
\end{itemize}
\end{problem}

\begin{theorem}\cite{AbboudRW17}
\label{th:hardness}
Under SETH, any algorithm for Problem~\ref{prob:inner-product}  runs in $\Omega(N^{2-\epsilon})$ time for all $\epsilon > 0$.
\end{theorem}

Now, we are ready to prove \Cref{thm:main lb}, and we state a more detailed version of it below.

\begin{theorem}

\label{thm:lowerbound:fully:dynamic}

Assume SETH holds. Let $N$ and $\text{{\sc Opt}}$ respectively denote the number of non-zero entries in the input LP and its optimal objective value. Then there exists some $\beta = \Theta\left( N^{o(1)} \right)$ such that:
\begin{itemize}
\item (i) Any algorithm which maintains an estimate $\frac{\text{{\sc Opt}}}{\beta} \leq \nu \leq \text{{\sc Opt}}$ for a packing or covering LP under fully dynamic entry updates must have an amortized update time of $\Omega(N^{1-o(1)})$. 
\item (ii) Any algorithm which maintains an estimate $\frac{\text{{\sc Opt}}}{\beta} \leq \nu \leq \text{{\sc Opt}}$ for a packing or covering LP under partially dynamic entry updates must have a worst-case update time of $\Omega(N^{1-o(1)})$. 
\end{itemize}
Furthermore, the same lower bounds hold even under translation updates.
\end{theorem}

\subsection{Proof of Theorem~\ref{thm:lowerbound:fully:dynamic} - part (i)}
\label{sec:thm:lowerbound:fully:dynamic:i}

We  prove the theorem for packing LPs. The proof for covering LPs immediately follows via LP duality.

\medskip
Recall the parameters $N, k, m, \beta^*$ from the statement of Problem~\ref{prob:inner-product}. Set $\beta := \beta^*/3$. Suppose that there is a dynamic algorithm $\A$ that maintains an estimate $\frac{\text{{\sc Opt}}}{\beta} \leq \nu \leq \text{{\sc Opt}}$ for a packing LP undergoing fully dynamic entry updates in $\Theta(T)$ amortized update time. Using $\A$ as a subroutine, we  design a static algorithm $\A^*$ for Problem~\ref{prob:inner-product}, as described below.

Upon receiving  its input, the algorithm $\A^*$  defines the following packing LP, where initially we have $\lambda_j = 1$ for all $j \in [m]$. The algorithm $\A^*$ now feeds this packing LP as an input to  $\A$ at  pre-processing.
\begin{eqnarray}
\label{eq:lowerbound:packingLP}
\text{Maximize } \sum_{j=1}^m x_j \\
\text{s.t. } \ \  a \cdot x & \leq & k^*  \qquad  \text{ for all } a \in A. \\
\lambda_j \cdot x_j & \leq & 1 \qquad \ \  \text{ for all } j \in [m]. \\
x_j & \geq & 0 \qquad \ \ \text{ for all } j \in [m].
\end{eqnarray}
Subsequently, the algorithm $\A^*$ performs the steps outlined in \Cref{fig:lower-bound}.

  \begin{figure*}[htbp]
                                                \centerline{\framebox{
                                                                \begin{minipage}{5.5in}
                                                                        \begin{tabbing}
                                                                        1. \ \ \=  {\sc For All} $b \in B$: \\
                                                                                2. \> \ \ \ \  \ \ \ \ \=     Set $\lambda_j = L := m/k^*$ for all $j \in [m]$ with~$b_j = 0$, and feed these updates to $\A$. \\
                                                                                3. \> \> {\sc If} $\A$ returns an estimate $\nu \leq 2 k^*$, {\sc then}  \\
                                                                                4. \> \> \ \ \ \ \ \ \ \ \= {\sc Report} that we are in case (1) of Problem~\ref{prob:inner-product}, and {\sc terminate}. \\
                                                                                5. \> \> Return the LP to its initial state, by setting $\lambda_j = 1$ for all $j \in [m]$ with $b_j = 0$, \\
                                                                                \> \> and feeding these updates to the dynamic algorithm $\A$. \\
                                                                                6. \> {\sc Report} that we are in case (2) of Problem~\ref{prob:inner-product} and {\sc terminate}.
                                                                               \end{tabbing}
                                                                \end{minipage}
                                                        }}
                                                        \caption{Solving Problem~\ref{prob:inner-product} using the dynamic algorithm $\A$.} \label{fig:lower-bound}
                                                \end{figure*}

\begin{claim}
\label{cl:lowerbound:1}
Suppose that  $a \cdot b \leq k^*$ for all $a \in A, b \in B$. Then the algorithm $\A^*$ correctly reports that we are in case (2) of Problem~\ref{prob:inner-product}.
\end{claim}

\begin{proof}
Consider any vector $b \in B$, and focus on the corresponding iteration of the {\sc For} loop in \Cref{fig:lower-bound}. After step (2) of this iteration, we have $\lambda_j = L$ for all $j \in [m]$ with $b_j = 0$, and $\lambda_j = 1$ for all $j \in [m]$ with $b_j = 1$. Now, define an $x$ where $x_j = 1$ for all $j \in [m]$ with $b_j = 1$, and $x_j = 0$ for all $j \in [m]$ with $b_j = 0$. Such an $x$ is a feasible solution to the LP and has an objective  of $\sum_{j=1}^m x_j = \left| \{ j \in [m] : b_j = 1\} \right| = k$. This implies that $\text{{\sc  Opt}} \geq k$. Accordingly, in step (3), the dynamic algorithm $\A$ returns a value $\nu \geq \frac{k}{\beta} = \frac{3k}{\beta^*} = 3 k^*$. Hence, the algorithm $\A^*$ does {\em not} terminate during this iteration.

The preceding discussion implies that  $\A^*$ does not terminate during any iteration of the {\sc For} loop in \Cref{fig:lower-bound}. Instead, at the end of the {\sc For} loop, it correctly reports that we are in case (2) of Problem~\ref{prob:inner-product}.
\end{proof}

\begin{claim}
\label{cl:lowerbound:2}
Suppose that there exists vectors $a \in A, b \in B$ such that 
 $a \geq b$. Then the algorithm $\A^*$ correctly reports that we are in case (1) of Problem~\ref{prob:inner-product}.
 \end{claim}
 
 \begin{proof}
Focus on the iteration of the {\sc For} loop in \Cref{fig:lower-bound} which deals with the vector $b \in B$. After step (2) of this iteration, we have $\lambda_j = L$ for all $j \in [m]$ with $b_j = 0$, and $\lambda_j = 1$ for all $j \in [m]$ with $b_j = 1$. Partition the set $[m]$ into two subsets: $P_0 = \{ j \in [m]: b_j = 0\}$ and $P_1 = \{ j \in [m] : b_j = 1\}$. Since $a \geq b$, we have $a_j = 1$ for all $j \in P_1$. Let $x$ be any feasible solution to LP~(\ref{eq:lowerbound:packingLP}). Observe that:
\begin{equation}
    \label{eq:lowerbound:10}
\sum_{j \in P_1}  x_j = \sum_{j \in P_1} a_j x_j \leq a \cdot x \leq k^*.
\end{equation}
Next, recall that $\lambda_j = L$ for all $j \in P_0$. Thus, we have $x_j \leq 1/L$ for all $j \in P_0$, and hence:
\begin{equation}
    \label{eq:lowerbound:11}
    \sum_{j \in P_0} x_j \leq \frac{m}{L} = \kappa^*.
\end{equation}
Summing~(\ref{eq:lowerbound:10}) and~(\ref{eq:lowerbound:11}), we now get: $\sum_{j=1}^m x_j \leq 2 k^*$. Since this inequality holds for {\em every} feasible solution $x$ to the LP, we get $\text{{\sc Opt}} \leq 2k^*$. Accordingly, during step (3) of the concerned iteration of the {\sc For} loop in \Cref{fig:lower-bound}, the dynamic algorithm $\A$ returns a value $\nu \leq 2 k^*$. Thus, the algorithm $\A^*$ terminates at this point, before correctly reporting that we are in case (1) of Problem~\ref{prob:inner-product}.
\end{proof}

The correctness of algorithm $\A^*$ follows from Claim~\ref{cl:lowerbound:1} and Claim~\ref{cl:lowerbound:2}. The claim below bounds its run time in terms of the  update time of the dynamic algorithm $\A$.

\begin{claim}
\label{cl:lowerbound:3}
The algorithm $\A^*$ runs in $\hat{O}(N \cdot T)$ time, where $T$ is the update time of  $\A$.
\end{claim}

\begin{proof}
The algorithm $\A^*$ starts by feeding LP~(\ref{eq:lowerbound:packingLP}), which is of $\hat{O}(N)$ size, to $\A$ as input at pre-processing. Subsequently, during each iteration $b \in B$ of the {\sc For} loop, it feeds $\Theta(m) = \hat{O}(1)$  updates to $\A$. Since the {\sc For} loop runs for at most $N$ iterations, overall the dynamic algorithm $\A$ needs to handle at most $\hat{O}(N)$  updates. Hence, the total time spent by the dynamic algorithm $\A$ is at most $\hat{O}(N \cdot T)$. The claims follows since the runtime of $\A^*$ is dominated by the total update time of $\A$.
\end{proof}

\Cref{th:hardness} implies that assuming SETH, the run time of algorithm $\A^*$ must be $\Omega(N^{2-\epsilon})$ for every constant $\epsilon > 0$. Hence, from Claim~\ref{cl:lowerbound:3} it follows that the amortized update time of the dynamic algorithm $\A$ is at least $T = \Omega(N^{1-o(1)})$. This concludes the proof of part-(i) of \Cref{thm:lowerbound:fully:dynamic}.

\medskip
If we wish to prove the same lower bound for translation updates, then the reduction remains almost the same. The only difference now is that in step (2) of Figure~\ref{fig:lower-bound}, instead of changing the value of the coefficient $\lambda_j$ we now change the constraints $x_j \leq 1$ to $x_j \leq 0$ for the concerned co-ordinates $j$, and un-do these changes in step (5) as before.

\subsection{Proof of Theorem~\ref{thm:lowerbound:fully:dynamic} - part (ii)}
\label{sec:thm:lowerbound:fully:dynamic:ii}

We  prove the theorem for packing LPs. The proof for covering LPs immediately follows via LP duality.

\medskip
As in \Cref{sec:thm:lowerbound:fully:dynamic:ii}, set $\beta := \beta^*/3$. Suppose that there is a dynamic algorithm $\A$ that maintains an estimate $\frac{\text{{\sc Opt}}}{\beta} \leq \nu \leq \text{{\sc Opt}}$ for a packing LP undergoing restricting entry updates in $\Theta(T)$ worst-case update time. Using $\A$ as a subroutine, we  design a static algorithm $\A^*$ for Problem~\ref{prob:inner-product}. To achieve this goal, we follow almost exactly the same strategy as in \Cref{sec:thm:lowerbound:fully:dynamic:ii}. The only difference being that in step (5) of Figure~\ref{fig:lower-bound}, instead of feeding a new set of updates to $\A$, we {\em roll back} the state of $\A$ to where it was at pre-processing. This allows us to deal with the fact that $\A$ can only handle restricting ({\em not} fully dynamic) updates. The rest of the arguments in the proof remain the same as in \Cref{sec:thm:lowerbound:fully:dynamic:ii}.

\medskip
The preceding discussion implies the desired lower-bound when the input LP undergoes restricting updates. To derive the analogous lower-bound for relaxing updates, we need to make only the following minor tweak in our design of algorithm $\A^*$: Initially, we set $\lambda_j = L$ for all $j \in [m]$. Subsequently, in step (2) of Figure~\ref{fig:lower-bound}, we set $\lambda_j = 1$ for all $j \in [m]$ with $b_j = 1$, and feed these updates to $\A$. Next, in step (5) of Figure~\ref{fig:lower-bound}, we {\em roll back} the state of $\A$ to where it was at pre-processing.

\medskip
Finally, we can easily extend our lower bounds to hold in settings where we only have translation updates, using the idea outlined in the last paragraph of \Cref{sec:thm:lowerbound:fully:dynamic:ii}.

%% file: appendix-whack-a-mole-mwu.tex
\section{Proofs of Theorem~\ref{th:MWU:static:approx:dyn} and Lemma~\ref{cor:bound:weight:dyn}}
\label{sec:dyn:missing:proofs}

\paragraph{Proof of Lemma~\ref{cor:bound:weight:dyn}:} To conclude Lemma~\ref{cor:bound:weight:dyn} notice that Lemma~\ref{cor:bound:weight} seamlessly extends to the dynamic setting. The only property of the dynamic input matrix $C$ the proof of Lemma~\ref{cor:bound:weight} exploits is that at the time of the $t$-th \textsc{Whack} sub-routine it holds that $(Cx^t)_{i_t}<1$ which is similarly guaranteed in the dynamic setting. To argue it formally let $C^t$ refer to the state of constraint matrix $C$ at the time of the $t$-th \textsc{Whack} call. Fix some $t \in [T]$. The following line of inequalities in analogous with the proof of Lemma~\ref{cor:bound:weight:dyn}.

\begin{eqnarray*}
\norm{\hat{x}^{t+1}}_1 - \norm{\hat{x}^t}_1  =  \sum_{j \in [n]} \left( \left( \hat{x}^{t+1}\right)_j - \left( \hat{x}^{t}\right)_j \right) = \sum_{j\in [n]} \left( \hat{x}^t \right)_j \cdot \left( \epsilon \cdot \frac{C^t_{i_t, j}}{\lambda}\right)   =  \frac{\epsilon}{\lambda} \cdot \left( C^t \hat{x}^t \right)_{i_t} < \frac{\epsilon}{\lambda} \cdot \norm{\hat{x}^t}_1
\end{eqnarray*}
The last inequality follows since $\left( C^t x^t \right)_{i_t} < 1$ and $x^t := \hat{x}^t/\norm{\hat{x}^t}_1$. Rearranging the terms, we get:
\begin{equation}
\nonumber
\norm{\hat{x}^{t+1}}_1 \leq \left( 1 + \frac{\epsilon}{\lambda} \right) \cdot \norm{\hat{x}^t}_1
\end{equation}
As $T = \frac{\lambda \ln (n)}{\epsilon^2}$ and $\norm{\hat{x}^0}_1 = n$ we get the following:
$\norm{\hat{x}^t}_1 \leq \left(1 + \frac{\epsilon}{\lambda}\right)^{T} \cdot \norm{\hat{x}^0}_1 \leq n^{(1/\epsilon)}$ for all $t \in [T]$.

\paragraph{Proof of Theorem~\ref{th:MWU:static:approx:dyn}:} The extension of Theorem~\ref{th:MWU:static:approx} to the dynamic setting is also relatively straightforward. Observe, that at the initialization of our dynamic algorithm we simply run the static implementation. Therefore before any update occurs Theorem~\ref{th:MWU:static:approx:dyn} can be concluded from Theorem~\ref{th:MWU:static:approx}. We will first argue that if the dynamic algorithm returns vector $\tilde{x}^t$ then it must satisfy that $C \cdot \tilde{x}^t \geq \mathbb{1} \cdot (1 - \epsilon)$ and $|\tilde{x}^t|_1 \leq 1 + \epsilon$ even after some restricting updates have occurred.

Consider the starting point of a given phase at time $\tau$. At this point $\mathbb{1}^\top \tilde{x}^\tau = \mathbb{1}^\top \hat{x}^\tau /\norm{\hat{x}^\tau}_1= 1$ and we fix $W = \norm{\hat{x}^\tau}_1$. During the phase $\tilde{x}^t$ may only monotonously increase. However, if at any point $\norm{\hat{x}^t}_1 > (1 - \epsilon/2)^{-1} \cdot W$ a new phase begins. As $\tilde{x}^t = \hat{x}^t/W$ during the phase $\norm{\tilde{x}^t}_1 \leq (1-\epsilon/2)^{-1} \leq 1 + \epsilon$.

Fix some $i \in [m]$. At the start of each phase the algorithm ensures that $(C\cdot \tilde{x}^t)_i \geq 1$ through enforcing all currently violating constraints. As during a phase $\tilde{x}^t$ may only monotonously increase until further restricting updates occur to constraint $i$ it will remain satisfied. At all times the algorithm keeps an $\epsilon$-approximate estimate of $(C \cdot \tilde{x}^t)$. Whenever a restricting update occurs to some constraint element $C_{i,j}$ and we observe that out estimate of the constraint value falls bellow $1$ we enforce constraint $i$ until $(C \cdot \tilde{x}^t)_i \geq 1$. Due to us keeping an $\epsilon$-approximate estimate of $(C \cdot \tilde{x}^t)_i$ this will never allow it to fall bellow $1 - \epsilon$. Note that the value of $(C \cdot \tilde{x}^t)_i$ only depends on the $i$-th row of $C$ and $\tilde{x}^t$ therefore if a restricting update occurs to constraint $i$ and we decide not to enforce the constraint itself then the values of all other constraints remain unaffected.

It remains to argue that whenever the dynamic algorithm returns vector $y$ following a series of updates it holds that $|y|_1 = 1$ and $C^\top y \leq \mathbb{1} \cdot (1 + 4 \cdot \epsilon)$. Once the algorithm has returned such a vector $y$ it will no longer change its output. As $C$-s elements may only get reduced over time $y$ will maintain these properties. Here we refer to the proof of Lemma~\ref{lm:whack:a:mole:key} with slight deviations.  Recall the {\em experts} setting from \cite{AroraHK12} as presented in the proof Lemma~\ref{lm:whack:a:mole:key}. Further recall the statement of Lemma~\ref{lm:MWU}. 

We map the whack-a-mole algorithm to the experts setting in the same manner however now accounting for the evolving constraint matrix. Let $C^t$ stand for the state of $C$ at the time of the $t$-th call to \textsc{Whack}. Each covering constraint $j \in [n]$ corresponds to an expert, each iteration of the {\sc For} loop in Figure~\ref{fig:dual_algorithm} corresponds to a round,  the vector $\hat{x}^t$ corresponds to the weight vector $w^t$ at the start of round $t \in [T]$, and finally, the payoff for an expert $j \in [n]$ in round $t \in [T]$ is given by $\left(p^t\right)_j := (1/\lambda) \cdot C^t_{i_t, j}$.  

Since $C^t_{i_t, j} \in [0, \lambda]$, we have $\left| \left( p^t \right)_j \right| = \left( p^t \right)_j$ for all $j \in [n], t \in [T]$. Thus, Lemma~\ref{lm:MWU} implies that:
\begin{eqnarray}
\sum^{T}_{t=1} \left(p^{t}\right)^{\top} \cdot x^t  \geq \sum_{t=1}^T (1-\epsilon) \cdot \left(p^t\right)_j - \frac{\ln(n)}{\epsilon}, \text{ for all experts } j \in [n]. \label{eq:experts:1:dynamic}
\end{eqnarray}
Diving both sides of the above inequality by $T$, and then rearranging the terms, we get:
\begin{eqnarray}
\frac{(1-\epsilon)}{T} \cdot \sum_{t=1}^T  \left(p^t\right)_j \leq \frac{1}{T} \cdot \sum^{T}_{t=1} \left(p^{t}\right)^{\top} \cdot x^t  + \frac{\ln (n)}{\epsilon \cdot T}, \text{ for all experts } j \in [n]. \label{eq:experts:10:dynamic}
\end{eqnarray}
We next upper bound the right hand side (RHS) of~(\ref{eq:experts:10:dynamic}). Since the algorithm picks a violated covering constraint to whack in each round, we have: $\left(p^{t}\right)^{\top} \cdot x^t = (1/\lambda) \cdot  \left(C^t x^t\right)_{i_t} \leq (1/\lambda)$. Taking the average of this inequality across all the $T$ rounds, we get:
$(1/T) \cdot \sum^{T}_{t=1} \left(p^{t}\right)^{\top} \cdot x^t \leq (1/\lambda)$.
Since $T = \lambda \ln (n)/\epsilon^2$, we derive the following upper bound on the RHS of~(\ref{eq:experts:10:dynamic}).
\begin{eqnarray}
\frac{1}{T} \cdot \sum^{T}_{t=1} \left(p^{t}\right)^{\top} \cdot x^t  + \frac{\ln (n)}{\epsilon \cdot T} \leq \frac{1}{\lambda} \cdot (1+\epsilon). \label{eq:experts:11:dynamic}
\end{eqnarray}
We now focus our attention on the left hand side (LHS) of~(\ref{eq:experts:10:dynamic}). Fix any expert $j \in [n]$. We first express the payoff obtained by this expert at round $t \in [T]$ in terms of the vector $y^t$, and get: $\left( p^t \right)_j = \frac{1}{\lambda} \cdot C_{i_t, j} = \frac{1}{\lambda} \cdot \left( (C^t)^{\top} y^t\right)_j$.
Since $y = \frac{1}{T} \cdot \sum_{t=1}^T y^t$, the average payoff for the expert $j$ across all the $T$ rounds is given by:
\begin{eqnarray}
\frac{1}{T} \cdot \sum_{t=1}^T  \left(p^t\right)_j = \frac{1}{T} \cdot \frac{1}{\lambda} \cdot \sum_{t=1}^T \left( (C^t)^{\top} y^t\right)_j \geq \frac{1}{T} \cdot \frac{1}{\lambda} \cdot \sum_{t=1}^T \left( (C^T)^{\top} y^t\right)_j = \frac{1}{\lambda} \cdot \left( (C^T)^{\top} y\right)_j. \label{eq:experts:12:dynamic}
\end{eqnarray} 

In \ref{eq:experts:12:dynamic} the inequality holds as the elements of $C$ may only reduce over time. The inequality is the most significant deviation from the proof in the static setting. From~(\ref{eq:experts:10:dynamic}),~(\ref{eq:experts:11:dynamic}) and~(\ref{eq:experts:12:dynamic}), we get:
$$\frac{(1-\epsilon)}{\lambda} \cdot \left( (C^T)^{\top} y\right)_j \leq \frac{1}{\lambda} \cdot (1+\epsilon), \text{ and hence } \left( (C^T)^{\top} y\right)_j \leq (1+4\epsilon) \text{ for all } j \in [n].$$
The last inequality holds since $\epsilon < 1/2$. This concludes the proof of Theorem~\ref{th:MWU:static:approx:dyn}.

%% file: appendix-greedy-mwu-1.tex
\section{Missing Proofs from Section~\ref{sec:static}}
\label{sec:app:proof}
\begin{claim}

A coordinate $k$ is cheap if and only if $\langle \nabla f_p(x), e_k \rangle \leq (1 + \epsilon) \langle \nabla f_c(x), e_k \rangle$.

\end{claim}

\begin{proof}

Follows from definitions: $\langle \nabla f_p(x), e_k \rangle = \frac{\sum_{i \in [m_p]} e^{\eta \cdot P_i x} \cdot P(i,k)}{\sum_{i \in [m_p]} e^{\eta \cdot P_i x}}$, $\langle \nabla f_c(x), e_k \rangle = \frac{\sum_{i \in [m_c]} e^{\eta \cdot C_i x} \cdot C(i,k)}{\sum_{i \in [m_c]} e^{\eta \cdot C_i x}}$, hence if $\frac{\langle \nabla f_p(x), e_k \rangle}{\langle \nabla f_c(x), e_k \rangle} \leq 1 + \epsilon$ we have that $\frac{\sum_{i \in [m_p]} e^{\eta \cdot P_i x} \cdot P(i,k)}{\sum_{i \in [m_c]} e^{\eta \cdot C_i x} \cdot C(i,k)} \cdot \frac{\sum_{i \in [m_c]} e^{\eta \cdot C_i x}}{\sum_{i \in [m_p]} e^{\eta \cdot P_i x}} = \lambda(x,k) \cdot \frac{w_c(x)}{w_p(x)} \leq 1 + \epsilon$.

\end{proof}

\subsection{Proof of Property~\ref{prop:infeasible}}

\label{sec:proof:prop:infeasible}

\begin{proof}

The proof is analogous to the argument of Young \cite{Young2014nearly}. Let $y$ be a feasible solution to the linear program. Then we have that $\frac{\sum_{i = 1}^{m_p}(Cy)_i \cdot w_c(x,i)}{w_c(x)} \geq 1 \geq \frac{\sum_{i = 1}^{m_p} (Py)_i \cdot w_p(x,i)}{w_p(x)}$ as $w_c(x,i)/w_c(x)$ and $w_p(x,i)/w_p(x)$ just are describing distributions over the the constraints. This could be phrased as:

$$\red{y \cdot (C^T \cdot \frac{w_c(x)}{|w_c(x)|} - P^T \cdot \frac{w_p(x)}{|w_p(x)|})}$$

\red{Which as $y$ is non-negative implies that there is an $i$ such that $C_i^T \cdot \frac{w_c(x)}{|w_c(x)|} - P_i^T \cdot \frac{w_p(x)}{|w_p(x)|} \geq 0$ which is equivalent to saying $\lambda(x,i) \leq \frac{w_p(x)}{w_c(x)}$ meaning there is a cheap coordinate.}\thatchaphol{Not type-checked? $P$ is a matrix. $w_p(x,i)$ is a number.}\thatchaphol{I think it would be better to conclude why this proves Property 2.3.}

\end{proof}

\subsection{Proof of Property~\ref{prop:stop}}

\label{sec:proof:prop:stop}

For the convenience of the reader we will restate the statement precisely: given $\epsilon$ is sufficiently small such that $\epsilon \leq 1/10$ we have that:

$$f_p(x + \delta \cdot e_k) - f_p(x) \leq (1 + 10 \cdot \epsilon) \cdot (f_c(x + \delta \cdot e_k) - f_c(x)) \leq (1 + 10 \cdot \epsilon) \cdot \epsilon \leq 2 \cdot \epsilon$$

\begin{proof}

Note, that if $\epsilon \leq 1/10$ then $exp(-\log(m_c + U/L)/\epsilon) \leq exp(-2 \cdot \ln(m_c + U/L) - 1/\epsilon + 2) \leq \frac{\epsilon \cdot L}{U \cdot m_c}$.

Let $k$ be a cheap coordinate which the algorithms boosts by $\delta$ and let $x$ stand for the value of the output vector before boosting. If coordinate $k$ is boosted we have that $\langle \nabla f_p(x) , e_k \rangle \leq (1 + \epsilon) \cdot \langle \nabla f_c(x) , e_k \rangle$. We will show that $\langle \nabla f_p(x + \delta \cdot e_k) , e_k \rangle \leq \langle \nabla f_p(x) , e_k \rangle \cdot (1 + O(\epsilon))$ and $\langle \nabla f_c(x + \delta \cdot e_k) , e_k \rangle \geq \langle \nabla f_p(x) , e_k \rangle \cdot (1 - O(\epsilon))$. As $f_c(x)$ and $f_p(x)$ monotonously reduces and increases respectively with $x$ this implies that during the boost it must hold that $\langle \nabla f_p(x) , e_k \rangle \leq (1 + O(\epsilon)) \cdot \langle \nabla f_c(x) , e_k \rangle$ which in turn implies that $f_p(x + \delta \cdot e_k) - f_p(x) \leq (f_c(x + \delta \cdot e_k) - f_c(x)) \cdot (1 - O(\epsilon))$ implying one half of the claim.

\begin{eqnarray}
\langle \nabla f_p(x + \delta \cdot e_k) , e_k \rangle  & = &  \frac{\sum_{i \in [m_p]} e^{\eta \cdot P_i (x + \delta \cdot e_k)} \cdot P(i,k)}{\sum_{i \in [m_p]} e^{\eta \cdot P_i (x + \delta \cdot e_k)}}\nonumber\\
& \leq & \frac{\sum_{i \in [m_p]} e^{\epsilon + \eta \cdot P_i x} \cdot P(i,k)}{\sum_{i \in [m_p]} e^{\eta \cdot P_i x}} \label{eq:prop:stop:1}\\
& = & e^\epsilon \cdot \langle \nabla f_p(x) , e_k \rangle \nonumber\\
& \leq & (1 + 2 \cdot \epsilon) \cdot \langle \nabla f_p(x) , e_k \rangle \label{eq:prop:stop:2}
\end{eqnarray} 

Inequality~\ref{eq:prop:stop:1} follows from the definition of $\delta$ and Inequality~\ref{eq:prop:stop:2} applies for any $\epsilon \leq 1$.

Define $I = \{ i \in [m_c] : C_i x > 2$ and $A = \{ i \in [m_c] : C_i x \leq 2 \}$ representing the inactive and active constraints of $C$.\thatchaphol{P is used for packing already. How about $I$ for inactive?} 

\begin{eqnarray}
\langle \nabla f_c(x) , e_k \rangle  & = &  \frac{\sum_{i \in [m_c]} e^{-\eta \cdot C_i x} \cdot C(i,k)}{\sum_{i \in [m_c]} e^{-\eta \cdot C_i x}}\nonumber\\
& = & \frac{\sum_{i \in A} e^{-\eta \cdot C_i x} \cdot C(i,k) + \sum_{i \in I} e^{-\eta \cdot C_i x} \cdot C(i,k)}{\sum_{i \in A} e^{-\eta \cdot C_i x} + \sum_{i \in I} e^{-\eta \cdot C_i x}} \nonumber\\
& \leq & \frac{\sum_{i \in A} e^{-\eta \cdot C_i x} \cdot C(i,k) + \sum_{i \in I} e^{-\eta \cdot C_i x} \cdot C(i,k)}{\sum_{i \in A} e^{-\eta \cdot C_i x}} \nonumber\\
& \leq & \frac{\sum_{i \in A} e^{-\eta \cdot C_i x} \cdot C(i,k) + \sum_{i \in I} e^{- 2 \cdot \eta} \cdot C(i,k)}{\sum_{i \in A} e^{-\eta \cdot C_i x}} \nonumber\\
& \leq & \frac{\sum_{i \in A} e^{-\eta \cdot C_i x} \cdot C(i,k) + e^{-\eta} \cdot m_c \cdot e^{-\frac{\log (m_c + U/L)}{\epsilon}} \cdot \max_{i \in [m_c]} \{C(i,k)\}}{\sum_{i \in A} e^{-\eta \cdot C_i x}} \nonumber\\
& \leq & \frac{\sum_{i \in A} e^{-\eta \cdot C_i x} \cdot C(i,k) + e^{-\eta} \cdot L \cdot \epsilon}{\sum_{i \in A} e^{-\eta \cdot C_i x}} \label{eq:prop:stop:3} \\
& \leq & \frac{\sum_{i \in A} e^{-\eta \cdot C_i x} \cdot C(i,k)}{\sum_{i \in A} e^{-\eta \cdot C_i x}} \cdot (1 + \epsilon) \label{eq:prop:stop:4}
\end{eqnarray}

Equation~\ref{eq:prop:stop:3} follows if $\epsilon$ is selected small enough such that $e^{-\frac{\log(m_c + U/L)}{\epsilon}}\leq \frac{\epsilon \cdot L}{U \cdot m_c}$. Equation~\ref{eq:prop:stop:4} follows from the fact that in order for a boost to occur there must be at least one unsatisfied covering constraint, hence $\sum_{i \in A} e^{-\eta \cdot C_i x} \cdot C(i,k) \geq e^{-\eta} \cdot L$.

\begin{eqnarray}
\langle \nabla f_c(x + \delta \cdot e_i) , e_k \rangle  & = &  \frac{\sum_{i \in [m_c]} e^{-\eta \cdot C_i (x + \delta \cdot e_i)} \cdot C(i,k)}{\sum_{i \in [m_c]} e^{-\eta \cdot C_i (x + \delta \cdot e_i)}} \nonumber \\
& = & \frac{\sum_{i \in A} e^{-\eta \cdot C_i (x + \delta \cdot e_i)} \cdot C(i,k) + \sum_{i \in I} e^{-\eta \cdot C_i (x + \delta \cdot e_i)} \cdot C(i,k)}{\sum_{i \in A} e^{-\eta \cdot C_i (x + \delta \cdot e_i)} + \sum_{i \in I} e^{-\eta \cdot C_i (x + \delta \cdot e_i)}} \nonumber \\
& \geq & \frac{\sum_{i \in A} e^{-\eta \cdot C_i (x + \delta \cdot e_i)} \cdot C(i,k)}{\sum_{i \in A} e^{-\eta \cdot C_i (x + \delta \cdot e_i)} + \sum_{i \in P} e^{-\eta \cdot C_i (x + \delta \cdot e_i)}} \nonumber \\
& \geq & \frac{\sum_{i \in A} e^{-\eta \cdot C_i (x + \delta \cdot e_i)} \cdot C(i,k)}{\sum_{i \in A} e^{-\eta \cdot C_i (x + \delta \cdot e_i)} + m_c \cdot e^{- 2 \cdot \eta}} \nonumber \\
& \geq & \frac{\sum_{i \in A} e^{-\eta \cdot C_i (x + \delta \cdot e_i)} \cdot C(i,k)}{\sum_{i \in A} e^{-\eta \cdot C_i (x + \delta \cdot e_i)}} \cdot \frac{1}{1 + \epsilon} \label{eq:prop:stop:5} \\
& \geq & \frac{\sum_{i \in A} e^{-\eta \cdot C_i x} \cdot C(i,k)}{\sum_{i \in A} e^{-\eta \cdot C_i x}} \cdot \frac{e^{-\epsilon}}{1 + \epsilon} \label{eq:prop:stop:6} 
\end{eqnarray}

Here Equation~\ref{eq:prop:stop:5} similarly follows from the fact that $\sum_{i \in A} e^{-\eta \cdot C_i (x + \delta \cdot e_i)} \geq e^{-\eta - \epsilon}$ (by the choice of $\delta$) and by the definition of $\eta$. Equation~\ref{eq:prop:stop:6} follows from the definition of $\delta$. Therefore, we get that $\langle \nabla f_c(x + \delta \cdot e_i) , e_k \rangle \cdot (1 + 4 \cdot \epsilon) \geq \langle \nabla f_c(x + \delta \cdot e_i) , e_k \rangle \cdot {\frac{(1+\epsilon)^2}{e^{-\epsilon}}} \geq \langle \nabla f_c(x) , e_k \rangle$ (given $\epsilon < 1/5$).

As $\langle \nabla f_p(x + \delta \cdot e_i) , e_k \rangle \leq (1 + 2 \cdot \epsilon) \langle \nabla f_p(x) , e_k \rangle$ and $\langle \nabla f_p(x) , e_k \rangle \leq \langle \nabla f_x(x) , e_k \rangle \cdot (1 + \epsilon)$ we have that $\langle \nabla f_p(x + \delta \cdot e_i) , e_k \rangle \leq \langle \nabla f_c(x + \delta \cdot e_i) , e_k \rangle \cdot (1 + \epsilon)\cdot (1 + 2\cdot \epsilon)\cdot(1 + 4\epsilon) \leq \langle \nabla f_c(x + \delta \cdot e_i) , e_k \rangle \cdot (1 + 10 \cdot \epsilon)$ (given $\epsilon \leq 1/10$).

As both $f_c(x)$ and $f_p(x)$ are monotone in $x$ we get that $f_p(x + \delta \cdot e_k) - f_p(x) \leq (f_c(x + \delta \cdot e_k) - f_c(x)) \cdot (1 + 10 \cdot \epsilon)$. Note that $f_c(x) \geq \min_{i \in [m_c]} C_i x$ and $f_c(x + \delta \cdot e_k) + \epsilon \leq \min_{i \in [m_c]} C_i (x + e_k \cdot \epsilon) \leq \min_{i \in [m_c]} C_i x + \epsilon$ therefore $f_c(x + \delta \cdot e_k) - f_c(x) \leq \epsilon$. Therefore, we can conclude the original inequality.

\end{proof}

%% file: appendix-greedy-mwu-2.tex
\section{Complete Algorithm from Section~\ref{sec:dynamic}}
\label{sec:app:algo}

\thatchaphol{Can we add explanation about the complication that we hid under the rug in the main body? This is so that the readers know that they missed exactly and/or understand why this section is long.}

In Section~\ref{sec:prob}, we presented an overview of our dynamic algorithm for positive LPs under relaxing updates. But we glossed over some implementation details to highlight the key ideas. For instance, we did not specify the dynamic data structure that will be needed to find the appropriate $\delta$ during a call to the {\sc Boost}$(x, k)$ subroutine (see step 1 in Figure~\ref{fig:boost}). Furthermore, in Section~\ref{subsec:dynamic:covering} we focused only on relaxing updates to the covering constraints, and we  sketched in Section~\ref{subsec:dynamic:full} how to deal with the relaxing updates to the packing constraints. Accordingly, to address these gaps, in this section we present a full, comprehensive version of our dynamic algorithm for positive LPs under relaxing updates.

\paragraph{Preliminaries:} Let $P \in \Re_{\geq 0}^{m_p \times n}$ and $C \in \Re_{\geq 0}^{m_c \times n}$ be input matrices to the linear program. Let $N$ stand for the number of entries in $P$ and $C$ which are non-zero at any point during the run of the algorithm. Let $U$ and $L$ stand for upper and lower bounds on the elements of the constraint matrices $P$, $C$ at all times, respectively. Let $t$ stands for the total number of relaxing updates progressed by the dynamic algorithm.

Let $R^P_i : i \in [m_p]$, $R^C_j : j \in [m_c]$ refer to the number of elements in the $i$-th row of $P$ and the $j$-th row of $C$ respectively which take a non-zero value at any time. Similarly, define $N^P_k : k \in [n]$ and $N^C_k : k \in [n]$ to represent the number of elements in the $k$-th column of $P$ and $C$ respectively taking a non-zero value at any point.

\begin{theorem}

\label{thm:mixed}

There is an algorithm for maintaining a certificate of infeasability or returning an $\epsilon$-approximate solution for a relaxing dynamic positive linear program in $O(\frac{N \cdot \log^2(m_p + m_c + U/L)}{\epsilon^2} + t)$ deterministic time (for $\epsilon < 1/200$).

\end{theorem}

We will dedicate this entire section to proving Theorem~\ref{thm:mixed}. For sake of simplicity we will first show how we can handle a dynamic linear program undergoing solely entry wise updates. In Subsection~\ref{subsec:app:mixed:translation} we will show to extend the algorithm for translation updates. We will make the following assumption on the distribution of inputs:

\begin{assumption}

\label{as:mixed:range}

Non-zero elements of input matrices $P$ and $C$ lie in the range of $[1/poly(N), poly(N)]$ at all times.

\end{assumption}

\subsection{Variables}

Define $P^*$, $C^*$ and $X^*$ as described in Section~\ref{subsec:dynamic:full}:

\begin{itemize}
\item  $P^* \in \mathbb{R}_{\geq 0}^{m_p \times (n+1)}$, where $P^*(i, k) = P(i, k)$ and $P^*(i, n+1) \geq 0$ for all $i \in [m_p], k \in [n]$. 
\item  $C^* \in \mathbb{R}_{\geq 0}^{m_c \times (n+1)}$, where $C^*(j, k) = C(j, k)$ and $C^*(j, n+1) = 0$ for all $j \in [m_c], k \in [n]$.
\item $x^* \in \mathbb{R}_{\geq 0}^{n+1}$ and $x$ refers to the first $n$ coordinates of $x^*$
\end{itemize}

From now on we will refer to the linear program defined by $Px \leq \mathbb{1}$, $CX \geq \mathbb{1}$ as the original linear program and to the linear program defined by $P^* x^* \leq \mathbb{1}$, $C^* x^* \geq \mathbb{1}$ as the extended linear program. Use variable $i$ and $j$ to refer to row indices of $P$ and $C$ respectively, and use $k$ for referring to column indices of $P$ or $C$.

\begin{observation}

\label{obs:mixed:equivalance}

If $x$ is a feasible solution to the original linear program then $\{x, 1\}$ ($x$ extended with a $1$) is a valid solution to the extended linear program. If $x^*$ is a feasible solution to the extended linear program then the first $n$ coordinates of $x^*$ form a feasible solution to the original linear program. Hence, the original linear program is feasible if and only if the extended linear program is feasible.

\end{observation}

Due to Observation~\ref{obs:mixed:equivalance} we will focus only on the extended linear program in this section.

\paragraph{Data-Structure:} We will initially define the following variables:

\begin{equation}
  \label{eq:mixed:eta}
  \eta = \frac{\log(m_c + m_p + U/L)}{\epsilon}
\end{equation}

\begin{itemize}
    \item $w_c(x^*,i) = 1$ for $i \in [m_c]$ representing covering constraint weights
    \item $w_p(x^*,j) = 1$ for $j \in [m_p]$ representing packing constraint weights
    \item $w_c(x^*) = \sum_{i \in [m_c]} w_c(x^*, i)$ and $w_p(x^*) = \sum_{j \in [m_p]} w_p(x^*,j)$
    \item $\lambda(x^*)_0 = \frac{w_p(x^*)}{w_c(x^*)}$
\end{itemize}

For the purposes of the analysis we will define $\lambda(x^*,k) = \forall k \in [n]: \lambda(x^*,k) = \frac{\sum_{j \in [m_p]} w_p(x^*,j) \cdot P^*(j,k)}{\sum_{i \in [m_c]} w_c(x^*,i) \cdot C^*(i,k)}$ representing the price of each coordinate. Note that after initialization the algorithm will not keep track of exact $\lambda(x^*,k)$ values. Similarly, for the purposes of the algorithm description and analysis define but do not maintain 

$$\forall k \in [n] : \delta_k = \max_\delta \{\max\{\max_{i \in [m_p]} P^*(i,k) \cdot \delta, \max_{j \in [m_c] | C^*_j x^* < 2} C^*(j,k) \cdot \delta\} = \epsilon/\eta\}$$

\paragraph{Approximate Variables:} The algorithm will initialize the following approximate variables: 

\begin{itemize}
    \item $\forall i \in [m_p] : \hat{w_p}(x^*,i) = w_p(x^*,i)$
    \item $\forall j \in [m_c] : \hat{w_c}(x^*,j) = w_c(x^*,j)$
    \item $\hat{w}_p(x^*) = w_p(x^*)$
    \item $\hat{w}_c(x^*) = w_c(x^*)$
    \item $\forall k \in [n]: \hat{\lambda}(x^*,k) = \lambda(x^*,k)$
    \item $\hat{\lambda}_0(x^*) = \lambda_0(x^*)$
\end{itemize}

\begin{invariant}

\label{inv:mixed}

The following inequalities are simultaneously satisfied:

\begin{itemize}
    \item $\forall i \in [m_c]: w_c(x^*,i) \leq \hat{w_c}(x^*,i) \leq w_c(x^*,i) \cdot (1 + \epsilon)$
    \item $\forall i \in [m_p]: w_p(x^*,i) \geq \hat{w_p}(x^*,i) \geq w_p(x^*,i) \cdot (1 - \epsilon)$
    \item $\forall k \in [n]: \hat{\lambda}(x^*,k) = \frac{\sum_{j \in [m_p]} \hat{w_p}(x^*,j) \cdot P^*(j,k)}{\sum_{i \in [m_c]} \hat{w_c}(x^*,i) \cdot C^*(i,k)}$ and $\hat{\lambda}_0(x^*) \geq \lambda_0(x^*) \cdot (1 - \epsilon)$
\end{itemize}

\end{invariant}

\subsection{Algorithm}

\paragraph{\textit{Iterate} subroutine:} Initially the algorithm calls the \textit{Iterate()} sub-routine which initializes $\hat{\lambda}_0(x^*) = \lambda_0(x^*)$ and iterates over all coordinates $k \in [n]$. \textit{Iterate()} boosts any coordinate $k$ satisfying $\hat{\lambda}(x^*,k) \leq \lambda_0(x^*) \cdot (1 + 5 \cdot \epsilon)$ until $\hat{\lambda}(x^*,k) > \lambda_0(x^*) \cdot (1 + 5 \cdot \epsilon)$. If due to the boosts $\lambda_0(x^*)$ has increased enough such that $\lambda_0(x^*) \cdot (1 - \epsilon) > \hat{\lambda}_0(x^*)$ \textit{Iterate()} is called recursively.

\paragraph{\textit{Boost} subroutine:} A call to \textit{Boost($k$)} increases the $k$-th coordinate of $x^*$. The algorithm selects an increment of $\delta$ such that all packing constraints and all not 'well' satisfied covering constraints are increased by at most $\epsilon/\eta$, that is $\delta = \max_\delta \{ \max \{\max_i P*(i,k) \cdot \delta, \max_{j | C^*_j < 2}C^*(j,k) \cdot \delta \} = \epsilon/\eta\}$. The algorithm increments $x^*$ by $e_k \cdot \delta$ and afterwards proceeds to restore the affected variables $w_c(x^*,i),w_p(x^*,j),w_c(x^*),w_p(x^*)$ and $\lambda_0(x^*)$ precisely. Afterwards \textit{UpdatePWeights, UpdateCWeights} subroutines are called to ensure the correctness of Invariant~\ref{inv:mixed}.

\paragraph{\textit{UpdatePweights, UpdateCWeights} subroutines:} These subroutines are used to update $\hat{w_p}, \hat{w_c}$ and $\hat{\lambda}$ values respectively. Without loss of generality assume that \textit{UpdateCWeights($k$)} was called. The algorithm checks for all $i \in [m_P]$ such that the $i$-th element of the $k$-th column of $P^*$ is not empty ($P*(i,k)>0$) weather $\hat{w}_p(x^*,i) < w_p(x^*,i) \cdot (1 - \epsilon)$ as these are the $w_p(x^*,i)$ values which have increased due to the boosting coordinate $k$. For rows where the slack between the approximate and actual values became too significant the algorithm readjusts the approximate parameter ($\hat{w}_p(x^*,i) \leftarrow w_p(x^*,i)$) and adjusts all values of $\hat{\lambda}(x^*,k')$ which might have been affected by this change (that is the ones corresponding to columns of $P^*$ where the $i$-th row is non empty, that is $P^*(i,k')>0$).

\paragraph{\textit{UpdateP} subroutine:} A call \textit{UpdateP($i,k,\Delta$)} handles a reduction of $\Delta$ to $P^*(i,k)$. After updating $P*(i,k)$ the algorithm increments $P*(i,n+1)$ with $\Delta \cdot x^*_i$ to guarantee that $w_p(x^*, i)$ remains the same after increasing $P^*(i,k)$. This way the only $\hat{\lambda}$ value which is affected is $\hat{\lambda}(x^*,k)$ which the subroutine updates to ensure the correctness of Invariant~\ref{inv:mixed}. Due to the update of $\hat{\lambda}(x^*,k)$ coordinate $k$ might become cheap (that is $\hat{\lambda}(x^*,k) \leq \lambda_0(x^*) \cdot (1 + 5 \cdot \epsilon)$) so the subroutine boosts $k$ until it is no longer cheap. If some boosts occur then $\lambda_0(x^*)$ might reduce enough such that $\hat{\lambda}_0(x^*) < \lambda_0(x^*) \cdot (1 - \epsilon)$ in which case a new call to \textit{Iterate} is made to ensure the correctness of Invariant~\ref{inv:mixed}.

\begin{algorithm}
    
	\SetKwInput{KwInput}{Input}                
	\SetKwInput{KwOutput}{Output}             
	\DontPrintSemicolon
	\caption{Main subroutine of the algorithm}
	
	\KwInput{$C^* \in \Re_{\geq 0}^{m_c \times n}, P^* \in \Re_{\geq 0}^{m_p \times n}$}
	\KwOutput{$x^* \in \Re_{\geq 0}^{n+1}|C^*x^* \geq \mathbb{1}, P^*x^* \leq \mathbb{1} \cdot (1 + 200 \cdot \epsilon)$}
	
	\SetKwFunction{FMain}{Main}
	\SetKwFunction{FIterate}{Iterate}
	
	\label{alg:mixed:main}
	
	\SetKwProg{Fn}{Function}{:}{\KwRet}
	\Fn{\FIterate{}}{
	    $\hat{\lambda}_0(x^*) \leftarrow \lambda_0(x^*)$ \;
		\For{$k \in [n]$}
		{
		    \While{$\hat{\lambda}(x^*,k) \leq \lambda_0(x^*) \cdot (1 + 5\cdot \epsilon)$}
		    {
		        Boost(k) \atcp{See Algorithm~\ref{alg:mixed:boost}}
		    }
		}
		\If{$\hat{\lambda}_0(x^*) < \lambda_0(x^*) \cdot (1 - \epsilon)$}
		{
		    Iterate()
		}
	}
	
	\SetKwProg{Fn}{Function}{:}{\KwRet}
	\Fn{\FMain{}}{
		Iterate() \;
	}

\end{algorithm}

\begin{algorithm}
    
	\SetKwInput{KwInput}{Input}                
	\SetKwInput{KwOutput}{Output}             
	\DontPrintSemicolon
	\caption{Boosting of coordinate $k$}
	
	\SetKwFunction{FBoost}{Boost}

	\label{alg:mixed:boost}
	
	\SetKwProg{Fn}{Function}{:}{\KwRet}
	\Fn{\FBoost{k}}{
	$\delta \leftarrow$ sample in $[\delta_k/4, \delta_k]$ \atcp{See Section~\ref{sec:mixed:findincrement}}
	$x^* \leftarrow x^* + e_k \cdot \delta$ \atcp{$e_k$ is the unit $k$-th coordinate vector}
	\If{$C^* x^* \geq \mathbb{1}$}
	{
	    Return $x^*$ \;
	}
	\For{$\forall i \in [m_p] | P^*(i,k) > 0$}
	{
	    $w_p(x^*, i) \leftarrow \exp(\eta \cdot P^*_i x^*)$ \;
	}
	$w_p(x^*) \leftarrow \sum_{i \in [m_p]} w_p(x^*,i)$ \atcp{Updated in $O(N^P_k)$ time}
	\For{$\forall j \in [m_c] | C^*(j,k) > 0$}
	{
	    $w_c(x^*, j) \leftarrow \exp(- \eta \cdot C^*_j x^*)$ \;
	}
	$w_c(x^*) \leftarrow \sum_{j \in [m_c]} w_c(x^*,j)$ \atcp{Updated in $O(N^C_k)$ time}
	$\lambda_0(x^*) = \frac{w_p(x^*)}{w_c(x^*)}$ \;
	UpdatePWeights(k) \atcp{See Algorithm~\ref{alg:mixed:updatePweights}}
	UpdateCWeights(k) \atcp{See Algorithm~\ref{alg:mixed:updateCweights}}
	}

\end{algorithm}

\begin{algorithm}
    
	\SetKwInput{KwInput}{Input}                
	\SetKwInput{KwOutput}{Output}             
	\DontPrintSemicolon
	\caption{The subroutine updating packing weights}
	
	\SetKwFunction{FUpdatePWeights}{UpdatePWeights}
	
	\label{alg:mixed:updatePweights}
	
	\SetKwProg{Fn}{Function}{:}{\KwRet}
	\Fn{\FUpdatePWeights{k}}{
	    \For{$j \in [m_p]|P^*_{j,k} > 0$}
	    {
	        \If{$\hat{w_p}(x^*, j) < w_p(x^*,j) \cdot (1 - \epsilon)$}
	        {
	            $\hat{w_p}(x^*,j) = w_p(x^*,j)$ \;
	            \For{$k' \in [n]|P^*_{j,k'} > 0$}
	            {
	                $\hat{\lambda}(x^*,k') = \frac{\sum_{j' \in [m_p]} \hat{w_p}(x^*,j') \cdot P^*(j',k')}{\sum_{i \in [m_c]} \hat{w_c}(x^*,i) \cdot C^*(i,k')}$
	            }
	        }
	    }
	}

\end{algorithm}

\begin{algorithm}
    
	\SetKwInput{KwInput}{Input}                
	\SetKwInput{KwOutput}{Output}             
	\DontPrintSemicolon
	\caption{The subroutine updating covering weights}
	
	\SetKwFunction{FUpdateCWeights}{UpdateCWeights}
	
	\label{alg:mixed:updateCweights}
	
	\SetKwProg{Fn}{Function}{:}{\KwRet}
	\Fn{\FUpdateCWeights{i}}{
	    \For{$i \in [m_c]|C^*_{i,k} > 0$}
	    {
	        \If{$\hat{w_c}(x^*,i) > w_c(x^*,i) \cdot (1 + \epsilon)$}
	        {
	            $\hat{w_c}(x^*,i) = w_c(x^*,i)$ \;
	            \For{$k' \in [n]|C^*_{i,k'} > 0$}
	            {
	                $\hat{\lambda}(x^*,k') = \frac{\sum_{j \in [m_p]} \hat{w_p}(x^*,j) \cdot P^*(j,k')}{\sum_{i' \in [m_c]} \hat{w_c}(x^*,i') \cdot C^*(i',k')}$
	            }
	        }
	    }
	}

\end{algorithm}

\begin{algorithm}
	\SetKwInput{KwInput}{Input}               
	\DontPrintSemicolon
	\label{alg:mixed:packingupdate}
	\KwInput{$(j,k,\Delta)$: $P^*(j,k)$ decreased by $\Delta$}
	\caption{Subroutine handling updates to the packing constraint matrix $P^*$}
	\SetKwFunction{FUpdateP}{UpdateP}
	
	\SetKwProg{Fn}{Function}{:}{\KwRet}
	\Fn{\FUpdateP{$j,k,\Delta$}}{
	$P^*(i,k) \leftarrow P^*(i,k) - \Delta$ \;
	$P^*(i,n+1) = P^*(i,n+1) + \Delta \cdot x^*_i$ \;
	$\hat{\lambda}(x^*,k) = \frac{\sum_{i' \in [m_p]} \hat{w_p}(x^*,i') \cdot P^*(i',k)}{\sum_{j \in [m_c]} \hat{w_c}(x^*,j) \cdot C^*(j,k)}$ \atcp{Updated in $O(1)$ time}
	\While{$\hat{\lambda}(x^*,k) \leq \lambda_0(x^*) \cdot ( 1 + 5 \cdot \epsilon)$}{
        Boost($k$) \atcp{See Algorithm~\ref{alg:mixed:boost}}
    }
    
    \If{$\hat{\lambda}_0(x^*) < \lambda_0(x^*) \cdot (1 - \epsilon)$}{
        Iterate() \atcp{See Algorithm~\ref{alg:mixed:main}}
    }

}

\end{algorithm}

\begin{algorithm}
	\SetKwInput{KwInput}{Input}               
	\DontPrintSemicolon
	
	\label{alg:mixed:coveringupdate}
	\caption{Subroutine handling updates to the covering constraint matrix $C^*$}
	\KwInput{$(j,k, \Delta)$: $C^*(j,k)$ increased by $\Delta$}
	
	\SetKwFunction{FUpdateC}{UpdateC}
	
	\SetKwProg{Fn}{Function}{:}{\KwRet}
	\Fn{\FUpdateC{$j,k,\Delta$}}{
	
	$C^*(j,k) \leftarrow C^*(j,k) + \Delta$ \;
	
	\If{$C^*x^* \geq \mathbb{1}$}
	{
	    Return $x^*$ \;
	}
	
	$w_c(x^*,j) \leftarrow \exp(-\eta \cdot C^*_j x^*)$ \;
	$w_c(x^*) \leftarrow \sum_{j \in [m_c]} w_c(x^*,j)$ \atcp{Updated in $O(1)$ time}
	$\lambda_0(x^*) = \frac{w_p(x^*)}{w_c(x^*)}$ \;
	\If{$w_c(x^*,j) < \hat{w_c}(x^*, j) \cdot (1 - \epsilon)$}{
	    $\hat{w_c}(x^*, j) = w_c(x^*,j)$ \;
        \For{$k' \in [n]|C^*_{j,k'}>0$}{
            $\hat{\lambda}(x^*,k') = \frac{\sum_{i \in [m_p]} \hat{w_p}(x^*,i) \cdot P^*(i,k')}{\sum_{j' \in [m_c]} \hat{w_c}(x^*,j') \cdot C^*(j',k')}$ \;
        }
        \For{$k' \in [n]|C^*_{j,k'}>0$}{
            \While{$\hat{\lambda}(x^*,k') \leq \lambda_0(x^*) \cdot ( 1 + 5 \cdot \epsilon)$}{
                Boost(k') \atcp{See Algorithm~\ref{alg:mixed:boost}}
            }
        }
	}
	$\hat{\lambda}(x^*,k) = \frac{\sum_{i \in [m_p]} \hat{w_p}(x^*,i) \cdot P^*(i,k)}{\sum_{j' \in [m_c]} \hat{w_c}(x^*,j') \cdot C^*(j',k)}$\;
	\While{$\hat{\lambda}(x^*,k) \leq \lambda_0(x^*) \cdot ( 1 + 5 \cdot \epsilon)$}{
        Boost($k$) \atcp{See Algorithm~\ref{alg:mixed:boost}}
    }
    \If{$\hat{\lambda}_0(x^*) < \lambda_0(x^*) \cdot (1 - \epsilon)$}{
        Iterate() \atcp{See Algorithm~\ref{alg:mixed:main}}
    }
}

\end{algorithm}

\subsection{Correctness}

\begin{claim}

\label{cl:mixed:inv}

Invariant~\ref{inv:mixed} is always satisfied by Algorithm~\ref{alg:mixed:main}.

\end{claim}

\begin{proof}

At initialization the invariant clearly holds by the definition of the approximate variables as their exact counterparts. To prove the correctness of the invariant we will consider the three instances where variables are updated and argue inductively that they preserve Invariant~\ref{inv:mixed}.

\paragraph{\textit{Boost} subroutine (Algorithm~\ref{alg:mixed:boost}):} when \textit{Boost($k$)} is called then elements of $w_p(x^*,i) | P^*(i,k) > 0$ and $w_p(x^*,j) | C^*(j,k) > 0$, $w_c(x^*)$, $w_p(x^*)$ and $\lambda_0(x^*)$ are updated. Afterwards \textit{UpdatePWeights($k$)} and \textit{UpdateCWeights($k$)} are called. These subroutines check variables $\hat{w}_p(x^*,i) | P^*(i,k) > 0$ and $\hat{w}_p(x^*,j) | C^*(j,k) > 0$ and update them to enforce the requirements of Invariant~\ref{inv:mixed}. \textit{Boost($k$)} might be called by subroutines \textit{Iterate, UpdateP} and \textit{UpdateC}. Each of these subroutines call \textit{Iterate} if the last requirement of Invariant~\ref{inv:mixed} ($\hat{\lambda}_0(x^*) < \lambda_0 \cdot (1 - \epsilon)$ is violated. As Invariant~\ref{inv:mixed} can't compile (as it calls itself) until $\hat{\lambda}_0(x^*) \geq \lambda_0 \cdot (1 - \epsilon)$ all of Invariant~\ref{inv:mixed} will be satisfied after a boost operation.

\paragraph{\textit{UpdateP} subroutine (Algorithm~\ref{alg:mixed:updatePweights}):} \textit{UpdateP($i,k,\Delta$)} doesn't affect constraint weights by define. The only variable it needs to update is $\hat{\lambda}(x^*,k)$ because of the change in $P^*(i,k)$. As argued previously, Invariant~\ref{inv:mixed} is restored by the \textit{Boost} calls inside of \textit{UpdateP}.

\paragraph{\textit{UpdateC} subroutine (Algorithm~\ref{alg:mixed:updateCweights}):} \textit{UpdateC} makes changes to $w_c(x^*,j),w_c(x^*)$ and $\lambda_0(x^*)$ and updates the approximate variables accordingly to fit Invariant~\ref{inv:mixed}. As argued previously the invariant is restored after proceeding \textit{Boost} calls.

\end{proof}

\begin{claim}

\label{cl:mixed:infeasability}

Whenever the algorithm halts without returning an answer the extended LP is infeasible.

\end{claim}

\begin{proof}

By Claim~\ref{cl:mixed:inv} we can assume that Invariant~\ref{inv:mixed} is satisfied whenever the algorithm halts without returning an answer. Following the proof of Property~\ref{prop:infeasible} from Section~\ref{sec:proof:prop:infeasible} if at any time for all $k \in [n]$ we have that $\lambda(x^*,k) < \lambda_0(x^*)$ then the linear program is infeasible. We will argue that Invariant~\ref{inv:mixed} implies $\forall k \in [n] : \lambda(x^*,k) < \lambda_0(x^*)$. Fix $k \in [n]$ for this proof. 

\medskip

By Invariant~\ref{inv:mixed} we always have that $\hat{\lambda}(x^*,k) \geq \lambda(x^*,k)$. Say whenever the algorithm checked if $\hat{\lambda}(x^*,k) \leq \lambda_0(x^*) \cdot (1 + 5 \cdot \epsilon)$ that the algorithm has checked the price of coordinate $k$. After a coordinates price is checked it is boosted until $\lambda(x^*,k) \geq \hat{\lambda}(x^*,k) > \lambda_0(x^*) \cdot (1 + 5 \cdot \epsilon)$. Each coordinates price is checked at least once at initialization. Let $\lambda_0^-(x^*)$ and $\lambda_0^+(x^*)$ represent the value of $\lambda_0(x^*)$ after the last time coordinate $k$-s price was checked and at compilation respectively. Let $\hat{\lambda}_0(x^*)^*$ represent the value of $\hat{\lambda}_0(x^*)$ at both of these time points (note that if $\hat{\lambda}_0(x^*)$ would have been updated between the last price checking of coordinate $k$ and compilation a new call to \textit{Iterate()} would have been made, hence $k$-s price would have been checked again). By Invariant~\ref{inv:mixed} we have that $\hat{\lambda}_0(x^*)^* \leq \lambda_0^-(x^*)$ and $\hat{\lambda}_0(x^*)^* \geq \lambda_0^+(x^*) \cdot (1 - \epsilon)$. Therefore, $\lambda_0^+(x^*) \leq \frac{\lambda_0^+(x^*)}{1 - \epsilon}$.

\medskip

$\lambda(x^*,k)$-s value might only have reduced since the last time $k$-s price was checked as at every instance it's price might increase it is checked. Therefore, we have that at that whenever the algorithm halts $\lambda(x^*,k) \geq \lambda_0(x^*) \cdot (1 + 5 \epsilon) \cdot (1 - \epsilon) \geq \lambda_0(x^*)$ if $\epsilon \leq 4/5$ which concludes the claim.

\end{proof}

\begin{claim}

\label{cl:mixed:approximation}

At all times the algorithm satisfies that $P^*x^* \leq \mathbb{1} \cdot (1 + 200 \cdot \epsilon)$ (given $\epsilon < 1/10$).

\end{claim}

\begin{proof}

We will show that the following invariant is maintained throughout the run of algorithm:

$$\max_{i \in [m_p]}\{P^*_i x^*\} \leq \min_{i \in [m_c]}\{C^*_i x^*\} \cdot (1 + 50\epsilon) + 100 \epsilon$$

Initially the inequality holds and the inequality will be affected by 3 kinds of changes: updates to $P^*$, $C^*$ and $x^*$. Notice, that whenever $C^*$ is updated only the right hand size of the inequality may change and it may only increase. Similarly, a decrease of $P^*$ may only decrease the left hand side. However, the algorithm does increase the last column of $P^*$ whenever an decrease of $P^*$ occurs. Note however, that the last coordinate of $x^*$ is is always $0$ as coordinate $n+1$ is never boosted (due to the last column of $C*$ being $0$ $\lambda(x^*,n+1)$ is unbounded hence coordinate $n+1$ is never cheap).

\medskip

It remains to argue that during Boost calls the inequality is maintained. Assume coordinate $i$ is boosted by $\delta$ and at the start of the boosting process the inequality holds. Define $f_p(x^*) = \eta \cdot \log (\sum_{i \in [m_p]} \exp (\eta P^*_i x^*))$ and $f_c(x^*) = -\eta \cdot \log (\sum_{i \in [m_c]} \exp (-\eta C^*_i x^*))$ analogously to Section~\ref{sec:lowerbound}. Following the proof of Property~\ref{prop:stop} we get that $f_p(x^* + \delta \cdot e_i) - f_p(x^*) \leq (f_c(x^* + \delta \cdot e_i) - f_c(x^*)) \cdot (1 + 50 \cdot \epsilon)$ (note that a different constant $50\epsilon$ appears as the amount of slack the algorithm uses is different). As both function are monotonous in $x^*$ and relaxing updates may only increase $f_c(x^*)$ and reduce $f_p(x^*)$ we have that at all times $f_c(x^*) \cdot (1 + 50 \cdot \epsilon) \geq f_p(x^*)$. 

\medskip

By the properties of $f_c(x^*)$ and $f_p(x^*)$ (stated in Section~\ref{sec:static} as Equation~\ref{eq:eta} we can conclude that $\max_{i \in [m_p]}\{P^*_i x^*\} \leq \min_{i \in [m_c]}\{C^*_i x^*\} \cdot (1 + 50 \cdot \epsilon) + \epsilon \cdot (1 + 50 \cdot \epsilon)) \leq \min_{i \in [m_c]}\{C^*_i x^*\} \cdot (1 + 50 \cdot \epsilon) + 100 \cdot \epsilon$ at all times. 

\medskip

Let $\delta$ stand for the value of the last boost. By definition $C^* \delta \leq \mathbb{1} \cdot \epsilon$ hence at all times $\min_{j \in [m_c]} \{C^*_j x^*\} \leq 1+\epsilon$. Therefore at all times $\max_{i \in [m_p]}\{P^*_i x^*\} \leq (1 + 50 \epsilon) \cdot (1 + \epsilon)  + 100 \epsilon \leq 1 + 200 \epsilon$ that is $P^*x^* \leq \mathbb{1} \cdot (1 + 200 \epsilon)$.

\end{proof}

\begin{corollary}

\label{cor:mixed:weightbounds}

If $\epsilon \leq 1/200$ then at all times $w_p(x^*) \leq \exp(3 \cdot \eta)$ and $w_x(x^*) \geq \exp(2 \cdot \eta)$, hence $\lambda_0(x^*) \leq \exp(5 \cdot \eta)$.

\end{corollary}

\begin{proof}

$w_p(x^*) \leq m_p \cdot \exp(\eta \cdot \max_{i \in [m_p]} \{P^*_i x^*\}) \leq \exp(\eta (2 + 200 \cdot \epsilon)) \leq \exp(\eta \cdot 3)$. As argued in the proof of Claim~\ref{cl:mixed:approximation} we have that $\min_{j \in [m_c]} \{C^*_j x^* \} \leq 1 + \epsilon$ at all times, hence $w_c(x^*) \geq \exp(-\eta \cdot \min_{j \in [m_c]} \{C^*_j x^* \}) \geq \exp(- 2 \cdot \eta)$..

\end{proof}

\begin{lemma}

Whenever the algorithm halts without returning an answer the extended linear program is infeasible. If the algorithm may only return an $O(\epsilon)$ approximate solution to the extended linear program.

\end{lemma}

\begin{proof}

The lemma follows from Claim~\ref{cl:mixed:infeasability} and Claim~\ref{cl:mixed:inv}

\end{proof}

\subsection{Running Time}

\begin{claim}

\label{cl:mixed:whackingbound}

For any $k \in [n]$ \textit{Boost(k)} may be called at most $O(\frac{\log^2(m +U/L)}{\epsilon^2}) = O(\frac{\log(m +U/L) \cdot \eta}{\epsilon})$ number of times.

\end{claim}

\begin{proof}

Fix a coordinate $k \in [n]$ for the rest of the proof. 
  
Consider any call to the {\sc Boost}$(x^*)$ subroutine during the course of our dynamic algorithm, which increases coordinate $k$ of the vector $x^* \in \mathbb{R}_{\geq 0}^n$ by some amount $\delta > 0$. From  Algorithm~\ref{alg:mixed:boost}, observe that just before this specific call to the  subroutine, either there exists a packing constraint $i \in [m_p]$ with $P^*(i, k) \cdot \delta  \leq \epsilon/\eta$, or there exists a {\em not-too-large} covering constraint $j \in [m_c]$\footnote{We say that a covering constraint $j \in [m_c]$ is {\em not-too-large} iff $C^*_j x^* < 2$.} with $C(j, k) \cdot \delta \leq \epsilon/\eta$.  In the former (resp.~latter) case, we refer to the index $i$ (resp.~$j$) as the {\em pivot} and the current value of $P(i, k)$ (resp.~$C(j, k)$) as the {\em pivot-value} corresponding to this specific call to the subroutine {\sc Boost}$(k)$.\footnote{Note that the values of the entries in the constraint matrix change over time due to the sequence of updates. The {\em pivot-value} of a call to {\sc Boost}$(x, k)$ refers to the value of the concerned entry in the constraint matrix {\em just before}  the specific call to  {\sc Boost}$(k)$.}
 
Due to Assumption~\ref{as:mixed:range} we can assume that $U = 2^{\tau} \cdot L$ for some integer $\tau = O(\log n)$. Partition the interval $(L, U]$ into $\tau$  segments $\mathcal{I}_0, \ldots, \mathcal{I}_{\tau-1}$, where $\mathcal{I}_{\ell} = \left(L \cdot 2^{\ell}, L \cdot 2^{\ell+1} \right]$ for all $\ell \in \{0, \ldots, \tau-1\}$. 

\begin{claim}

\label{cl:mixed:bound:pivot:covering}

Consider any $\ell \in \{0, \ldots, \tau-1\}$.  Throughout the duration of our dynamic algorithm, at most $16 \eta/\epsilon$ calls with pivot-values in $\mathcal{I}_{\ell}$ and pivots in $[m_c]$ are made to the {\sc Boost}$(k)$ subroutine.

\end{claim}

\begin{proof}

Let $x_k \in \mathbb{R}_{\geq 0}$ denote the $k^{th}$ coordinate of the vector $x \in \mathbb{R}_{\geq 0}^n$. For ease of exposition, we say that a call to  {\sc Boost}$(k)$ is  {\em covering-critical} iff its pivot is in $[m_c]$ and its pivot-value is in $\mathcal{I}_{\ell}$. We wish to upper bound the total number of covering-critical calls made during the course of our dynamic algorithm. 

Consider any specific covering-critical call to {\sc Boost}$(k)$  with pivot $j \in [m_c]$ and pivot-value  $\alpha \in \mathcal{I}_{\ell} = \left( L \cdot 2^{\ell-1}, L \cdot 2^{\ell}\right]$. From the definition of Algorithm~\ref{alg:mixed:boost}, it follows that this call increases $x_k$ by $\delta = \epsilon/(\eta \alpha \cdot 4) \geq \epsilon/(\eta L 2^{\ell+2})$.

Suppose that $T$ covering-critical calls have been made to the {\sc Boost}$(x, k)$ subroutine, where $T > 16 \cdot \eta/\epsilon$. Let us denote these calls in increasing order of time by: $\Gamma_1, \Gamma_2, \ldots, \Gamma_T$. Thus, for all $t \in [T]$, we let $\Gamma_{t}$ denote the $t^{th}$ covering-critical call made during the course of our dynamic algorithm.  From the discussion in the preceding paragraph,  each of these critical calls $\Gamma_t$ increases $x_k$ by at least $\epsilon/(\eta L 2^{\ell + 2})$. 

Since $x_k = 0$ initially and it increases monotonically over time, just before the last critical call $\Gamma_{T}$ we have $x_k \geq (T-1) \cdot \epsilon/(\eta L 2^{\ell + 2}) \geq (16 \eta/\epsilon) \cdot \epsilon/(\eta L 2^{\ell + 2}) = 1/(L 2^{\ell-2})$. Let $j_T \in [m_c]$ be the pivot of the  call $\Gamma_T$. Since the pivot-value of $\Gamma_T$ lies in $\mathcal{I}_{\ell}$, it follows that $C(j_T, k) > L \cdot 2^{\ell-1}$ just before the  call $\Gamma_T$. Thus, just before the  call $\Gamma_T$, we have $C_{j_T} \cdot x \geq C(j_T, k) \cdot x_k > L  2^{\ell-1} \cdot 1/(L 2^{\ell-2}) = 2$. But if $C_{j_T} \cdot x > 2$ just before the call $\Gamma_T$, then $j_T$ cannot be the pivot of $\Gamma_T$. This leads to a contradiction. Hence, we must have $T \leq 16 \cdot \eta/\epsilon$, and this concludes the proof of the claim.

\end{proof}

\begin{claim}

\label{cl:mixed:bound:pivot:packing}
Consider any $\ell \in \{0, \ldots, \tau-1\}$.  Throughout the duration of our dynamic algorithm, at most $16 \cdot \eta/\epsilon$ calls with pivot-values in $\mathcal{I}_{\ell}$ and pivots $i \in [m_p]$ are made to the {\sc Boost}$(k)$ subroutine.

\end{claim}

\begin{proof}(Sketch) As in the proof of Claim~\ref{cl:mixed:bound:pivot:covering}, let $x_k \in \mathbb{R}_{\geq 0}$ denote the $k^{th}$ coordinate of the vector $x \in \mathbb{R}_{\geq 0}^n$. Say that a call to  {\sc Boost}$(k)$ is  {\em packing-critical} iff its pivot is in $[m_p]$ and its pivot-value is in $\mathcal{I}_{\ell}$.

Suppose that   $T$ covering-critical calls have been made to the {\sc Boost}$(k)$ subroutine, where $T > 16 \cdot \eta/\epsilon$. Let us denote these calls in increasing order of time by: $\Gamma_1, \Gamma_2, \ldots, \Gamma_T$. Thus, for all $t \in [T]$, we let $\Gamma_{t}$ denote the $t^{th}$ packing-critical call made during the course of our dynamic algorithm. Let $i_T \in [m_p]$ be the pivot of the last packing-critical call $\Gamma_T$. Following the same argument as in the proof of Claim~\ref{cl:mixed:bound:pivot:covering}, we conclude that $P_{i_T} \cdot x > 2$ just before the call $\Gamma_T$. By Claim~\ref{cl:mixed:approximation} this leads to a contradiction (assuming $\epsilon < 1/200$). Thus, it must be the case that $T \leq 16\eta/\epsilon$, and this concludes the proof of the claim.

\end{proof}

Any call to the subroutine {\sc Boost}$(k)$ has pivot in $[m_p] \cup [m_c]$ and pivot-value in $\mathcal{I}_{\ell}$ for some $\ell \in \{0, \ldots, \tau-1\}$.\footnote{This holds because the pivot-values lie in the range $(L, U]$, and this range has been partitioned into subintervals: $\mathcal{I}_0, \ldots, \mathcal{I}_{\tau-1}$. } Accordingly, from Claim~\ref{cl:bound:pivot:covering} and Claim~\ref{cl:bound:pivot:packing}, it follows that the total number of calls made to the {\sc Boost}$(k)$ during the course of  our dynamic algorithm is at most $2 \cdot \tau \cdot (16 \eta/\epsilon) = O(\frac{\log^2(m + U/L)}{\epsilon^2})$. 

\end{proof}

\begin{claim}

\label{cl:mixed:weightupdate}

All calls to UpdatePWeights and UpdateCWeights be handled in $O(\frac{\log^2(m + U/L) \cdot N}{\epsilon^2})$ total time.

\end{claim}

\begin{proof}

As each call to \textit{Boost} calls the weight update functions once by Claim~\ref{cl:mixed:whackingbound} they will be called a total of $ O(\frac{\log(m +U/L) \cdot \eta}{\epsilon} \cdot n)$ times. Fix $k \in [n]$. \textit{UpdatePWeights($k$)} and \textit{UpdateCWeights($k$)} will be called $O(\frac{\log(m +U/L) \cdot \eta}{\epsilon})$ times. Hence, the outer loops of both function will take $O(N \cdot \frac{\log(m +U/L) \cdot \eta}{\epsilon})$ total time to handle. 

\medskip

Without loss of generality focus on \textit{UpdatePWeights}. For any $i \in [m_p]$ between any two updates to $\hat{w}_p(x^*,i)$ $w_p(x^*,i)$ must increase by a factor of $1 + \Omega(\epsilon)$. Due to Corollary~\ref{cor:mixed:weightbounds} this may happen at most $O(\log^2(m + U/L)/\epsilon^2)$ times. After updating $\hat{w}_p(x^*,i)$ the algorithm will spend $O(R^P_i)$ time updating $\hat{\lambda}_0(x^*)(x^*,k)$ values. Therefore, over all $i \in [m_p]$ the total work to maintain $\hat{\lambda}_0(x^*)(x^*,k)$ values will be in $O(\frac{\log^2(m + U/L)}{\epsilon^2})$. The proof proceeds similarly for the work \textit{UpdateCWeights} does on maintaining $\hat{\lambda}_0(x^*)(x^*,k)$ values.

\end{proof}

\begin{claim}

\label{cl:mixed:boost:totaltime}

The total time spent by the algorithm through calls to the \textit{Boost} subroutine is in $O(\frac{\log^2(m + U/L) \cdot N}{\epsilon^2})$

\end{claim}

\begin{proof}

Each call to \textit{Boost($k$)} requires 3 non-constant time tasks. Firstly an estimate of $\delta_k$ has to be sampled. By Claim~\ref{cl:mixed:oracle:runningtime} and Claim~\ref{cl:mixed:whackingbound} these samplings will take $O(N \cdot \frac{\log(m + U/L)}{\epsilon^2})$ total time. 

Afterwards the algorithm updates all $w_p(x^*,i)$ for $i \in [m_p] | P^*(i,k) > 0$ and $w_c(x^*,j)$ for $j \in [m_c] | C^*(j,k) > 0$ in $O(N^P_k + N^C_k)$ time. Over all calls to boost this will take total $O(N \cdot \frac{\log^2(m + U/L)}{\epsilon^2})$ time by Claim~\ref{cl:mixed:oracle:runningtime}.

Finally, the algorithm makes calls to \textit{UpdatePWeights($k$)} and \textit{UpdateCWeights($k$)} (see Algorithms~\ref{alg:mixed:updateCweights} and \ref{alg:mixed:updatePweights}). This will take total time $O(N \cdot \frac{\log^2(m + U/L) \cdot N}{\epsilon^2})$ by Claim~\ref{cl:mixed:weightupdate}.

\end{proof}

\begin{claim}

\label{cl:mixed:update:totaltime}

The total time by the \textit{UpdateP} and \textit{UpdateC} subroutines (outside of calls to \textit{Iterate} and \textit{Boost}) is bounded by $O(\frac{\log^2(m + U/L) \cdot N}{\epsilon^2} + t)$.

\end{claim}

\begin{proof}

The only action \textit{UpdateP} does which is not in $O(1)$ is the while loop. Due to Claim~\ref{cl:mixed:whackingbound} the while loops will test a coordinates price at most $O(n \cdot \frac{\log^2(m + U/L}{\epsilon^2} + t)$ time over the total run of the algorithm.

Similarly, the while loops of \textit{UpdateC} may also only have $O(n \cdot \frac{\log^2(m + U/L}{\epsilon^2} + t)$ iterations due to Claim~\ref{cl:mixed:whackingbound}. The only other non-constant time operation of \textit{UpdateC} is the process of updating $\hat{w}_c(x^*,j)$ and affected $\hat{\lambda}(x^*,k')$ values. Whenever $\hat{w}_c(x^*,j)$ is updated $R^C_j$ values of $\hat{\lambda}(x^*,k')$ will be updated. As an update to $\hat{w}_c(x^*,j)$ increases its value by a multiplicative factor of at least $1 + \epsilon$ due to Assumption~\ref{as:mixed:range} there can be at most $O(\log(N)/\epsilon)$ updates of $\hat{w}_c(x^*,j)$. Therefore, the maintenance of the slack parameters will take at most $O(\frac{N \cdot \log(N)}{\epsilon} + t)$ total time for \textit{UpdateC}. This finishes the proof. 

\end{proof}

\begin{lemma}

\label{lm:mixed:totalrunningtime}

Algorithm~\ref{alg:mixed:main} runs in $O(\frac{N \cdot \log^2(m + U/L)}{\epsilon^2} + t)$ deterministic time.

\end{lemma}

\begin{proof}

The time spent by calls to \textit{Boost} or updates is handled in $O(\frac{N \cdot \log^2(m + U/L)}{\epsilon^2} + t)$ total time by Claim~\ref{cl:mixed:boost:totaltime} and Claim~\ref{cl:mixed:update:totaltime}. Between each call to iterate $\lambda_0(x^*)$ increases by a factor of $1/(1-\epsilon)$. As $\lambda_0(x^*) \leq \exp(5 \cdot \eta)$ by Corollary~\ref{cor:mixed:weightbounds} there will be at most $O(\eta/\epsilon)=O(\log(m + U/L)/\epsilon^2)$ calls to \textit{Iterate()} which will take $O(\frac{N \cdot \log^2(m + U/L)}{\epsilon^2})$ to complete (apart from \textit{Boost} calls). This completes the lemma.

\end{proof}

\subsection{Maintenance of an estimate of $\delta_k$ for Algorithm~\ref{alg:mixed:boost}}

\label{sec:mixed:findincrement}

To imitate the increments of the static algorithm described by Figure~\ref{fig:boost} when \textit{Boost($k$)} is called we would like to find $\delta = \max_\delta\{\max\{\max_i P^*(i,k) \cdot \delta, \max_{j|C^*_j x^* < 2} C^*(j,k) \cdot \delta \} = \epsilon/\eta$. Unfortunately, under dynamic relaxing updates returning the exact $\delta_k$ is too slow. Therefore, we will implement an efficient oracle which returns $\delta \in [\delta_k \cdot (1 - \epsilon), \delta_k]$ when queried in $O(1)$ time. 

\paragraph{Max-Heap:} Using folklore data structures we will assume that a max-heap on at most $m = m_c + m_p$ objects is a black-box tool which can handle the following actions:

\begin{itemize}
    \item In $O(1)$ the max-heap returns the maximum-element stored in the heap
    \item In $O(\log(m))$ time an element can be removed from the heap
    \item In $O(\log(m))$ time an element can be added to the heap
\end{itemize}

\paragraph{Initialization:} At initialization the oracle implements $n$ max heaps $H_k : k \in [n]$, one for each coordinate of $x^*$. We will now describe the implementation of a specific heap $H_k$. $H_k$ will consist of at most $N^C_k + N^P_k$ objects each representing a non-zero element of the $k$-th column of $C^*$ or $P^*$. Each non-zero element of these two columns will be inserted into the heap with it's initial value and an identifier. Hence, the top element of $H_k$ will represent the largest value in the $k$-th column of $C^*$ and $P^*$.

\paragraph{Update to $P^*$:} Assume that due to a relaxing update $P^*(i,k)$ reduced by $\Delta$ to some $P_+^*(i,k)$. The oracle will look-up the value $P^*(i,k)$ currently represented under in $H_k$, say this value is $P_-^*(i,k)$. If $P_-^*(i,k) > 2 \cdot P_+^*(i,k)$ then the oracle will remove $P_-^*(i,k)$ from $H_k$ and re-enter it as $P_+^*(i,k)$. If $P_-^*(i,k) \leq 2 \cdot P_+^*(i,k)$ then the oracle doesn't change anything in the heap $H_k$.

\paragraph{Update to $C^*$:} Assume that due to a relaxing update $C^*(j,k)$ increased by $\Delta$ to some $C_+^*(j,k)$. If previously to the update $(C^*)_i x^* < 2$ however after the update $(C^+_*)_i x^* \geq 2$ the oracle proceeds over all non-zero entries of the $j$-th row of $C^*$ (iterate over all $k' \in [m_c] | C^*(j,k')>0$) of the form $C^*(j,k') > 0$ and removes them from their corresponding heaps. Otherwise the oracle proceeds similarly to its behaviour under a packing update: say that the value of $C^*(j,k)$ stored in $H_k$ is $C^*_-(j,k)$. If $C^*_-(j,k) < 2 \cdot C^*_+(j,k)$ then the oracle removes the $C^*_-(j,k)$ node from the heap and re-enters it as the current value $C^*_+(j,k)$. If $C^*_-(j,k) \geq 2 \cdot C^*_+(j,k)$ the oracle rests.

\paragraph{Query of $\delta_k$}: Let the top element of heap $H_k$ when the query is made be $\kappa$. The oracle will return value $\delta = \frac{\eta}{\epsilon \cdot \kappa/2}$. Afterwards, the oracle proceeds to adjust it's max-heaps to the boost of coordinate $k$. Let $x^*_+$ and $x^*_-$ represent the state of $x^*$ after and before having been incremented by $\delta \cdot e_k$ through \textit{Boost($k$)} respectively. The oracle iterates through all $j \in [m_c]$ such that $C^*(j,k) > 0$. For any such $j$ if $(C^*_-)_j x^* < 2$ however $(C^*_+)_j x^* > 2$ the oracle proceeds as follows: for all $k' \in [n] | C^*(j,k')>0$ the oracle removes $C^*(j,k')$ (or the value stored as the representative of $C^*(j,k')$ which might be somewhat smaller) from heap $H_{k'}$.

\begin{claim}

\label{cl:mixed:oracle:approximation}

If the oracle is queried to return an estimate of $\delta_k$ it returns a value in $[\delta_k/4, \delta_k]$

\end{claim}

\begin{proof}

Assume \textit{Boost($k$)} was called and the oracle is queried to return an estimate of $\delta_k$. Observe, that at all time points the heap $H_k$ contains an object corresponding to each non-zero elements of the $k$-th column of $P^*$. Furthermore, $H_k$ contains an element for each $C^*(j,k) : j \in [m_c] | C^*_j < 2$. Fix an element of $P^*$ represented in $H_k$, say $P^*(i,k)$ and let $P^*_-(i,k)$ stand for the value it is represented under in the heap $H_k$. Observe that by construction $P^*_-(i,k) \leq P^*(i,k) \leq P^*_-(i,k)$. Fix an element of $C^*$ represented in $H_k$, say $C^*(j,k)$ and let it's representation in $H_k$ have value $C^*_-(j,k)$. Similarly, observe by construction that $C^*_-(j,k) \cdot 2 \geq C^*(j,k) \geq C^*_-(j,k)$. 

Define $\kappa = \max \{\max_{i \in [m_p]} P^*(i,k), \max_{j \in [m_c] | C^*_j x^* < 2} C^*(j,k)\}$. By the observations above we can conclude that the maximum element of $H_k$ is in $[\kappa/2, \kappa \cdot 2]$. Therefore, the oracle returns a value in $[\frac{\epsilon}{4 \cdot \kappa \cdot \eta}, \frac{\epsilon}{\kappa \cdot \eta}] = [\delta_k/4, \delta_k]$.

\end{proof}

\begin{claim}

\label{cl:mixed:oracle:runningtime}

Let $B_k : k \in [n]$ represent the number of times Algorithm~\ref{alg:mixed:main} calls subroutine \textit{Boost($k$)}. The total time it takes to maintain the oracle and return it's queries is in $O(\sum_{k \in [n]} B_k N^C_k + N \cdot \log^2(N) + t)$.

\end{claim}

\begin{proof}

Say that the oracle re-adjust a constraint parameter $P^*(j,k)$ or $C^*(i,k)$ whenever it removes their outdated value for heap $H_k$ and inserts their current value in the constraint matrix. Whenever a replacement of a constraint parameter occurs it's value in $H_k$ has to double if it's a member of $C^*$ or halve if its a member of $P^*$ respectively. By Assumption~\ref{as:mixed:range} non-zero parameters of constraint matrices always lie in $[1/poly(N), poly(N)]$ therefore each constraint element may has its value re-adjusted at most $O(\log(N))$ times. This means that handling re-adjustments of all elements of the heaps maintained by the oracle can be completed in $O(N \cdot \log^2(N))$ total time.

When an update is made to $P^*$ the oracle checks the value of a single element in the heap, therefore it's work outside of possibly re-adjusting said value can be completed in $O(1)$ time. This means $O(N \cdot \log^2(N) + t)$ upper bounds the work of the oracle handling packing constraint updates.

When an update is made to $C^*$, say to $C^*(j,k)$, the oracle first has to check weather $C^*_j x^*$ has exceeded $2$ for the first time. Checking the value of $C^*_j x^*$ can be completed in $O(1)$ time (assuming the algorithm keeps track of each constraints progress at all times). The value of $C^*(j,k)$ may only exceed $2$ once over the total run of the iteration. At this point the oracle removes $R^C_j$ elements from it's heaps. All of these removals therefore will take at most $O(N \cdot \log(N))$ total time. Therefore, $O(N \cdot \log^2(N) + t)$ upper bounds the work of the oracle handling covering constraint updates.

When a coordinate $k$ is boosted by Algorithm~\ref{alg:mixed:boost} the oracle has to iterate over all non-zero elements of the $k$-th column of $C^*$ which takes $O(N^C_k)$ time. For all $j \in [m_c] | C^*(j,k) > 0$ the oracle calculates $C^*_j x^*$ in $O(1)$. Afterwards, the oracle may proceed to remove all elements in any such line of $C^*$, however each element may only be removed at most once. Therefore, all work of the oracle due to boosting queries can be upper bounded by $O(\sum_{k \in [n]} B_k N^C_k + N \cdot \log(N) + t)$. This finishes the proof.

\end{proof}

\subsection{The Handling of Translation Updates}

\label{subsec:app:mixed:translation}

Using Theorem~\ref{thm:mixed} as a black-box statement we will show how the algorithm can handle translation updates in addition to element-wise updates without incurring any blow-up in the total running time. To re-iterate a mixed packing covering linear program defined by constraint matrices $(P,C)$ is the problem of finding $x$ such that $Px \leq \mathbb{1}, Cx \geq \mathbb{1}$. A relaxing translation update of the packing constraint could formalized as follows: for some coordinate $i : i \in [m_p]$ and $\gamma>0$ we update the linear program such that we are looking for $x$ satisfying $Px \leq \mathbb{1} + \gamma \cdot e_i, Cx \geq \mathbb{1}$. Similarly, a relaxing translation update to the covering matrix is an updated defined by some $j \in [m_c]$ and $\gamma > 0$ such that we update the linear program to $Px \leq \mathbb{1}, Cx \geq \mathbb{1} - \gamma \cdot e_j$.

\begin{observation}

\label{obs:app:mixed:translation:slack}

Let $\Delta_P \in \Re_{\geq 0}^{m_p}$ and $\Delta_C \in \Re_{\geq 0}^{m_c}$ where $0 \leq \Delta_P \leq \mathbb{1} \cdot \epsilon$ and $0 \leq \Delta_C \leq \mathbb{1} \cdot \epsilon$. If $x$ is an $\epsilon$-approximate solution to the linear program defined by constraint matrices $(P,C)$ then $x$ is a $2 \cdot \epsilon$ approximate solution to the linear program defined by $(P + \Delta_P, C - \Delta_C)$.

\end{observation}

As demonstrated by Observation~\ref{obs:app:mixed:translation:slack} if the algorithm ignores translation updates until any constraint is translated by more than an $\epsilon$ factor the solution returned is still $O(\epsilon)$ approximate to the linear program. Once a constraint (say the $i$-th pacing constraint) is translated by at least an $\epsilon$ factor the algorithm can simulate this relaxation through $R^P_i$ entry updates of the packing matrix by relaxing all elements in the $i$-th row of $P$ by the same multiplicative factor. Assuming similarly as in Assumption~\ref{as:mixed:range} that all input parameters lie in the $[1/poly(N), poly(N)]$ range we can be certain that each constraint will need to be updated at most $O(\log_{1 + \epsilon}(poly(N)) = O(\frac{\log(N)}{\epsilon})$ times due to being translated by more than an $\epsilon$ factor. Hence, the total number of entry-wise updates required to handle all translation updates will be bounded by $O(N \cdot \frac{\log(N)}{\epsilon})$. 

Therefore, the algorithm can handle translation updates with $O(N \cdot \frac{\log(N)}{\epsilon} + t)$ additional time (where $t$ here stands for the number of translation updates) as the algorithm of Theorem~\ref{thm:mixed} runs in linear time with respect to the number of entry wise updates. This implies that Theorem~\ref{thm:mixed} can be extended to work for translation updates without incurring any loss in running time.